\documentclass[runningheads,envcountsame]{llncs}
\pdfoutput=1

\usepackage[utf8]{inputenc}
\DeclareUnicodeCharacter{03BB}{\ensuremath{\lambda}}
\DeclareUnicodeCharacter{03BC}{\ensuremath{\mu}}

\usepackage[T1]{fontenc}
\usepackage[expansion=true,protrusion=true]{microtype}

\usepackage[english]{babel}
\usepackage[autostyle]{csquotes}

\usepackage[backend=biber,doi=false,isbn=false,url=false,style=lncs]{biblatex}
\usepackage[shortlabels,inline]{enumitem}
\usepackage{fnpct}
\usepackage{graphicx}
\usepackage{subcaption}
\usepackage{todonotes}
\usepackage{xargs}
\usepackage{xcolor}
\usepackage{xifthen}

\usepackage{cmll}
\usepackage{amsfonts}
\usepackage{amsmath}
\usepackage{amssymb}
\usepackage{mathtools}
\usepackage{mathrsfs}
\usepackage{stmaryrd}

\usepackage{amsthm}

\usepackage{tikz}
\usetikzlibrary{arrows.meta,graphs,positioning,quotes,tikzmark}
\usepackage[SEQ]{prftree}

\usepackage[pdfusetitle]{hyperref}
\usepackage[nameinlink]{cleveref}

\hypersetup{
    colorlinks=true,
    linkcolor=blue,
    citecolor=blue,
    urlcolor=blue,
}

\DeclareMathOperator{\pwr}{\mathscr{P}}

\newcommand{\N}{\mathbb{N}}

\newcommand{\cceq}{\mathrel{\mathop{::=}}}

\let\th\thesis

\newcommand{\rst}{{\upharpoonright}}

\let\co\overline
\let\dl\overline
\let\sdl\shorteroverline

\newcommand{\data}[1]{\mathtt{#1}}

\newcommand{\ifemp}[3]{\ifthenelse{\isempty{#1}}{#2}{#3}}

\newcommandx{\rl}[2][1]{%
    \ifemp{#2}%
        {\ensuremath{\mathtt{#1}}}%
        {\ensuremath{\mathtt{#2}_{#1}}}%
}

\newcommandx{\rwr}[2][1,2]{\mathrel{%
    \ifemp{#1}%
        {\ifemp{#2}%
            {\longrightarrow}%
            {\stackrel{#2}{\longrightarrow}}%
        }%
        {\ifemp{#2}%
            {\longrightarrow_{\data{#1}}}%
            {{\stackrel{\raisebox{-2pt}{\ensuremath{\scriptstyle {#2}}}}{\longrightarrow}}_{\data{#1}}}%
        }%
}}

\newcommand{\nm}[2]{{#2}^{#1}}
\newcommand{\axg}[1]{\llbracket {#1} \rrbracket}

\newcommand{\blwk}[1]{\data{Wk}^{\data{bl}}_{#1}}
\newcommand{\blid}[1]{\data{Id}^{\data{bl}}_{#1}}

\newcommand{\blaxg}[1]{\llparenthesis {#1} \rrparenthesis}

\newcommand{\adin}{\mathrel{\Yleft}}

\newcommand{\hvp}{\vphantom{\frac{0}{0}}}

\makeatletter
\let\blg@defaultnm\nm
\newcommand{\blg@markname}[1]{--blg-mark-#1}
\newcommand{\blg@branchnm}[2]{\blg@defaultnm{#1}{\tikzmarknode{\blg@markname{#1}}{#2}}}
\newsavebox{\blg@branchbox}
\newcommand{\blgbranch}[2]{%
    \begingroup%
    \let\nm\blg@branchnm%
    \savebox{\blg@branchbox}{\(#1\)}%
    \prfassumption{%
        \begin{tikzpicture}[remember picture, thick, black]
            \node[outer sep=0, inner xsep=0, inner ysep=-.2em] (branch) at (0,0) {\usebox{\blg@branchbox}};
            \foreach \x / \y in {#2}
                \draw ([yshift=.2em]branch.north -| {\blg@markname{\x}}.north) .. controls ++(0,1em) and ++(0,1em) .. ([yshift=.2em]branch.north -| {\blg@markname{\y}}.north);
        \end{tikzpicture}%
    }%
    \endgroup%
}
\makeatother

\begin{document}

\title{On the semantics of proofs in classical sequent calculus.}
\titlerunning{On the semantics of proofs in classical sequent calculus}

\author{Fabio Massaioli}
\authorrunning{F. Massaioli}

\institute{Scuola Normale Superiore, Pisa, Italy
\email{fabio.massaioli@sns.it}}
\maketitle              % typeset the header of the contribution
\begin{abstract}
    We discuss the problem of finding non-trivial invariants of non-deterministic, symmetric
    cut-reduction procedures in the classical sequent calculus. We come to the conclusion
    that (an enriched version of) the propositional fragment of GS4 -- i.e.\ the one-sided
    variant of Kleene's context-sharing style sequent system G4, where independent rule
    applications permute freely -- is an ideal framework in which to attack the problem.
    We show that the graph induced by axiom rules linking dual atom occurrences is preserved
    under arbitrary rule permutations in the cut-free fragment of GS4. We then refine the
    notion of axiom-induced graph so as to extend the result to derivations with cuts, and
    we exploit the invertibility of logical rules to define a global normalisation procedure
    that preserves the refined axiom-induced graphs, thus yielding a non-trivial invariant
    of cut-elimination in GS4. Finally, we build upon the result to devise a new proof system
    for classical propositional logic, where the rule permutations of GS4 reduce to identities.

    \keywords{Classical propositional logic \and Sequent calculus \and Cut-e\-lim\-i\-na\-tion
        \and Invariants of cut-reduction \and Proof-identity \and Denotational semantics.}
\end{abstract}

\section{Introduction}\label{sec:intro}

Cut-elimination procedures in classical sequent calculus are notoriously non-deterministic
and non-confluent, both in the original formulation by Gentzen and in later reformulations~%
\cite{Gen35,DJS95,BB96,BBS97,AT14,Pul22}. It is natural to ask whether those instances of
non-confluence are superficial in nature, i.e.\ whether distinct normal forms of the same
derivation are in fact correlated in a non-trivial way. A famous counter-example by Lafont~%
\cite{GLT89} purports to show that the answer is negative, that is, any notion of proof
equivalence compatible with classical cut-elimination must be a trivial one that identifies
all proofs of the same sequent. Specifically, the counter-example involves a derivation of
the form
\[
    \prftree[r]{\rl{ctr}}
        {\prftree[r]{\rl{cut}}
            {\prftree[r]{\rl{wk}}
                {\prfsummary[\(P\)]{\th A}}
                {\th A, B}}
            {\prftree[r]{\rl{wk}}
                {\prfsummary[\(Q\)]{\th A}}
                {\th A, \dl{B}}}
            {\th A, A}}
        {\th A}
\]
where \(P,Q\) are a pair of arbitrary derivations of the same formula \(A\)\footnote{We
use the one-sided formulation of sequent calculus in the style of Tait \cite{Tai68} to
reduce the number of rules to be treated. Negation is defined as an involution on formulas
through De Morgan dualities. Sequents are considered as multisets, i.e.\ quotiented up to
arbitrary permutations of their elements, hence the exchange rule is implicit}. Because
both cut-formulas are introduced by weakening, there are two ways to eliminate the cut,
leading to two potentially very different derivations:
\[
    \prftree[r]{\rl{ctr}}
        {\prftree[r]{\rl{wk}}
            {\prfsummary[\(P\)]{\th A}}
            {\th A, A}}
        {\th A}
    \qquad\qquad
    \prftree[r]{\rl{ctr}}
        {\prftree[r]{\rl{wk}}
            {\prfsummary[\(Q\)]{\th A}}
            {\th A, A}}
        {\th A}
\]
Any notion of proof identity that is compatible with the cut-elimination process should then
identify all three derivations shown above, leading to the identification of~\(P\) with \(Q\)
under the reasonable assumption that the weakening-contraction sequence on~\(A\) be an
irrelevant detour. \(P\) and~\(Q\), however, were \emph{arbitrary} derivations of~\(A\),
i.e.\ potentially not correlated in any way: this has catastrophic consequences, making
our hypothetical notion of proof identity trivial and patently unreasonable. For example,
the following three proofs would be identified:
\[
    \prftree[r]{\rl{\lor}}
        {\prfsummary{\th A}}
        {\th A \lor B}
    \qquad\qquad
    \prftree[r]{\rl{\lor}}
        {\prfsummary{\th B}}
        {\th A \lor B}
    \qquad\qquad
    \prftree[r]{\rl{\lor}}
        {\prfsummary{\th A, B}}
        {\th A \lor B}
\]
meaning that it does not matter at all whether a disjunction holds in virtue of the left
or of the right disjunct, or by contradiction. In a sense we are identifying not just
all proofs of the same thesis, but indeed \emph{all} proofs altogether.

Under such a notion of proof identity, it becomes utterly useless to look for alternative
proofs of the same theorem, as they are virtually the same as any other one and there is
nothing to be learned from them: a conclusion that plainly contradicts millennia of mathematical
practice.

Many solutions have been proposed. We mention among them Krivine's \emph{classical realizability}
program~\cite{Kri09,Kri14}, which has yielded a fruitful analysis of classical reasoning
principles in computational terms; the approaches based on polarised deductive systems,
like Girard's calculus LC~\cite{Gir91}, Parigot's \(\lambda\mu\)-calculus~\cite{Par92lmu},
Danos Joinet and Schellinx' calculi LKT and LKQ~\cite{DJS95}, and finally the approaches
based on fine grained focalization and embeddings of the classical sequent calculus into
linear logic~\cite{LM09,DJS97}.

From the point of view of sequent calculus, the common thread to all mentioned solutions
is to restrict cut-elimination (and sometimes the logical rules too) in a principled way,
breaking its simmetry and therefore solving the non-deterministic cases. Typically the
approaches based on \(\lambda\)-calculus correspond to various double-negation translations
from classical to intuitionistic logic that preserve provability. While this approach
bears many interesting results, it is not entirely satisfactory.

On the one hand, every restriction is somewhat arbitrary. Looking back at Lafont's
example above, it is hard to see any actual reason to privilege one reduct over
the other, and doing so arguably fails to capture the content of the original derivation,
which effectively offers two possible ways to prove the same conclusion. On the other hand,
the problem is by no means exclusive to sequent calculus. Analogous examples can be constructed
in a great variety of classical proof systems, with the notable exception of those who have
been accurately tuned to avoid them, and there is even a categorical counterpart due to Joyal.
The solutions recalled above sacrifice the inherent symmetries of classical proof systems
to the possibility of extracting at least some computational content from them. It is
important to underline here that while those symmetries are not strictly needed to characterize
classical provability, the fact that classical logic allows them to arise seems to be somehow
essential to it.

A long standing open question has been then whether it could possible to work around
the non-deterministic reduction steps by natural and non-trivial adjustments of the calculus
and/or of cut-reduction steps, without resorting to symmetry-breaking techniques like
polarization or embeddings into intuitionistic or linear logic. It is clear that Lafont's
example requires a change in the calculus. Two simple solutions have been discussed many
times in the literature: the \emph{mix} rule and \emph{non-deterministic sums}.
\[
    \makebox[\textwidth]{\small\(
        \vcenter{\prftree[r]{\rl{mix}}
            {\prfsummary[\(P\)]{\th \Gamma}}
            {\prfsummary[\(Q\)]{\th \Delta}}
            {\th \Gamma, \Delta}}
        \quad\stackrel{?}{\longleftarrow}\quad
        \vcenter{\prftree[r]{\rl{cut}}
            {\prftree[r]{\rl{wk}}
                {\prfsummary[\(P\)]{\th \Gamma}}
                {\th \Gamma, A}}
            {\prftree[r]{\rl{wk}}
                {\prfsummary[\(Q\)]{\th \Delta}}
                {\th \Delta, \dl{A}}}
            {\th \Gamma, \Delta}}
        \quad\stackrel{?}{\longrightarrow}\quad
        \vcenter{\prftree[r]{\rl{\oplus}}
            {\prftree[r,d]{\rl{wk}}
                {\prfsummary[\(P\)]{\th \Gamma}}
                {\th \Gamma, \Delta}}
            {\prftree[r,d]{\rl{wk}}
                {\prfsummary[\(Q\)]{\th \Delta}}
                {\th \Gamma, \Delta}}
            {\th \Gamma, \Delta}}
    \)}
\]
While those two rules do not contribute to classical provability, they increase the amount
of available proofs~\cite{LS04}, accommodating for cut-free proofs that intuitively “provide
multiple alternatives.”

The non-deterministic sum (on the right) expresses the intuitive idea that the resulting
proof might be seen either as~\(P\) or~\(Q\), without the reader being free to choose between
the two alternatives. Its use can be extended coherently to all problematic situations,
viewing proofs with cuts as the sum of their possible normal forms. While this approach
is known to be non-trivial and non-reducible to intuitionistic or linear embeddings (see
e.g.~\cite{BBS97,Lai01}), it is also not very deep, as it does not get to the point of what
the content of a classical proof actually is; notably, it is not known to yield a canonical
representation of classical proofs up to the related notion of proof identity. Moreover,
it has no faithful representation on paper, where the choice between \(P\) and~\(Q\)
is clearly available to the reader.

On the other hand, the mix rule (on the left) expresses the idea that two proofs are
effectively provided in parallel, both equally available to the reader. Informally, this
might be viewed as writing down two different proofs of the same theorem one after the
other: something that makes sense and might actually happen in mathematics textbooks.
Extending the calculus with the mix rule effectively solves the weakening-weakening
problem\footnote{Note however that there is a subtle technical issue to be handled:
because logical rules might hide applications of weakenings on their subformulas, weakening
cuts cannot be eliminated directly; instead, all weakenings must first be reduced to atomic
form, something which is well-known to be possible in classical logic.}. The question
becomes then how to handle the other two non-deterministic reduction steps (\cref{fig:non-det-reductions}).

\begin{figure}[t]
    \centering
    \scalebox{0.8}{\makebox[\textwidth]{\(
        \vcenter{\prftree[r]{\rl{cut}}
            {\prftree[r]{\rl{cut}}
                {\prfsummary[\(P\)]{\th \Gamma, A, B \vphantom{\dl{A}}}}
                {\prfsummary[\(Q\)]{\th \Delta, \dl{A}}}
                {\th \Gamma, \Delta, B \vphantom{\dl{A}}}}
            {\prfsummary[\(R\)]{\th \Sigma, \dl{B}}}
            {\th \Gamma, \Delta, \Sigma}}
        \quad\longleftarrow\quad
        \vcenter{\prftree[r]{\rl{cut}}
            {\prftree[r]{\rl{\lor}}
                {\prfsummary[\(P\)]{\th \Gamma, A, B \vphantom{\dl{AB}}}}
                {\th \Gamma, A \lor B \vphantom{\dl{AB}}}}
            {\prftree[r]{\rl{\land}}
                {\prfsummary[\(Q\)]{\th \Delta, \dl{A}}}
                {\prfsummary[\(R\)]{\th \Sigma, \dl{B}}}
                {\th \Delta, \Sigma, \dl{A} \land \dl{B}}}
            {\th \Gamma, \Delta, \Sigma}}
        \quad\longrightarrow\quad
        \vcenter{\prftree[r]{\rl{cut}}
            {\prftree[r]{\rl{cut}}
                {\prfsummary[\(P\)]{\th \Gamma, A, B \vphantom{\dl{B}}}}
                {\prfsummary[\(R\)]{\th \Sigma, \dl{B}}}
                {\th \Gamma, A, \Sigma \vphantom{\dl{A}}}}
            {\prfsummary[\(Q\)]{\th \Delta, \dl{A}}}
            {\th \Gamma, \Delta, \Sigma}}
    \)}}\\[1.5em]
    \scalebox{0.8}{\makebox[\textwidth]{\(
        \vcenter{\prftree[r,d]{\rl{wk}}
            {\prfsummary[\(P\)]{\th \Gamma}}
            {\th \Gamma, \Delta}}
        \quad\longleftarrow\quad
        \vcenter{\prftree[r]{\rl{cut}}
            {\prftree[r]{\rl{wk}}
                {\prfsummary[\(P\)]{\th \Gamma}}
                {\th \Gamma, A \vphantom{\dl{A}}}}
            {\prftree[r]{\rl{wk}}
                {\prfsummary[\(Q\)]{\th \Delta}}
                {\th \Delta, \dl{A}}}
            {\th \Gamma, \Delta}}
        \quad\longrightarrow\quad
        \vcenter{\prftree[r,d]{\rl{wk}}
            {\prfsummary[\(Q\)]{\th \Delta}}
            {\th \Gamma, \Delta}}
    \)}}\\[.5em]
    \scalebox{0.8}{\makebox[\textwidth]{\(
        \vcenter{\prftree[r,d]{\rl{ctr}}
            {\prftree[r]{\rl{cut}}
                {\prftree[r]{\rl{ctr}}
                    {\prfsummary[\(P\)]{\th \Gamma, A, A \vphantom{\dl{A}}}}
                    {\th \Gamma, A \vphantom{\dl{A}}}}
                {\prftree[r]{\rl{cut}}
                    {\prftree[r]{\rl{ctr}}
                        {\prfsummary[\(P\)]{\th \Gamma, A, A}}
                        {\th \Gamma, A}}
                    {\prfsummary[\(Q\)]{\th \Delta, \dl{A}, \dl{A}}}
                    {\th \Gamma, \Delta, \dl{A}}}
                {\th \Gamma, \Delta, \Delta}}
            {\th \Gamma, \Delta}}
        \quad\longleftarrow\quad
        \vcenter{\prftree[r]{\rl{cut}}
            {\prftree[r]{\rl{ctr}}
                {\prfsummary[\(P\)]{\th \Gamma, A, A \vphantom{\dl{A}}}}
                {\th \Gamma, A \vphantom{\dl{A}}}}
            {\prftree[r]{\rl{ctr}}
                {\prfsummary[\(Q\)]{\th \Delta, \dl{A}, \dl{A}}}
                {\th \Delta, \dl{A}}}
            {\th \Gamma, \Delta}}
        \quad\longrightarrow\quad
        \vcenter{\prftree[r,d]{\rl{ctr}}
            {\prftree[r]{\rl{cut}}
                {\prftree[r]{\rl{cut}}
                    {\prfsummary[\(P\)]{\th \Gamma, A, A}}
                    {\prftree[r]{\rl{ctr}}
                        {\prfsummary[\(Q\)]{\th \Delta, \dl{A}, \dl{A}}}
                        {\th \Delta, \dl{A}}}
                    {\th \Gamma, \Delta, A \vphantom{\dl{A}}}}
                {\prftree[r]{\rl{ctr}}
                    {\prfsummary[\(Q\)]{\th \Delta, \dl{A}, \dl{A}}}
                    {\th \Delta, \dl{A}}}
                {\th \Gamma, \Delta, \Delta}}
            {\th \Gamma, \Delta}}
    \)}}
    \caption{Non-deterministic cut-reduction steps in the classical sequent calculus; also
        called \emph{logical} and \emph{structural dilemmas} in~\cite{DJS97}.}
    \label{fig:non-det-reductions}
\end{figure}

We start from an idea which is well-established in proof-theory, i.e.\ that the way axioms
link together the subformulas of the conclusion is somehow essential to the content of a proof.
Such an idea underlines, e.g., the theory of proof-nets~\cite{Gir87} and the research program
known as \emph{Geometry of Interaction}~\cite{Gir90b}, but has also been widely explored --
for various purposes and in many different ways -- in the classical setting, among others by
Andrews~\cite{And72,And80}, Statman~\cite{Sta74}, Buss and Carbone~\cite{Bus91,Car97},
Lamarche and Straßburger~\cite{LS04,Str11}, Hughes~\cite{Hug06a,Hug06b}, Guglielmi and
Gundersen~\cite{GG08}.

The approach of Andrews, Lamarche and Straßburger in particular is of interest here: it consists
in extracting graphs from each derivation, whose vertices are the atomic formula occurrences
in the conclusion and whose edges join those dual occurrences that are related by some axiom.
Andrews first considered graphs of this kind (which he called \emph{matings}) in~\cite{And72},
and showed in~\cite{And80} how to reconstruct natural deduction proofs from them. Lamarche
and Straßburger used them in~\cite{LS04} to develop a system of proof-nets for classical
logic (which they called~\emph{\(\mathbb{B}\)-nets}), and provided a correctness criterion
and a sequentialization theorem, as well as a confluent and terminating cut-elimination
procedure on nets which corresponds essentially to a composition of graphs by contraction
of alternating paths, in the style of the Geometry of Interaction.

Having such a procedure is of the essence when extracting graphs from derivations with cuts.
One has to trace the history of all atomic occurrences through the derivation, then combine
the traces with axioms and cuts in an appropriate way to obtain the resulting graph. As an
example consider the following derivation:
{\[
    \prftree[r]{\rl{cut}}
        {\prftree[r]{\rl{ctr}}
            {\prftree[r]{\rl{\land}}
                {\prfbyaxiom{\rl{ax}}{\th \dl{\tikzmarknode{x-l-1}{\alpha}}, \tikzmarknode{y-l-1}{\alpha}}}
                {\prfbyaxiom{\rl{ax}}{\th \dl{\alpha}, \alpha}}
                {\th \dl{\tikzmarknode{x-l-2}{\alpha}} \land \dl{\alpha}, \tikzmarknode{y-l-2}{\alpha}, \alpha}}
            {\th \dl{\tikzmarknode{x-l-3}{\alpha}} \land \dl{\alpha}, \tikzmarknode{y-l-3}{\alpha}}}
        {\prftree[r]{\rl{ctr}}
            {\prftree[r]{\rl{\land}}
                {\prfbyaxiom{\rl{ax}}{\th \dl{\alpha}, \alpha}}
                {\prfbyaxiom{\rl{ax}}{\th \dl{\tikzmarknode{y-r-1}{\alpha}}, \tikzmarknode{z-r-1}{\alpha}}}
                {\th \dl{\alpha}, \dl{\tikzmarknode{y-r-2}{\alpha}}, \alpha \land \tikzmarknode{z-r-2}{\alpha}}}
            {\th \dl{\tikzmarknode{y-r-3}{\alpha}}, \alpha \land \tikzmarknode{z-r-3}{\alpha}}}
        {\th \dl{\tikzmarknode{x-4}{\alpha}} \land \dl{\alpha}, \alpha \land \tikzmarknode{z-4}{\alpha}}
\]%
\begin{tikzpicture}[remember picture, overlay, thick, blue]
    % Axioms
    \draw ([yshift=2pt]x-l-1.north) .. controls ++(0,1em) and ++(0,1em) .. ([yshift=2pt]y-l-1.north);
    \draw ([yshift=2pt]y-r-1.north) .. controls ++(0,1em) and ++(0,1em) .. ([yshift=2pt]z-r-1.north);
    % History of x
    \draw ([yshift=2pt]x-4.north) .. controls ++(0,10pt) and ++(0,-10pt) .. ([yshift=-1pt]x-l-3.south);
    \draw ([yshift=2pt]x-l-3.north) .. controls ++(0,3pt) and ++(0,-3pt) .. ([yshift=-1pt]x-l-2.south);
    \draw ([yshift=2pt]x-l-2.north) .. controls ++(0,5pt) and ++(0,-5pt) .. ([yshift=-1pt]x-l-1.south);
    % History of y (left)
    \draw ([yshift=1pt]y-l-3.north) .. controls ++(0,3pt) and ++(0,-3pt) .. ([yshift=-1pt]y-l-2.south);
    \draw ([yshift=1pt]y-l-2.north) .. controls ++(0,5pt) and ++(0,-5pt) .. ([yshift=-1pt]y-l-1.south);
    % History of y (right)
    \draw ([yshift=2pt]y-r-3.north) .. controls ++(0,3pt) and ++(0,-3pt) .. ([yshift=-1pt]y-r-2.south);
    \draw ([yshift=2pt]y-r-2.north) .. controls ++(0,5pt) and ++(0,-5pt) .. ([yshift=-1pt]y-r-1.south);
    % History of z
    \draw ([yshift=1pt]z-4.north) .. controls ++(0,10pt) and ++(0,-10pt) .. ([yshift=-1pt]z-r-3.south);
    \draw ([yshift=1pt]z-r-3.north) .. controls ++(0,3pt) and ++(0,-3pt) .. ([yshift=-1pt]z-r-2.south);
    \draw ([yshift=1pt]z-r-2.north) .. controls ++(0,5pt) and ++(0,-5pt) .. ([yshift=-1pt]z-r-1.south);
    % Cuts
    \path[draw, densely dashed, semithick, rounded corners=2pt]
        ([xshift=1pt]y-l-3.east) -- ([xshift=5pt]y-l-3.south east) -- ([xshift=-5pt]y-r-3.south west) -- ([xshift=-1pt]y-r-3.west);
\end{tikzpicture}}%
We select an atomic formula occurrence in the conclusion and start tracing its history up
through the derivation; when we reach an axiom we move to the linked dual occurrence and
start traveling down; when we reach a cut we move to the corresponding dual occurrence
in the other premiss and start moving up again; we continue doing so until we get back
to the conclusion: the graph shall then contain an edge linking the initial and final
points of the resulting path.

The key idea is however not just to track the axioms, but also to \emph{avoid counting them}.
To understand why this approach might hold promise for a better treatment of the contraction
case, consider the following example:
{\[
    \prftree[r]{\rl{cut}}
        {\prftree[r]{\rl{\land}}
            {\prfbyaxiom{\rl{ax}}{\th \dl{\tikzmarknode{x-1}{\alpha}}, \tikzmarknode{y-1}{\alpha}}}
            {\prftree[r]{\rl{ctr}}
                {\prfsummary{\th B, C, C}}
                {\th B, C \vphantom{\dl{\alpha}}}}
            {\th \dl{\tikzmarknode{x-2}{\alpha}}, \tikzmarknode{y-2}{\alpha} \land B, C}}
        {\prftree[r]{\rl{ctr}}
            {\prfsummary{\th \dl{C}, \dl{C}}}
            {\th \dl{C}}}
        {\th \dl{\tikzmarknode{x-3}{\alpha}}, \tikzmarknode{y-3}{\alpha} \land B}
\]%
\begin{tikzpicture}[remember picture, overlay, thick, blue]
    % Axioms
    \draw ([yshift=2pt]x-1.north) .. controls ++(0,1em) and ++(0,1em) .. ([yshift=2pt]y-1.north);
    % History of x
    \draw ([yshift=2pt]x-3.north) .. controls ++(0,5pt) and ++(0,-10pt) .. ([yshift=-1pt]x-2.south);
    \draw ([yshift=2pt]x-2.north) .. controls ++(0,3pt) and ++(0,-7pt) .. ([yshift=-1pt]x-1.south);
    % History of y
    \draw ([yshift=1pt]y-3.north) .. controls ++(0,10pt) and ++(0,-5pt) .. ([yshift=-1pt]y-2.south);
    \draw ([yshift=1pt]y-2.north) .. controls ++(0,7pt) and ++(0,-3pt) .. ([yshift=-1pt]y-1.south);
\end{tikzpicture}}%
Notice that there is exactly one axiom rule linking the two dual occurences of~\(\alpha\)
in the conclusion. There are two possible ways to reduce the cut:
\[
    \scalebox{0.8}{\makebox[\textwidth]{\(
        \vcenter{\prftree[r]{\rl{\land}}
            {\prfbyaxiom{\rl{ax}}{\th \dl{\tikzmarknode{x-1}{\alpha}}, \tikzmarknode{y-1}{\alpha}}}
            {\prftree[r]{\rl{cut}}
                {\prftree[r]{\rl{ctr}}
                    {\prfsummary{\th B, C, C}}
                    {\th B, C \vphantom{\dl{C}}}}
                {\prftree[r]{\rl{ctr}}
                    {\prfsummary{\th \dl{C}, \dl{C}}}
                    {\th \dl{C}}}
                {\th B \vphantom{\dl{\alpha}}}}
            {\th \dl{\tikzmarknode{x-2}{\alpha}}, \tikzmarknode{y-2}{\alpha} \land B}}
        \qquad\qquad
        \vcenter{\prftree[r,d]{\rl{ctr}}
            {\prftree[r]{\rl{cut}}
                {\prftree[r]{\rl{\land}}
                    {\prfbyaxiom{\rl{ax}}{\th \dl{\tikzmarknode{z-1}{\alpha}}, \tikzmarknode{u-1}{\alpha}}}
                    {\prftree[r]{\rl{ctr}}
                        {\prfsummary{\th B, C, C}}
                        {\th B, C \vphantom{\dl{\alpha}}}}
                    {\th \dl{\tikzmarknode{z-2}{\alpha}}, \tikzmarknode{u-2}{\alpha} \land B, C}}
                {\prftree[r]{\rl{cut}}
                    {\prftree[r]{\rl{\land}}
                        {\prfbyaxiom{\rl{ax}}{\th \dl{\tikzmarknode{v-1}{\alpha}}, \tikzmarknode{w-1}{\alpha}}}
                        {\prftree[r]{\rl{ctr}}
                            {\prfsummary{\th B, C, C}}
                            {\th B, C \vphantom{\dl{\alpha}}}}
                        {\th \dl{\tikzmarknode{v-2}{\alpha}}, \tikzmarknode{w-2}{\alpha} \land B, C}}
                    {\prfsummary{\th \dl{C}, \dl{C}}}
                    {\th \dl{\tikzmarknode{v-3}{\alpha}}, \tikzmarknode{w-3}{\alpha} \land B, \dl{C}}}
                {\th \dl{\tikzmarknode{z-3}{\alpha}}, \tikzmarknode{u-3}{\alpha} \land B, \dl{\tikzmarknode{v-4}{\alpha}}, \tikzmarknode{w-4}{\alpha} \land B}}
            {\th \dl{\tikzmarknode{zv}{\alpha}}, \tikzmarknode{uw}{\alpha} \land B}}
    \)}%
    \begin{tikzpicture}[remember picture, overlay, thick, blue]
        % Axioms
        \draw ([yshift=2pt]x-1.north) .. controls ++(0,1em) and ++(0,1em) .. ([yshift=2pt]y-1.north);
        \draw ([yshift=2pt]z-1.north) .. controls ++(0,1em) and ++(0,1em) .. ([yshift=2pt]u-1.north);
        \draw ([yshift=2pt]v-1.north) .. controls ++(0,1em) and ++(0,1em) .. ([yshift=2pt]w-1.north);
        % History of x
        \draw ([yshift=2pt]x-2.north) .. controls ++(0,5pt) and ++(0,-10pt) .. ([yshift=-1pt]x-1.south);
        % History of y
        \draw ([yshift=1pt]y-2.north) .. controls ++(0,10pt) and ++(0,-5pt) .. ([yshift=-1pt]y-1.south);
        % History of z
        \draw ([yshift=2pt]zv.north) .. controls ++(0,3pt) and ++(0,-7pt) .. ([yshift=-1pt]z-3.south);
        \draw ([yshift=2pt]z-3.north) .. controls ++(0,10pt) and ++(0,-15pt) .. ([yshift=-1pt]z-2.south);
        \draw ([yshift=2pt]z-2.north) .. controls ++(0,3pt) and ++(0,-7pt) .. ([yshift=-1pt]z-1.south);
        % History of u
        \draw ([yshift=1pt]uw.north) .. controls ++(0,7pt) and ++(0,-3pt) .. ([yshift=-1pt]u-3.south);
        \draw ([yshift=1pt]u-3.north) .. controls ++(0,18pt) and ++(0,-10pt) .. ([yshift=-1pt]u-2.south);
        \draw ([yshift=1pt]u-2.north) .. controls ++(0,7pt) and ++(0,-3pt) .. ([yshift=-1pt]u-1.south);
        % History of v
        \draw ([yshift=2pt]zv.north) .. controls ++(0,7pt) and ++(0,-3pt) .. ([yshift=-1pt]v-4.south);
        \draw ([yshift=2pt]v-4.north) .. controls ++(0,15pt) and ++(0,-10pt) .. ([yshift=-1pt]v-3.south);
        \draw ([yshift=2pt]v-3.north) .. controls ++(0,5pt) and ++(0,-10pt) .. ([yshift=-1pt]v-2.south);
        \draw ([yshift=2pt]v-2.north) .. controls ++(0,3pt) and ++(0,-7pt) .. ([yshift=-1pt]v-1.south);
        % History of w
        \draw ([yshift=1pt]uw.north) .. controls ++(0,3pt) and ++(0,-7pt) .. ([yshift=-1pt]w-4.south);
        \draw ([yshift=1pt]w-4.north) .. controls ++(0,10pt) and ++(0,-15pt) .. ([yshift=-1pt]w-3.south);
        \draw ([yshift=1pt]w-3.north) .. controls ++(0,10pt) and ++(0,-5pt) .. ([yshift=-1pt]w-2.south);
        \draw ([yshift=1pt]w-2.north) .. controls ++(0,7pt) and ++(0,-3pt) .. ([yshift=-1pt]w-1.south);
    \end{tikzpicture}}
\]
On the left, the cut has been commuted with the conjunction rule. Axioms in the new derivation
have not changed at all. On the right, the cut has been reduced into two cuts of lower complexity,
and the left subderivation has been duplicated. There are now clearly more axioms than before,
\emph{but they link the exact same subformulas.} For example, two axioms now link the dual
occurrences of~\(\alpha\) in the conclusion, but no new link has been created.

The question is then whether axiom-induced graphs, being insensitive to changes in the number
of axioms in a derivation, are indeed preserved by all non-deterministic cut-reduction steps.
Lamarche and Straßburger were motivated by results obtained in the family of formalisms known
as \emph{Deep Inference}, and while they provided an interpretation of the sequent calculus
into \(\mathbb{B}-nets\), as far as we know they didn't investigate in detail the behaviour
of the interpretation under cut-elimination. Führmann and Pym have proven in~\cite{FP04}
-- as a general theorem for a whole class of interpretations of the classical sequent
calculus -- that the axiom-induced graphs cannot gain edges under cut-reduction, i.e.\ that
our intuition that no axioms are created is indeed correct. They also show that the graphs
are preserved by the two logical cut-reduction steps, and it is possible to prove that they
are also invariant under atomic weakening reductions when the non-deterministic case is solved
by the mix rule.

Unfortunately, it turns out they are \emph{not} preserved by all cut-reduction steps involving
contraction, as there are cases where some paths disappear. The prototypical counter-example,
also due to Führmann and Pym~\cite{FP04}, looks like this:
{\[
    \prftree[r]{\rl{cut}}
        {\prftree[r]{\rl{ctr}}
            {\prftree[r]{\rl{\land}}
                {\prftree[r]{\rl{\land}}
                    {\prfsummary{\th A}}
                    {\prfsummary[\(P\)]{\th \tikzmarknode{bc-l-1}{B}, \tikzmarknode{ac-l-1}{A}}}
                    {\th A \land \tikzmarknode{bc-l-2}{B}, \tikzmarknode{ac-l-2}{A}}}
                {\prfsummary{\th B}}
                {\th A \land \tikzmarknode{bc-l-3}{B}, \tikzmarknode{ac-l-3}{A} \land B}}
            {\th \tikzmarknode{ac-l-4}{A} \land \tikzmarknode{bc-l-4}{B}}}
        {\prftree[r]{\rl{ctr}}
            {\prftree[r]{\rl{\land}}
                {\prftree[r]{\rl{\lor}}
                    {\prftree[r]{\rl{wk}}
                        {\prfbyaxiom{\rl{ax}}{\th \dl{\tikzmarknode{ac-r-1}{A}}, \tikzmarknode{a-r-1}{A}}}
                        {\th \dl{\tikzmarknode{ac-r-2}{A}}, \dl{B}, \tikzmarknode{a-r-2}{A}}}
                    {\th \dl{\tikzmarknode{ac-r-3}{A}} \lor \dl{B}, \tikzmarknode{a-r-3}{A}}}
                {\prftree[r]{\rl{\lor}}
                    {\prftree[r]{\rl{wk}}
                        {\prfbyaxiom{\rl{ax}}{\th \dl{\tikzmarknode{bc-r-1}{B}}, \tikzmarknode{b-r-1}{B}}}
                        {\th \dl{A}, \dl{\tikzmarknode{bc-r-2}{B}}, \tikzmarknode{b-r-2}{B}}}
                    {\th \dl{A} \lor \dl{\tikzmarknode{bc-r-3}{B}}, \tikzmarknode{b-r-3}{B}}}
                {\th \dl{\tikzmarknode{ac-r-4}{A}} \lor \dl{B}, \dl{A} \lor \dl{\tikzmarknode{bc-r-4}{B}}, \tikzmarknode{a-r-4}{A} \land \tikzmarknode{b-r-4}{B}}}
            {\th \dl{\tikzmarknode{ac-r-5}{A}} \lor \dl{\tikzmarknode{bc-r-5}{B}}, \tikzmarknode{a-r-5}{A} \land \tikzmarknode{b-r-5}{B}}}
        {\th \tikzmarknode{a-6}{A} \land \tikzmarknode{b-6}{B}}
\]%
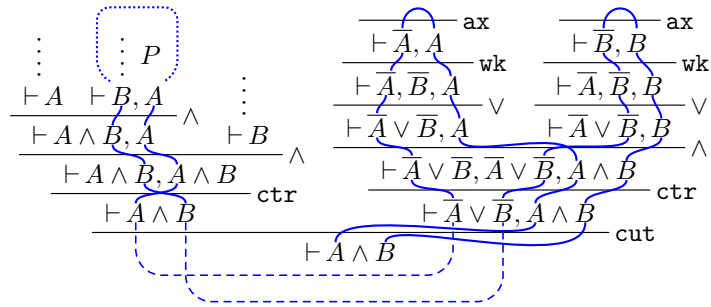
\begin{tikzpicture}[remember picture, overlay, thick, blue]
    % Axioms
    \path[draw, densely dotted, rounded corners=.5em]
        ([yshift=2pt]bc-l-1.north west) -- ([shift={(-5pt,5pt)}]bc-l-1.north west) -- ([shift={(-5pt,3em)}]bc-l-1.north west)
            -- ([shift={(5pt,3em)}]ac-l-1.north east) -- ([shift={(5pt,5pt)}]ac-l-1.north east) -- ([yshift=2pt]ac-l-1.north east);
    \draw ([yshift=2pt]ac-r-1.north) .. controls ++(0,1em) and ++(0,1em) .. ([yshift=2pt]a-r-1.north);
    \draw ([yshift=2pt]bc-r-1.north) .. controls ++(0,1em) and ++(0,1em) .. ([yshift=2pt]b-r-1.north);
    % History of bc (left)
    \draw (bc-l-4.north) .. controls ++(0,5pt) and ++(0,-5pt) .. (bc-l-3.south);
    \draw (bc-l-3.north) .. controls ++(0,5pt) and ++(0,-5pt) .. (bc-l-2.south);
    \draw (bc-l-2.north) .. controls ++(0,3pt) and ++(0,-3pt) .. (bc-l-1.south);
    % History of ac (left)
    \draw (ac-l-4.north) .. controls ++(0,5pt) and ++(0,-5pt) .. (ac-l-3.south);
    \draw (ac-l-3.north) .. controls ++(0,5pt) and ++(0,-5pt) .. (ac-l-2.south);
    \draw (ac-l-2.north) .. controls ++(0,3pt) and ++(0,-3pt) .. (ac-l-1.south);
    % History of ac (right)
    \draw ([yshift=2pt]ac-r-5.north) .. controls ++(0,5pt) and ++(0,-5pt) .. (ac-r-4.south);
    \draw ([yshift=2pt]ac-r-4.north) .. controls ++(0,5pt) and ++(0,-5pt) .. (ac-r-3.south);
    \draw ([yshift=2pt]ac-r-3.north) .. controls ++(0,3pt) and ++(0,-3pt) .. (ac-r-2.south);
    \draw ([shift={(2pt,2pt)}]ac-r-2.north) .. controls ++(0,3pt) and ++(0,-3pt) .. (ac-r-1.south);
    % History of a (right)
    \draw (a-6.north) .. controls ++(0,15pt) and ++(0,-10pt) .. (a-r-5.south);
    \draw (a-r-5.north) .. controls ++(0,5pt) and ++(0,-5pt) .. (a-r-4.south);
    \draw (a-r-4.north) .. controls ++(0,15pt) and ++(0,-10pt) .. (a-r-3.south);
    \draw (a-r-3.north) .. controls ++(0,3pt) and ++(0,-3pt) .. (a-r-2.south);
    \draw (a-r-2.north) .. controls ++(0,5pt) and ++(0,-5pt) .. (a-r-1.south);
    % History of bc (right)
    \draw ([yshift=2pt]bc-r-5.north) .. controls ++(0,5pt) and ++(0,-5pt) .. (bc-r-4.south);
    \draw ([yshift=2pt]bc-r-4.north) .. controls ++(0,7pt) and ++(0,-7pt) .. (bc-r-3.south);
    \draw ([yshift=2pt]bc-r-3.north) .. controls ++(0,3pt) and ++(0,-3pt) .. (bc-r-2.south);
    \draw ([yshift=2pt]bc-r-2.north) .. controls ++(0,3pt) and ++(0,-3pt) .. (bc-r-1.south);
    % History of b (right)
    \draw (b-6.north) .. controls ++(0,10pt) and ++(0,-20pt) .. (b-r-5.south);
    \draw (b-r-5.north) .. controls ++(0,5pt) and ++(0,-5pt) .. (b-r-4.south);
    \draw (b-r-4.north) .. controls ++(0,5pt) and ++(0,-5pt) .. (b-r-3.south);
    \draw (b-r-3.north) .. controls ++(0,3pt) and ++(0,-3pt) .. (b-r-2.south);
    \draw (b-r-2.north) .. controls ++(0,3pt) and ++(0,-3pt) .. (b-r-1.south);
    % Cuts
    \path[draw, densely dashed, semithick, rounded corners=1em]
        (ac-l-4.south) -- ([yshift=-2em]ac-l-4.south) -- ([yshift=-2em]ac-r-5.south) -- (ac-r-5.south);
    \path[draw, densely dashed, semithick, rounded corners=1em]
        (bc-l-4.south) -- ([yshift=-3em]bc-l-4.south) -- ([yshift=-3em]bc-r-5.south) -- (bc-r-5.south);
\end{tikzpicture}}\\
Observe that, under the assumption that the subderivation \(P\) have a path connecting \(B\)
with \(A\) (represented here as a dotted line), there is in the axiom graph of the complete
derivation a path connecting (some subformula of) \(A\) with (some subformula of) \(B\) in the
conclusion. When the cut is reduced by duplicating the right subderivation and commuting it
up the two conjunctions, nothing bad happens. However, if the cut is reduced by duplicating
the left subderivation, the path traced above is lost:\\[1em]
{\[
    \makebox[\textwidth]{\small\(
        \prftree[r]{\rl{cut}}
            {\prftree[r]{\rl{ctr}}
                {\prftree[r]{\rl{\land}}
                    {\prftree[r]{\rl{\land}}
                        {\prfsummary{\th A}}
                        {\prfsummary[\(P\)]{\th \tikzmarknode{bc-l1-1}{B}, \tikzmarknode{ac-l1-1}{A}}}
                        {\th A \land \tikzmarknode{bc-l1-2}{B}, \tikzmarknode{ac-l1-2}{A}}}
                    {\prfsummary{\th B}}
                    {\th A \land \tikzmarknode{bc-l1-3}{B}, \tikzmarknode{ac-l1-3}{A} \land B}}
                {\th \tikzmarknode{ac-l1-4}{A} \land \tikzmarknode{bc-l1-4}{B}}}
            {\prftree[r]{\rl{cut}}
                {\prftree[r]{\rl{ctr}}
                    {\prftree[r]{\rl{\land}}
                        {\prftree[r]{\rl{\land}}
                            {\prfsummary{\th A}}
                            {\prfsummary[\(P\)]{\th \tikzmarknode{bc-l2-1}{B}, \tikzmarknode{ac-l2-1}{A}}}
                            {\th A \land \tikzmarknode{bc-l2-2}{B}, \tikzmarknode{ac-l2-2}{A}}}
                        {\prfsummary{\th B}}
                        {\th A \land \tikzmarknode{bc-l2-3}{B}, \tikzmarknode{ac-l2-3}{A} \land B}}
                    {\th \tikzmarknode{ac-l2-4}{A} \land \tikzmarknode{bc-l2-4}{B}}}
                {\prftree[r]{\rl{\land}}
                    {\prftree[r]{\rl{\lor}}
                        {\prftree[r]{\rl{wk}}
                            {\prfbyaxiom{\rl{ax}}{\th \dl{\tikzmarknode{ac-r-1}{A}}, \tikzmarknode{a-r-1}{A}}}
                            {\th \dl{\tikzmarknode{ac-r-2}{A}}, \dl{B}, \tikzmarknode{a-r-2}{A}}}
                        {\th \dl{\tikzmarknode{ac-r-3}{A}} \lor \dl{B}, \tikzmarknode{a-r-3}{A}}}
                    {\prftree[r]{\rl{\lor}}
                        {\prftree[r]{\rl{wk}}
                            {\prfbyaxiom{\rl{ax}}{\th \dl{\tikzmarknode{bc-r-1}{B}}, \tikzmarknode{b-r-1}{B}}}
                            {\th \dl{\tikzmarknode{aw-r-1}{A}}, \dl{\tikzmarknode{bc-r-2}{B}}, \tikzmarknode{b-r-2}{B}}}
                        {\th \dl{\tikzmarknode{aw-r-2}{A}} \lor \dl{\tikzmarknode{bc-r-3}{B}}, \tikzmarknode{b-r-3}{B}}}
                    {\th \dl{\tikzmarknode{ac-r-4}{A}} \lor \dl{B}, \dl{\tikzmarknode{aw-r-3}{A}} \lor \dl{\tikzmarknode{bc-r-4}{B}}, \tikzmarknode{a-r-4}{A} \land \tikzmarknode{b-r-4}{B}}}
                {\th \dl{\tikzmarknode{ac-r-5}{A}} \lor \dl{\tikzmarknode{bc-r-5}{B}}, \tikzmarknode{a-r-5}{A} \land \tikzmarknode{b-r-5}{B}}}
            {\th \tikzmarknode{a-6}{A} \land \tikzmarknode{b-6}{B}}
    \)}
\]%
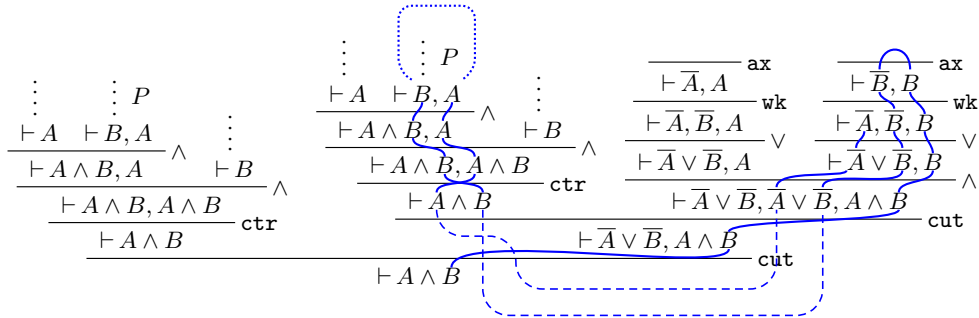
\begin{tikzpicture}[remember picture, overlay, thick, blue]
    % Axioms
    % \path[draw, densely dotted, rounded corners=.5em]
    %     ([yshift=2pt]bc-l1-1.north west) -- ([shift={(-5pt,5pt)}]bc-l1-1.north west) -- ([shift={(-5pt,3em)}]bc-l1-1.north west)
    %         -- ([shift={(5pt,3em)}]ac-l1-1.north east) -- ([shift={(5pt,5pt)}]ac-l1-1.north east) -- ([yshift=2pt]ac-l1-1.north east);
    \path[draw, densely dotted, rounded corners=.5em]
        ([yshift=2pt]bc-l2-1.north west) -- ([shift={(-5pt,5pt)}]bc-l2-1.north west) -- ([shift={(-5pt,3em)}]bc-l2-1.north west)
            -- ([shift={(5pt,3em)}]ac-l2-1.north east) -- ([shift={(5pt,5pt)}]ac-l2-1.north east) -- ([yshift=2pt]ac-l2-1.north east);
    % \draw ([yshift=2pt]ac-r-1.north) .. controls ++(0,1em) and ++(0,1em) .. ([yshift=2pt]a-r-1.north);
    \draw ([yshift=2pt]bc-r-1.north) .. controls ++(0,1em) and ++(0,1em) .. ([yshift=2pt]b-r-1.north);
    % % History of bc (left1)
    % \draw (bc-l1-4.north) .. controls ++(0,5pt) and ++(0,-5pt) .. (bc-l1-3.south);
    % \draw (bc-l1-3.north) .. controls ++(0,5pt) and ++(0,-5pt) .. (bc-l1-2.south);
    % \draw (bc-l1-2.north) .. controls ++(0,3pt) and ++(0,-3pt) .. (bc-l1-1.south);
    % % History of ac (left1)
    % \draw (ac-l1-4.north) .. controls ++(0,5pt) and ++(0,-5pt) .. (ac-l1-3.south);
    % \draw (ac-l1-3.north) .. controls ++(0,5pt) and ++(0,-5pt) .. (ac-l1-2.south);
    % \draw (ac-l1-2.north) .. controls ++(0,3pt) and ++(0,-3pt) .. (ac-l1-1.south);
    % History of bc (left2)
    \draw (bc-l2-4.north) .. controls ++(0,5pt) and ++(0,-5pt) .. (bc-l2-3.south);
    \draw (bc-l2-3.north) .. controls ++(0,5pt) and ++(0,-5pt) .. (bc-l2-2.south);
    \draw (bc-l2-2.north) .. controls ++(0,3pt) and ++(0,-3pt) .. (bc-l2-1.south);
    % History of ac (left2)
    \draw (ac-l2-4.north) .. controls ++(0,5pt) and ++(0,-5pt) .. (ac-l2-3.south);
    \draw (ac-l2-3.north) .. controls ++(0,5pt) and ++(0,-5pt) .. (ac-l2-2.south);
    \draw (ac-l2-2.north) .. controls ++(0,3pt) and ++(0,-3pt) .. (ac-l2-1.south);
    % % History of ac (right)
    % \draw ([yshift=2pt]ac-r-5.north) .. controls ++(0,5pt) and ++(0,-5pt) .. (ac-r-4.south);
    % \draw ([yshift=2pt]ac-r-4.north) .. controls ++(0,5pt) and ++(0,-5pt) .. (ac-r-3.south);
    % \draw ([yshift=2pt]ac-r-3.north) .. controls ++(0,3pt) and ++(0,-3pt) .. (ac-r-2.south);
    % \draw ([shift={(2pt,2pt)}]ac-r-2.north) .. controls ++(0,3pt) and ++(0,-3pt) .. (ac-r-1.south);
    % % History of a (right)
    % \draw (a-6.north) .. controls ++(0,15pt) and ++(0,-10pt) .. (a-r-5.south);
    % \draw (a-r-5.north) .. controls ++(0,5pt) and ++(0,-5pt) .. (a-r-4.south);
    % \draw (a-r-4.north) .. controls ++(0,15pt) and ++(0,-10pt) .. (a-r-3.south);
    % \draw (a-r-3.north) .. controls ++(0,3pt) and ++(0,-3pt) .. (a-r-2.south);
    % \draw (a-r-2.north) .. controls ++(0,5pt) and ++(0,-5pt) .. (a-r-1.south);
    % History of aw (right)
    \draw ([yshift=2pt]aw-r-3.north) .. controls ++(0,7pt) and ++(0,-3pt) .. (aw-r-2.south west);
    \draw ([yshift=2pt]aw-r-2.north) .. controls ++(0,3pt) and ++(0,-3pt) .. (aw-r-1.south);
    % History of bc (right)
    \draw ([yshift=2pt]bc-r-4.north) .. controls ++(0,7pt) and ++(0,-7pt) .. (bc-r-3.south);
    \draw ([yshift=2pt]bc-r-3.north) .. controls ++(0,3pt) and ++(0,-3pt) .. (bc-r-2.south);
    \draw ([yshift=2pt]bc-r-2.north) .. controls ++(0,3pt) and ++(0,-3pt) .. (bc-r-1.south);
    % History of b (right)
    \draw (b-6.north) .. controls ++(0,15pt) and ++(0,-10pt) .. (b-r-5.south);
    \draw (b-r-5.north) .. controls ++(0,5pt) and ++(0,-5pt) .. (b-r-4.south);
    \draw (b-r-4.north) .. controls ++(0,5pt) and ++(0,-5pt) .. (b-r-3.south);
    \draw (b-r-3.north) .. controls ++(0,3pt) and ++(0,-3pt) .. (b-r-2.south);
    \draw (b-r-2.north) .. controls ++(0,3pt) and ++(0,-3pt) .. (b-r-1.south);
    % Cuts
    \path[draw, densely dashed, semithick, rounded corners=1em]
        (ac-l2-4.south) -- ([yshift=-1em]ac-l2-4.south) -- ([shift={(3em,-1em)}]ac-l2-4.south)
            -- ([shift={(3em,-3em)}]ac-l2-4.south) -- ([yshift=-3em]aw-r-3.south) -- (aw-r-3.south);
    \path[draw, densely dashed, semithick, rounded corners=1em]
        (bc-l2-4.south) -- ([yshift=-4em]bc-l2-4.south) -- ([yshift=-4em]bc-r-4.south) -- (bc-r-4.south);
\end{tikzpicture}}\\
We show for clarity only the surviving part of the path that is connected to~\(B\). Note how
the new path must end at the weakening rule. The old path depended critically on the ability
to pass through both premises of the contraction rule; now that the contraction rule has
vanished and we have two independent cuts it becomes impossible to construct the same path.

There is however a better way to look at the same problem. Notice how the following two
derivations are associated to the same axiom graph:
{\[
    \vcenter{\prfbyaxiom{\rl{ax}}{\th \dl{\tikzmarknode{ac-l-1}{A}} \lor \dl{\tikzmarknode{bc-l-1}{B}}, \tikzmarknode{a-l-1}{A} \land \tikzmarknode{b-l-1}{B}}}
    \qquad\qquad
    \vcenter{\prftree[r]{\rl{ctr}}
        {\prftree[r]{\rl{\land}}
            {\prftree[r]{\rl{\lor}}
                {\prftree[r]{\rl{wk}}
                    {\prfbyaxiom{\rl{ax}}{\th \dl{\tikzmarknode{ac-r-1}{A}}, \tikzmarknode{a-r-1}{A}}}
                    {\th \dl{\tikzmarknode{ac-r-2}{A}}, \dl{B}, \tikzmarknode{a-r-2}{A}}}
                {\th \dl{\tikzmarknode{ac-r-3}{A}} \lor \dl{B}, \tikzmarknode{a-r-3}{A}}}
            {\prftree[r]{\rl{\lor}}
                {\prftree[r]{\rl{wk}}
                    {\prfbyaxiom{\rl{ax}}{\th \dl{\tikzmarknode{bc-r-1}{B}}, \tikzmarknode{b-r-1}{B}}}
                    {\th \dl{A}, \dl{\tikzmarknode{bc-r-2}{B}}, \tikzmarknode{b-r-2}{B}}}
                {\th \dl{A} \lor \dl{\tikzmarknode{bc-r-3}{B}}, \tikzmarknode{b-r-3}{B}}}
            {\th \dl{\tikzmarknode{ac-r-4}{A}} \lor \dl{B}, \dl{A} \lor \dl{\tikzmarknode{bc-r-4}{B}}, \tikzmarknode{a-r-4}{A} \land \tikzmarknode{b-r-4}{B}}}
        {\th \dl{\tikzmarknode{ac-r-5}{A}} \lor \dl{\tikzmarknode{bc-r-5}{B}}, \tikzmarknode{a-r-5}{A} \land \tikzmarknode{b-r-5}{B}}}
\]
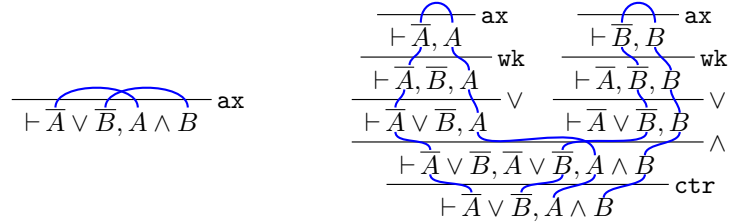
\begin{tikzpicture}[remember picture, overlay, thick, blue]
    % Axioms
    \draw ([yshift=2pt]ac-l-1.north) .. controls ++(0,1em) and ++(0,1em) .. ([yshift=2pt]a-l-1.north);
    \draw ([yshift=2pt]bc-l-1.north) .. controls ++(0,1em) and ++(0,1em) .. ([yshift=2pt]b-l-1.north);
    \draw ([yshift=2pt]ac-r-1.north) .. controls ++(0,1em) and ++(0,1em) .. ([yshift=2pt]a-r-1.north);
    \draw ([yshift=2pt]bc-r-1.north) .. controls ++(0,1em) and ++(0,1em) .. ([yshift=2pt]b-r-1.north);
    % History of ac (right)
    \draw ([yshift=2pt]ac-r-5.north) .. controls ++(0,5pt) and ++(0,-5pt) .. (ac-r-4.south);
    \draw ([yshift=2pt]ac-r-4.north) .. controls ++(0,5pt) and ++(0,-5pt) .. (ac-r-3.south);
    \draw ([yshift=2pt]ac-r-3.north) .. controls ++(0,3pt) and ++(0,-3pt) .. (ac-r-2.south);
    \draw ([shift={(2pt,2pt)}]ac-r-2.north) .. controls ++(0,3pt) and ++(0,-3pt) .. (ac-r-1.south);
    % History of a (right)
    \draw (a-r-5.north) .. controls ++(0,5pt) and ++(0,-5pt) .. (a-r-4.south);
    \draw (a-r-4.north) .. controls ++(0,15pt) and ++(0,-10pt) .. (a-r-3.south);
    \draw (a-r-3.north) .. controls ++(0,3pt) and ++(0,-3pt) .. (a-r-2.south);
    \draw (a-r-2.north) .. controls ++(0,5pt) and ++(0,-5pt) .. (a-r-1.south);
    % History of bc (right)
    \draw ([yshift=2pt]bc-r-5.north) .. controls ++(0,5pt) and ++(0,-5pt) .. (bc-r-4.south);
    \draw ([yshift=2pt]bc-r-4.north) .. controls ++(0,7pt) and ++(0,-7pt) .. (bc-r-3.south);
    \draw ([yshift=2pt]bc-r-3.north) .. controls ++(0,3pt) and ++(0,-3pt) .. (bc-r-2.south);
    \draw ([yshift=2pt]bc-r-2.north) .. controls ++(0,3pt) and ++(0,-3pt) .. (bc-r-1.south);
    % History of b (right)
    \draw (b-r-5.north) .. controls ++(0,5pt) and ++(0,-5pt) .. (b-r-4.south);
    \draw (b-r-4.north) .. controls ++(0,5pt) and ++(0,-5pt) .. (b-r-3.south);
    \draw (b-r-3.north) .. controls ++(0,3pt) and ++(0,-3pt) .. (b-r-2.south);
    \draw (b-r-2.north) .. controls ++(0,3pt) and ++(0,-3pt) .. (b-r-1.south);
\end{tikzpicture}}%
The peculiarity of this graph is that it is behaves as an identity w.r.t.\ composition,
i.e.\ whatever derivation \(P\) we may cut against one of the two derivations above,
the resulting graph will be precisely that of~\(P\). This is not in the least surprising
in the case of the derivation of the left, which is just an axiom hence an identity
w.r.t.\ cut-elimination. The situation is different for the derivation on the right,
which is not in general an identity w.r.t.\ cut-elimination.

The principle expressed by the derivation on the right is the \emph{syntactic invertibility
of conjunctions}, i.e.\ the fact that any derivation~\(P\) of the sequent~\(\th \Gamma, A \land B\)
can be turned into a derivation~\(P'\) of the same sequent whose last rule (up to some auxiliary
contraction on the context) introduces the conjunction \(A \land B\). One way to obtain
that result is precisely to cut \(P\) with the derivation on the right, then eliminate
the cut by duplicating \(P\) and commuting its copies up.

The difficulty with this formulation of the classical sequent calculus is that it contains
derivations of conjunctions whose axiom graph cannot be expressed by a derivation ending
with the conjunction rule. This is precisely the case for the left subderivation of the
counter-example described above. The result of our analysis suggests the conjecture
that drives our methodological approach: either axiom graphs are the natural semantics
of a proof system where the invertibility of logical rules is a fundamental property,
and essentially an identity, or at least we may hope to solve the problem by moving
to such a system and refining the notion of axiom graph until it becomes invariant under
inversion of logical rules.

The sequent calculus \(\data{GS4}\) (displayed in \cref{fig:gs4}), both in full and in
its cut-free fragment, is precisely such a system. It is the one-sided presentation
of Kleene's context-sharing style calculus \(\data{G4}\)~\cite{Kle67,Hug10,TS00},
where independent rule applications permute freely. It is known that, when the axioms
of~\(\data{GS4}\) are restricted to atomic conclusions, the set of axioms of any derivation
is invariant under arbitrary permutations of logical rules~\cite{PP20,NP01}. The calculus
admits an elegant proof of completeness and enjoys a cut-elimination procedure~\cite{Pul22}.

\begin{figure}[t]
    \centering
    \begin{gather*}
        \vcenter{\prfbyaxiom{\rl{ax}\ (\(\Gamma\) atomic)}{\th \Gamma, \alpha, \dl{\alpha}}}
        \qquad\qquad
        \vcenter{\prftree[r]{\rl{cut}}
            {\prfassumption{\th \Gamma, A}}
            {\prfassumption{\th \Gamma, \dl{A}}}
            {\th \Gamma}}
        \\[1em]
        \prftree[r]{\rl{\lor}}
            {\prfassumption{\th \Gamma, A, B}}
            {\th \Gamma, A \lor B}
        \qquad\qquad
        \prftree[r]{\rl{\land}}
            {\prfassumption{\th \Gamma, A}}
            {\prfassumption{\th \Gamma, B}}
            {\th \Gamma, A \land B}
    \end{gather*}
    \caption{The sequent calculus \(\data{GS4}\).}
    \label{fig:gs4}
\end{figure}

One difficulty in adopting \(\data{GS4}\) for our investigation is that every provable
sequent has a unique derivation up to permutations of logical rules. This is due to
the fact that weakenings are absorbed into the axioms and sequents are identified up
to arbitrary permutations of their elements, hence we lose in general the ability to
tell which pairs of atomic formulas are linked by the axiom:
\[
    \prfbyaxiom{\rl{ax}}{\th \beta, \dl{\beta}, \alpha, \dl{\alpha}}
    \quad=\quad
    \prfbyaxiom{\rl{ax}}{\th \alpha, \dl{\alpha}, \beta, \dl{\beta}}
\]
In order to restore that ability we perform two steps (\cref{sec:preliminaries}):
\begin{itemize}
    \item we replace usual formulas with \emph{named formulas} (\cref{defn:formulas}),
    i.e.\ formulas where each atomic formula occurrence is assigned a distinct name:
    \[
        \nm{x}{\alpha} \lor (\nm{y}{\beta} \land \nm{z}{\alpha}).
    \]
    It is a tedious but unfortunately necessary technical detail, as it is the only way
    keep track of the history of each occurrence through the derivation;\footnote{The
    alternative is the sequents-as-lists approach, which requires no modification whatsoever
    to the calculus, but makes reasoning on proof transformations exceedingly complicated
    because of the need to insert exchange rules everywhere.}

    \item we enrich the calculus with a \emph{deterministic axiom} and a \emph{superposition
    rule}:
    \[
        \vcenter{\prfbyaxiom{\rl[\{\nm{x}{\alpha},\nm{y}{\dl{\alpha}}\}]{ax}}
            {\th \Gamma, \nm{x}{\alpha},\nm{y}{\dl{\alpha}}}}
        \qquad\qquad
        \vcenter{\prftree[r]{\rl{\sqcup}}
            {\prfassumption{\th \Gamma}}
            {\prfassumption{\th \Gamma}}
            {\th \Gamma}}
    \]
\end{itemize}
We additionally allow the use of non-atomic axioms, which are interesting in that they
greatly reduce the size of derivations; this makes reasoning about the calculus a bit
more complicated, but it provides us with a nice proof of derivability of the contraction
rule (\cref{sec:ctr-wkn}).

Deterministic axioms allow us to recover the standard form that axioms have in calculi
with explicit structural rules without losing the benefits of axioms with embedded
weakenings. They essentially provide a way to tell which formulas in the context come
from a weakening, and which ones do not. Observe that this is not a trivial addition
to~\(\data{GS4}\): the complexity of checking the correctness of standard axiom rule
applications ranges from linear to quadratic in the size of the context, while for
deterministic axioms it ranges from sub-linear to linear.

The superposition rule is an adaptation of the mix-rule to the context-sharing framework.
The shape of the rule is reminiscent of the non-deterministic sums discussed earlier,
but its behaviour under cut-elimination is, as we shall see, closer to that of the
mix-rule. We choose to adopt a non-standard name for the rule to remark this fact
and avoid confusion.

In \cref{sec:inversion-lemmas} we recall the main results about the invertibility of
logical rules in~\(\data{GS4}\), and introduce some related notation. We slightly depart
from standard terminology in that we speak of \emph{inversion} when recovering derivations
of the premises of a logical rule, while we speak of \emph{isolation} when permuting
rule applications to the bottom of the derivation. The two properties are provably equivalent,
hence they often receive the same name. In our case, however, we needed a way to distinguish
the two operations.

We are then able to formalize the axiom graph construction on the enriched calculus
(\cref{sec:axiom-graphs}, \cref{defn:axiom-graphs}) and show the remarkable fact that
axiom graphs are preserved under inversion of logical rules in the cut-free fragment
(\cref{propo:axiom-graphs-isl-cut-free-invariance}). The full calculus on the other hand
fails to satisfy the same property; two counter-examples are provided in \cref{fig:axiom-graph-inequality-example1,%
fig:axiom-graph-inequality-example2} and discussed in \cref{sec:axiom-graphs-failure}.

This motivates us to seek a refinement of the axiom graph construction. The main problem
identified in \cref{sec:axiom-graphs-failure} is that the composition operator for axiom
graphs might join edges which occur in incompatible branches of the derivation, i.e.\ branches
that prove distinct conjuncts of some conjunction occurring in the conclusion. Thus a
path is obtained which provably cannot exist in a cut-free derivation.

Our solution is to enrich axiom graphs with labels tracking for each edge the branch it
came from: we are then able to define a branch-sensitive composition operator (\cref{defn:bl-composition}).
We develop the refinement in \cref{sec:bl-axiom-graphs}, then show in \cref{sec:bl-behaviour}
that the new construction is invariant under inversion of logical rules in the full fragment
of the calculus (\cref{thm:bl-axiom-graphs-invariance}).

Unfortunately, the refined axiom graphs are no longer preserved by the standard logical
cut-reduction steps (as defined e.g.\ in~\cite{Pul22}). The reasons are subtle, we discuss
them briefly along with possible solutions in \cref{sec:cut-reduction-failure}.

We then put our hand to developing a new normalization procedure that preserves the refined
axiom graphs. Naturally, invertibility is at the core of the procedure (described in the
proof of \cref{thm:normalisation}), which exploits it to permute all cuts up the derivation
until they are reduced to \emph{atomic contexts}, i.e.\ they are of the form
\[
    \prftree[r]{\rl{cut}}
        {\prfassumption{\th \Gamma, A}}
        {\prfassumption{\th \Gamma, \dl{A}}}
        {\th \Gamma}
\]
where \(\Gamma\) contains only atomic formulas, while \(A\) may be arbitrarily complex.
The lack of a proper cut-reduction procedure (\cref{sec:cut-reduction-failure}) forces
us to proceed from this point with a normalisation-by-evaluation argument, i.e.\ we show
that whenever the graph of a derivation is not empty, there is a cut-free derivation of
the same conclusion and associated to the same graph.

Such a result was essentially already available in~\cite{LS04}. We manage to restrict
the need for this kind of argument to the significantly simpler case of cuts with atomic
context between cut-free derivations. Nonetheless the main lemma (\cref{lemma:semantic-cut-admissibility},
a kind of graph-based cut-admissibility result) still requires a very complex proof (provided
in \cref{sec:semantic-cut-admissibility-proof}) which a proper cut-reduction procedure,
if available, would simplify greatly.

Finally, we provide in \cref{sec:totality} a direct characterisation of the class of
axiom graphs that come from \(\data{GS4}\) derivations. The resulting condition, which
we call \emph{totality}, is analogous to the correctness criterion of~\cite{LS04} and
provides a kind of sequentialization theorem \cref{thm:sequentialization}. We exploit
this fact to devise a classical proof system (which we call~\(\data{BLG}\)) optimized
for invertibility, in the sense that the inversion procedures become immediate and the
isolation procedures are just identities. Proofs in~\(\data{BLG}\) have the following
shape:
\[
    \makebox{\prftree
        {\blgbranch{\nm{x}{\alpha}, \nm{z}{\beta}, \nm{v}{\dl{\alpha}}, \nm{w}{\gamma}}{x/v}}
        {\blgbranch{\nm{x}{\alpha}, \nm{u}{\dl{\gamma}}, \nm{v}{\dl{\alpha}}, \nm{w}{\gamma}}{x/v, u/w}}
        {\blgbranch{\nm{y}{\dl{\beta}}, \nm{z}{\beta}, \nm{v}{\dl{\alpha}}, \nm{w}{\gamma}}{y/z}}
        {\blgbranch{\nm{y}{\dl{\beta}}, \nm{u}{\dl{\gamma}}, \nm{v}{\dl{\alpha}}, \nm{w}{\gamma}}{u/w}}
        {\th \nm{x}{\alpha} \land \nm{y}{\dl{\beta}}, \nm{z}{\beta} \land \nm{u}{\dl{\gamma}}, \nm{v}{\dl{\alpha}} \lor \nm{w}{\gamma}}}
\]
i.e.\ a set of atomic decompositions of the conclusion together with a graph providing
the axiom links. All \(\data{GS4}\) rules, including the cut-rule, are admissible in~\(\data{BLG}\).

It is important to stress how \(\data{BLG}\) is different from the \(\mathbb{B}\)-nets
of~\cite{LS04}, apart from disallowing some proofs. The question of whether a certain
formalism provides a proof system for propositional classical logic is delicate. Because
propositional tautologies are decidable in exponential time in the size of the formula,
it has been argued (e.g.\ in~\cite{CR79,Hug06b}) that correctness criteria for proofs should
have significantly lower complexity. In particular, Cook and Reckhow argue in~\cite{CR79}
that, as a minimal requirement, proofs should be checkable in polynomial time for some
adequate notion of proof size.

\(\mathbb{B}\)-nets notoriously failed to be a proof system in this sense: the only known
correctness criterion is exponential in the size of the proof object, which is itself polynomially
bounded by the size of the conclusion; moreover, Das showed that the existence of a polynomial
time correctness criterion for Andrews' matings~\cite{And72} or for \(\mathbb{B}\)-nets
would imply \(\mathbf{NP} = \mathbf{coNP}\).

The situation is different for~\(\data{BLG}\), where we are able to argue (quite informally)
that correctness is checkable in polynomial time in the size of proof objects (\cref{propo:polynomial-time-totality}).
The idea is that a correctness check amounts to constructing a~\(\data{GS4}\) derivation of the
same conclusion, then check that the decomposition provided by the \(\data{BLG}\) proof
object matches the one obtained through~\(\data{GS4}\). Crucially, it is possible to construct
the derivation one branch at a time and thus check the correctness of the decomposition
incrementally: we are then able to show that the number of required steps is polynomially
bounded by the size of the~\(\data{BLG}\) proof object.

\section{Tracking atom occurrences}\label{sec:preliminaries}

Let us fix a \emph{countably infinite} set~\(\mathcal{N}\) of~\emph{names} and an
arbitrary set~\(\mathcal{A}\) of~\emph{propositional atoms}, together with a fixpoint-free
involution
\[
    \sdl{(\cdot)}: \mathcal{A} \to \mathcal{A}
\]
i.e.\ a map such that \(\alpha \neq \dl\alpha\) and~\(\dl{\dl\alpha} = \alpha\) for all
atoms~\(\alpha \in \mathcal{A}\): this means that atoms come in pairs symmetrically related
by the involution. We use letters \(x, y, z, \ldots\) to range over~\(\mathcal{N}\) and greek
letters \(\alpha, \beta, \gamma, \ldots\) to range over~\(\mathcal{A}\).

\begin{definition}[Named formulas]
    \label{defn:formulas}
    The set \(\mathcal{F}\) of \emph{classical propositional formulas} is defined by the grammar
    \[
        F, G \cceq \alpha \mid F \lor G \mid F \land G
    \]
    where \(\alpha \in \mathcal{A}\). \emph{Named formulas} are formulas where every
    occurrence of a propositional atom is assigned a name from the set \(\mathcal{N}\).
    Formally, we define the set \(\mathcal{F^N}\) of \emph{formulas with names
    in~\(\mathcal{N}\)} by the grammar
    \[
        A, B \cceq \nm{x}{\alpha} \mid A \lor B \mid A \land B
    \]
    where \(\alpha \in \mathcal{A}\) and~\(x \in \mathcal{N}\). Call (possibly named)
    formulas of the form \(\alpha\) (resp.~\(\nm{x}{\alpha}\)) \emph{atomic}. Each named
    formula~\(A \in \mathcal{F^N}\) is naturally associated to the anonymous formula~\(|A|
    \in \mathcal{F}\) obtained by forgetting all names. Write \(A \equiv B\) iff \(|A| = |B|\),
    i.e.\ if \(A\) and~\(B\) are identical up to a change of names.
\end{definition}

We define as usual an involution \(\sdl{(\cdot)}: \mathcal{F} \to \mathcal{F}\) (resp.~%
\(\sdl{(\cdot)}: \mathcal{F^N} \to \mathcal{F^N}\)) over propositional formulas and named
formulas, expressing negation through De Morgan's duality:
\[\arraycolsep=1em
    \begin{array}{ccc}
        \sdl{(\alpha)} = \dl\alpha & \dl{F \lor G} = \dl{F} \land \dl{G} & \dl{F \land G} = \dl{F} \lor \dl{G}; \\[.5em]
        \sdl{(\nm{x}{\alpha})} = \nm{x}{\dl\alpha} & \dl{A \lor B} = \dl{A} \land \dl{B} & \dl{A \land B} = \dl{A} \lor \dl{B}.
    \end{array}
\]

\begin{definition}[Named sequents]
    \label{defn:named-sequents}
    \emph{Named classical propositional sequents} are expressions of the form \(\th \Gamma,
    \th \Delta, \ldots\) where \(\Gamma, \Delta, \ldots\) are \emph{finite sets} of named
    formulas. Write \(\Gamma \equiv \Delta\) (resp.~\((\th \Gamma) \equiv (\th \Delta)\))
    if sets \(\Gamma, \Delta\) (resp.\ sequents \(\th \Gamma, \th \Delta\)) are identical up
    to a change of names, formally if there is a \emph{bijection} \(\phi: \Gamma \to \Delta\)
    such that \(A \equiv \phi A\) for all \(A \in \Gamma\).
\end{definition}

For every named formula~\(A \in \mathcal{F^N}\), let~\(\data{names}(A)\) denote the set of
all names occurring in~\(A\), with the obvious inductive definition. The definition extends
easily to sets and named sequents by taking unions over all their members, i.e.
\[
    \data{names}(\Gamma) = \data{names}(\th \Gamma) = \bigcup_{A \in \Gamma} \data{names}(A),
\]
where \(\Gamma\) is any finite set of named formulas.

\begin{definition}[Sharing-free formulas, sets, sequents]
    \label{defn:sharing-freedom}
    Say that named formulas \(A, B\) (resp.\ sets~\(\Gamma, \Delta\) of formulas,
    sequents~\(\th \Gamma, \th \Delta\)) \emph{share names} if their name sets overlap
    (i.e.\ have non-empty intersection), otherwise say that they \emph{share no names};
    in the case of formulas we may also say that they are \emph{disjoint}, while we shall
    not use this terminology with sets and sequents to avoid confusion with the set-theoretical
    meaning of the term.

    Call a named formula~\(A\) \emph{sharing-free} if each name appears at most once in~\(A\);
    formally, by structural induction, if either \(A\) is atomic or it is of the form~\(B \lor C\),
    \(B \land C\) where \(B, C\) are disjoint and themselves sharing-free. Call a set~\(\Gamma\)
    of named formulas \emph{sharing-free} if all formulas in~\(\Gamma\) are sharing-free and
    pairwise disjoint. Call a sequent~\(\th \Gamma\) \emph{sharing-free} iff so is the
    set~\(\Gamma\).
\end{definition}

\begin{definition}[Atom indexing notation]
    \label{defn:atom-indexing}
    Let \(A\) (resp.~\(\Gamma\)) be a \emph{sharing-free} named formula (resp.\ set of named
    formulas). By \cref{defn:sharing-freedom} there is for each name \(x \in \data{names}(A)\)
    (resp.~\(\data{names}(\Gamma)\)) a unique atom \(\alpha \in \mathcal{A}\) such that
    \(\nm{x}{\alpha}\) is a subformula of \(A\) (resp.\ of~\(\Gamma\)). We write \(A[x]\)
    (resp.~\(\Gamma[x]\)) to denote that unique~\(\alpha\).
\end{definition}

Note that negation preserves names when applied to named formulas, hence for all~\(A \in
\mathcal{F^N}\) we have \(\data{names}(A) = \data{names}(\dl{A})\), and~\(\dl{A}[x] =
\dl{A[x]}\) for all~\(x \in \data{names}(A)\).

From now on we are going to assume that \emph{all} named formulas and sequents be sharing-free.
As a consequence, when writing down sequents through the customary comma notation:
\[
    \th \Gamma, A, B, \Delta, \ldots
\]
we shall assume implicitly that every pair of comma-separated components share no names. Observe
that disjoint formulas are necessarily distinct, and therefore any pair of sets of named formulas
sharing no names must also be disjoint in the set-theoretical sense. In particular, in the example
above, \(\Gamma, \Delta\) are disjoint, \(A \notin \Gamma\), and so on\ldots

\begin{remark}
    While named sequents are defined as sets, because of the presence of names this is
    by no means equivalent to the standard \emph{sequents-as-sets} approach: for example,
    the sequent
    \[
        \th \nm{x}{\alpha}\! \land \nm{y}{\beta}, \nm{z}{\alpha}\! \land \nm{w}{\beta}
    \]
    contains two occurrences of the same anonymous formula, distinguished by their names.
    In other words, stripping named sequents of all names results in multisets of formulas.
\end{remark}

We now redefine the usual \(\data{GS4}\) sequent system using named sequents in place of
traditional ones. As announced before, we also enrich the system with \emph{deterministic
axioms} and a \emph{superposition rule}.

\begin{definition}
    \label{defn:named-derivations}
    \emph{Named \(\data{GS4}\) derivations} are finite trees whose nodes (also called \emph{rule
    applications}) are labeled by an inference rule and by a \emph{sharing-free} named sequent (called
    \emph{conclusion} of the rule), inductively constructed in accordance with the following rules:
    \begin{itemize}%[--]
        \item identity rules:
        \[
            \prfbyaxiom{\rl[\{A,\dl{B}\}]{ax}}{\th \Gamma, A, \dl{B}}
            \qquad
            \prftree[r]{\rl{cut}}
                {\prfassumption{\th \Gamma, A}}
                {\prfassumption{\th \Gamma, \dl{A}}}
                {\th \Gamma}
        \]
        where \(A \equiv B\);

        \item superposition rule:
        \[
            \prftree[r]{\rl{\sqcup}}
                {\prfassumption{\th \Gamma}}
                {\prfassumption{\th \Gamma}}
                {\th \Gamma}
        \]

        \item logical rules:
        \[
            \prftree[r]{\rl{\lor}}
                {\prfassumption{\th \Gamma, A, B}}
                {\th \Gamma, A \lor B}
            \qquad
            \prftree[r]{\rl{\land}}
                {\prfassumption{\th \Gamma, A}}
                {\prfassumption{\th \Gamma, B}}
                {\th \Gamma, A \land B}
        \]
    \end{itemize}
    Let \(\data{GS4}^\mathcal{N}\) denote the set of all named derivations; we use letters
    \(P, Q, R, \ldots\) to range over \(\data{GS4}^\mathcal{N}\). We let~\(\data{names}(P)\)
    denote the set of all names occurring in some formula in~\(P\).
\end{definition}

\begin{remark}
    The peculiar form of axiom rules is due to the fact that the pairs of formulas they relate
    must share no names: because named formulas \(A\) and~\(\dl{A}\) share the same
    names, the conclusion of an hypothetical derivation tree of the form
    \[
        \prfbyaxiom{\rl[\{A, \dl{A}\}]{ax}}{\th \Gamma, A, \dl{A}}
    \]
    would not be sharing-free, as required by \cref{defn:named-derivations}.
\end{remark}

Correctness, completeness and cut admissibility results apply as usual: by forgetting
names one obtains a sound proof in the negative presentation of sequent calculus as
a tree of multisets of formulas, while completeness and cut admissibility can be obtained
either by adding names to anonymous proofs, or more easily by a straightforward adaptation
of Schutte's completeness proof (see~\cite{Sch77,AT14}).

\section{Inversion and isolation of logical rules}\label{sec:inversion-lemmas}

The foremost property of \(\data{GS4}^\mathcal{N}\) derivations is the invertibility of
all logical inference rules. As is well-known, the reason why they are called \emph{invertible}
is not just that their conclusions logically entail their premises; in fact, a stronger
property holds: each derivation~\(P\) of \(\th \Gamma, A\) with \(A\) non-atomic may be
rewritten in such a way as to recover derivations in the context~\(\Gamma\) of each premiss
of the logical rule introducing~\(A\). We call this rewriting step \emph{inversion}.
One may then apply the logical rule again to obtain a derivation~\(P'\) of \(\th \Gamma, A\)
where the last rule application introduces~\(A\): we say that \(P'\) has been obtained
by \emph{isolating}~\(A\) in~\(P\). We state below the relevant facts and introduce some
notation; the rewriting procedures are described in detail in \cref{sec:inversion-proofs}.

\begin{lemma}
    \label{lemma:disjunction-transform}
    For each derivation \(P \in \data{GS4}^\mathcal{N}\) with conclusion~\(\th \Gamma, A \lor B\),
    there is a derivation \(\data{inv}(P, A \lor B) \in \data{GS4}^\mathcal{N}\) with
    conclusion~\(\th \Gamma, A, B\), cut-free if so is \(P\).
\end{lemma}

\begin{corollary}
    \label{coro:disjunction-isolation}
    For each derivation \(P \in \data{GS4}^\mathcal{N}\) with conclusion~\(\th \Gamma, A \lor B\),
    there is a derivation \(\data{isl}(P, A \lor B) \in \data{GS4}^\mathcal{N}\) with the same
    conclusion, cut-free if so is \(P\), and whose last rule application introduces~\(A \lor B\).
\end{corollary}
\begin{proof}
    Let \(\data{isl}(P, A \lor B)\) be the derivation
    \[
        \prftree[r]{\rl{\lor}}
            {\prfsummary[\(\data{inv}(P, A \lor B)\)]{\th \Gamma, A, B}}
            {\th \Gamma, A \lor B}
        \qedhere
    \]
\end{proof}

\begin{lemma}
    \label{lemma:left-conjunction-transform}
    For each derivation \(P \in \data{GS4}^\mathcal{N}\) with conclusion~\(\th \Gamma,
    A \land B\), there is a derivation \(\data{inv_l}(P, A \land B) \in \data{GS4}^\mathcal{N}\)
    with conclusion~\(\th \Gamma, A\), cut-free if so is \(P\).
\end{lemma}

\begin{lemma}
    \label{lemma:right-conjunction-transform}
    For each derivation \(P \in \data{GS4}^\mathcal{N}\) with conclusion~\(\th \Gamma,
    A \land B\), there is a derivation \(\data{inv_r}(P, A \land B) \in \data{GS4}^\mathcal{N}\)
    with conclusion~\(\th \Gamma, B\), cut-free if so is \(P\).
\end{lemma}

\begin{corollary}
    \label{coro:conjunction-isolation}
    For each derivation \(P \in \data{GS4}^\mathcal{N}\) with conclusion~\(\th \Gamma,
    A \land B\), there is a derivation \(\data{isl}(P, A \land B) \in \data{GS4}^\mathcal{N}\)
    with the same conclusion, cut-free if so is \(P\), and whose last rule application
    introduces~\(A \land B\).
\end{corollary}
\begin{proof}
    Let~\(\data{isl}(P, A \land B)\) be the derivation
    \[
        \prftree[r]{\rl{\land}}
            {\prfsummary[\(\data{inv_l}(P, A \land B)\)]{\th \Gamma, A}}
            {\prfsummary[\(\data{inv_r}(P, A \land B)\)]{\th \Gamma, B}}
            {\th \Gamma, A \land B}
        \qedhere
    \]
\end{proof}

\subsection{Admissibility of contraction and weakening}\label{sec:ctr-wkn}

\begin{lemma}
    \label{lemma:wk-transform}
    For each derivation \(P \in \data{GS4}^\mathcal{N}\) with conclusion~\(\th \Gamma\)
    and \emph{sharing-free} set~\(\Delta\) of named formulas such that \(\Gamma, \Delta\)
    share no names, there is a derivation \(\data{wk}(P, \Delta) \in \data{GS4}^\mathcal{N}\)
    with conclusion~\(\th \Gamma, \Delta\), cut-free if so is \(P\).
\end{lemma}

The contraction rule is derivable through the use of cuts, and therefore admissible:
\[
    \vcenter{\prftree[r]{\rl{cut}}
        {\prfsummary[\(P\)]{\th \Gamma, A, B}}
        {\prfbyaxiom{\rl[\{A,\dl{B}\}]{ax}}{\th \Gamma, A, \dl{B}}}
        {\th \Gamma, A}}
    \qquad
    (A \equiv B)
\]
On the other hand, the weakening rule is not derivable because of the way the cut rule is
formulated in \(\data{GS4}\) (i.e.\ with context sharing).

\section{Axiom graphs}\label{sec:axiom-graphs}

We are now ready to formalize the idea of axiom-induced graphs. We shall rely on a standard
notion of simple graph: the related notation and properties are summarized in \cref{sec:graphs}.
We start from the notion of \emph{name graph}, i.e.\ a simple graph whose vertex set is a subset
of the set~\(\mathcal{N}\) of names. We associate a name graph to each derivation; following
the intuitions given in the introduction, only axioms and cuts will receive a special treatment,
while all other rules shall be interpreted trivially by graph unions.

\begin{definition}[Weakening and identity graphs]
    \label{defn:wk-id-graphs}
    For any sharing-free set \(\Gamma\) of named formulas, let \(\data{Wk}_\Gamma\)
    denote the graph
    \[
        \data{Wk}_\Gamma = \langle \data{names}(\Gamma), \emptyset \rangle
    \]
    whose vertices are the names occurring in~\(\Gamma\) and whose edge set is empty.
    We call~\(\data{Wk}_\Gamma\) the \emph{edgeless} or \emph{weakening graph on~\(\Gamma\)}.

    For any pair \(A, \dl{B}\) of disjoint and sharing-free named formulas such that
    \(A \equiv B\), we define by induction on the height\footnote{Which is necessarily
    identical to the height of~\(B\), see \cref{defn:formula-complexity} from \cref{sec:complexity}.}
    of~\(A\) the graph \(\data{Id}_{\{A, \dl{B}\}}\):
    \begin{gather*}
        \data{Id}_{\{\nm{x}{\alpha}, \nm{y}{\dl\alpha}\}} = \langle \{x,y\}, \{xy\} \rangle;\\[.5em]
        \data{Id}_{\{A_1 \lor A_2, \dl{B}_1 \land \dl{B}_2\}} =
            \data{Id}_{\{A_1 \land A_2, \dl{B}_1 \lor \dl{B}_2\}} =
                \data{Id}_{\{A_1, \dl{B}_1\}} \sqcup \data{Id}_{\{A_2, \dl{B}_2\}}.
    \end{gather*}
    We call~\(\data{Id}_{\{A, \dl{B}\}}\) the \emph{identity graph on~\(A, \dl{B}\)}.
    It is easy to check by induction on~\(A\) that \(V_{\data{Id}_{\{A, \dl{B}\}}} =
    \data{names}(\{A, \dl{B}\})\).
\end{definition}

In order to interpret cuts, we need to implement the informal idea of paths alternating between
the two cut subderivations through the cut-formula. The implementation is independent of the
logical framework and can be specified directly in terms of name graphs:

\begin{definition}[Alternating paths]
    \label{defn:simple-alternating-path}
    Let \(G, H\) be arbitrary name graphs, \(I \subseteq \mathcal{N}\) a set of names
    called \emph{interface}. Furthermore, let \(x_1, \ldots, x_n \in V_G \cup V_H\) be
    a sequence of \emph{pairwise distinct} vertices from \(G, H\) (for some \(n > 0\)),
    and let \(e_i\) denote the unordered pair \(x_i x_{i+1}\) for all \(1 \leq i < n\).

    \(x_1, \ldots, x_n\) is an \emph{alternating path} between \(G\) and~\(H\) \emph{through
    the interface~\(I\)} if and only if
    \begin{enumerate}[(i)]
        \item \(x_i \in I\) for all~\(1 < i < n\) (those we call \emph{internal vertices} of
        the path);

        \item either \(e_i \in E_G\) for all odd~\(1 \leq i < n\) and~\(e_i \in E_H\) for
        all even~\(1 \leq i < n\), or \(e_i \in E_H\) for all odd~\(1 \leq i < n\) and~\(e_i
        \in E_G\) for all even~\(1 \leq i < n\).
    \end{enumerate}
\end{definition}

In other words, alternating paths are built by choosing ‘adjacent’ edges alternately in~\(G\)
and~\(H\) or \emph{vice versa}, under the restriction that all vertices except the first and
the last be in the interface. Observe that cycles are ignored -- since vertices in the
path are required to be pairwise distinct -- hence when \(G\) and~\(H\) are finite, there
are for any given interface only finitely many alternating paths between them.

\begin{definition}[Composition of name graphs]
    \label{defn:name-graph-composition}
    Given name graphs \(G, H\) and a set of names~\(I \subseteq \mathcal{N}\), we define
    the \emph{composite of~\(G\) and~\(H\) on interface~\(I\)} as the graph
    \[
        G \odot_I H = \langle V, E \rangle
    \]
    where
    \[
        V = (V_G \cup V_H) \setminus I
    \]
    and for all \(x \neq y \in V\), \(xy \in E\) if and only if there is an alternating path
    \(z_1, \ldots, z_n\) between \(G\) and~\(H\) through the interface~\(I\) such that
    \(z_1 = x\) and~\(z_n = y\).
\end{definition}

From now on, given any name graphs~\(G, H\) and a named formula \(A\), we may write
\(G \odot_A H\) as a shorthand for \(G \odot_{\data{names}(A)} H\).

\begin{definition}[Axiom graphs]
    \label{defn:axiom-graphs}
    To each derivation tree~\(P \in \data{GS4}^\mathcal{N}\) we associate a named graph~\(\axg{P}\),
    called \emph{axiom graph of~\(P\)}, defined by structural induction on~\(P\):
    \begin{itemize}%[--]
        \item if \(P\) has the form
        \[
            \prfbyaxiom{\rl[\{A, \dl{B}\}]{ax}}
                {\th \Gamma, A, \dl{B}}
        \]
        where \(A \equiv B\), then let \(\axg{P} = \data{Wk}_\Gamma \sqcup \data{Id}_{\{A,\dl{B}\}}\);

        \item if \(P\) has the form
        \[
            \prftree[r]{\rl{cut}}
                {\prfsummary[\(Q\)]{\th \Gamma, A \vphantom{\dl{A}}}}
                {\prfsummary[\(R\)]{\th \Gamma, \dl{A}}}
                {\th \Gamma}
        \]
        then let \(\axg{P} = \axg{Q} \odot_A \axg{R}\);

        \item if \(P\) has the form
        \[
            \prftree[r]{\(r\)}
                {\prfsummary[\(Q_1\)]{\th \Gamma_1}}
                {\prfassumption{\cdots}}
                {\prfsummary[\(Q_n\)]{\th \Gamma_n}}
                {\th \Gamma}
        \]
        where \(r \in \{\rl{\sqcup}, \rl{\lor}, \rl{\land}\}\), then let
        \(\axg{P} = \axg{Q_1} \sqcup \ldots \sqcup \axg{Q_n}\).
    \end{itemize}
\end{definition}

\begin{proposition}
    \label{propo:axiom-graphs-conclusion}
    For all derivations~\(P \in \data{GS4}^\mathcal{N}\) with conclusion~\(\th \Gamma\),
    \[
        V_{\axg{P}} = \data{names}(\Gamma).
    \]
\end{proposition}

\subsection{Axiom graphs do not increase under inversion and isolation}

\begin{theorem}
    \label{thm:axiom-graphs-decrease}
    Let~\(P \in \data{GS4}^\mathcal{N}\) be any named derivation tree with conclusion~\(\th \Gamma, A\),
    where \(A\) is any \emph{non-atomic} named formula: then
    \[
        \axg{\data{isl}(P, A)} \sqsubseteq \axg{P}.
    \]
\end{theorem}

From now on, and especially in the statements and proofs of the following lemmas, we shall
write expressions of the form \(\axg{P} \rst_{\Gamma, A}\) as a shorthand for \(\axg{P}
\rst_{\data{names}(\Gamma, A)}\) (see \cref{defn:subgraphs} for the subgraph relation and
the graph restriction notation). All omitted proofs can be found in \cref{sec:axiom-graph-proofs}.

\begin{lemma}
    \label{lemma:axiom-graph-disjunction-transform}
    For all derivations~\(P \in \data{GS4}^\mathcal{N}\) with conclusion~\(\th \Gamma, A \lor B\),
    \[
        \axg{\data{inv}(P, A \lor B)} = \axg{P}.
    \]
\end{lemma}

\begin{lemma}
    \label{lemma:axiom-graph-conjunction-transforms}
    For all derivations~\(P \in \data{GS4}^\mathcal{N}\) with conclusion~\(\th \Gamma, A \land B\),
    \[
        \axg{\data{inv_l}(P, A \land B)} \sqsubseteq \axg{P} \rst_{\Gamma, A} \text{ and }
        \axg{\data{inv_r}(P, A \land B)} \sqsubseteq \axg{P} \rst_{\Gamma, B};
    \]
\end{lemma}

\begin{proof}[Proof of \cref{thm:axiom-graphs-decrease}]
    Immediate consequence of \cref{lemma:axiom-graph-disjunction-transform,%
    lemma:axiom-graph-conjunction-transforms} together with the definition of
    the \(\data{isl}\)-transformation (\cref{coro:disjunction-isolation,%
    coro:conjunction-isolation}). The result may also be seen as a very specific
    instance of a more general one by Führmann and Pym~\cite{FP04}.
\end{proof}

\subsection{Axiom graphs are not invariants of isolation and normalisation}\label{sec:axiom-graphs-failure}

The inequality from \cref{thm:axiom-graphs-decrease} can be upgraded to an equality in the
case of cut-free proofs:

\begin{proposition}
    \label{propo:axiom-graphs-isl-cut-free-invariance}
    For all cut-free derivations~\(P \in \data{GS4}^\mathcal{N}\) of the sequent~\(\th \Gamma, A\)
    with \(A\) \emph{non-atomic},
    \[
        \axg{\data{isl}(P, A)} = \axg{P}.
    \]
\end{proposition}

On the other hand, the invariance result does not hold in the case of proofs with cuts,
as there are derivations whose axiom graph decreases strictly under isolation of some conjunction
in their conclusions. \Cref{fig:axiom-graph-inequality-example1} shows one such counterexample,
where there is an alternating path (marked in blue in~\cref{fig:axiom-graph-inequality-example1:pre})
that is lost as soon as the conjunction in the conclusion is isolated, as shown in~%
\cref{fig:axiom-graph-inequality-example1:post}.

The unstable path in \cref{fig:axiom-graph-inequality-example1} is precisely the one that
connects names occurring on the two distinct sides of a conjunction: in fact, it is possible
to prove that no such edge may exist in a cut-free derivation, hence the axiom graph construction
cannot be invariant under cut-elimination.

One might be tempted to modify the construction so that conjunction-crossing paths be erased
as soon as possible; this, however, is not enough to guarantee invariance under isolation.
\Cref{fig:axiom-graph-inequality-example2} shows a more complex example where an alternating
path is lost, whose endpoints are not separated by a conjunction. The problem in this case
is that the disappearing path contains edges coming from subderivations of both conjuncts
of a conjunction that appears in the conclusion. As soon as the conjunction is isolated
(\cref{fig:axiom-graph-inequality-example2:post}), it becomes impossible to construct such
a path.

\section{Branch-labeled axiom graphs}\label{sec:bl-axiom-graphs}

Our solution is to refine the axiom graph construction by attaching labels to each edge,
whose purpose is to track the branches each edge came from in the interpreted derivation.
We do this by replacing the edge set with a relation associating sets of names (representing
branches of derivations) to unordered pairs of vertices (representing unoriented edges).
Every edge must come from some branch: if some pair of vertices has no associated name set,
then we consider there to be no edge between the two vertices. We then use the additional
information to discard alternating paths composed by edges belonging to incompatible branches.

\subsection{Naming branches}

It is a well-known fact \cite{NP01,PP20} that when the axiom rule is restricted to atomic
conclusions, the set of axiom rule conclusions (also called \emph{top-sequents}) of any given
\(\data{GS4}\) derivation is uniquely determined by the conclusion of the derivation.
For obvious reasons, every branch is terminated by a unique axiom-rule application and can
then be named by its conclusion -- in fact it is sufficient to consider the set of names that
occur in the leaf's conclusion. In the presence of superposition rules the naming will not be
unique, but this is not a problem as identically named branches can be collapsed into one
by reducing superpositions to atomic form.

In order to ensure that branch names be stable under isolation, we need to take into account
the possible expansions of non-atomic axioms: branches terminated by such rule applications
correspond in fact to multiple “virtual” atomic branches, still uniquely determined by the
leaf's conclusion. We need then to define the unique set~\(\data{Br}(\Gamma)\) of atomic
branch names determined by a given sequent~\(\th \Gamma\). One possible approach -- followed
e.g.\ in~\cite{PP20} -- is to rely upon \(\data{GS4}\) inference rules to obtain the unique
atomic decomposition of the sequent; we prefer however to provide a direct characterisation
of the set~\(\data{Br}(\Gamma)\), then show that it is compatible with the inference rules.

\begin{definition}[Formula and sequent branches]
    \label{defn:formula-branches}
    We associate inductively a set of branch names to each sharing-free named formula~\(A\)
    as follows:
    \begin{itemize}
        \item \(\data{Br}(\nm{x}{\alpha}) = \{\{x\}\}\);
        \item \(\data{Br}(B \lor C) = \{ X \cup Y \mid X \in \data{Br}(B), Y \in \data{Br}(C) \}\);
        \item \(\data{Br}(B \land C) = \data{Br}(B) \cup \data{Br}(C)\);
    \end{itemize}
    then let, for all sharing-free sets \(\Gamma\) of named formulas,
    \[
        \data{Br}(\Gamma) = \{ X \subseteq \data{names}(\Gamma) \mid \forall A \in \Gamma.\,(X \cap \data{names}(A)) \in \data{Br}(A) \}.
    \]
\end{definition}

The construction for sets is meant to treat them as generalized disjunctions over their elements.
While syntactic binary disjunction distinguishes between a left and a right subformula, elements
of a set have no preferred ordering, hence the need for a slightly less straightforward definition.
We provide for clarity an alternative characterisation of~\(\data{Br}(\Gamma)\) (as usual
we provide detailed proofs in \cref{sec:bl-axiom-graph-proofs}):

\begin{lemma}
    \label{lemma:sequent-branches-alt}
    Let \(\Gamma\) be any sharing-free set of named formulas. For any branch name \(X
    \subseteq \mathcal{N}\), \(X \in \data{Br}(\Gamma)\) if and only if there is a family
    \((X_A)_{A \in \Gamma}\) of branch names such that \(X = \bigcup_{A \in \Gamma} X_A\),
    with \(X_A \in \data{Br}(A)\) for all~\(A \in \Gamma\).
\end{lemma}

\begin{lemma}
    \label{lemma:branches-of-union}
    Let \(\Gamma, \Delta\) be sharing free sets of named formulas that share no names:
    then
    \[
        \data{Br}(\Gamma \cup \Delta) = \{ X \cup Y \mid X \in \data{Br}(\Gamma), Y \in \data{Br}(\Delta) \}.
    \]
\end{lemma}

\begin{corollary}
    \label{coro:branches-of-difference}
    \(\data{Br}(\Gamma) = \{ X \setminus \data{names}(\Delta) \mid X \in \data{Br}(\Gamma \cup \Delta) \}\).
\end{corollary}

It is an easy consequence of \cref{lemma:sequent-branches-alt} that \(\data{Br}(\{A\}) =
\data{Br}(A)\). Therefore, from this point on we are going to abuse systematically the usual
sequent notation and write, e.g., \(\data{Br}(\Gamma, A, \Delta)\) for \(\data{Br}(\Gamma \cup
\{A\} \cup \Delta)\).

\begin{proposition}
    \label{propo:sequent-branches-decomposition}
    Let \(\Gamma\) be any sharing-free set of named formulas, \(A, B\) disjoint and sharing-free
    named formulas that share no name with~\(\Gamma\):
    \begin{enumerate}[(i)]
        \item if all formulas in~\(\Gamma\) are atomic, then \(\data{Br}(\Gamma) = \{\data{names}(\Gamma)\}\);
        \item \(\data{Br}(\Gamma, A \lor B) = \data{Br}(\Gamma, A, B)\);
        \item \(\data{Br}(\Gamma, A \land B) = \data{Br}(\Gamma, A) \cup \data{Br}(\Gamma, B)\);
        \item \(\data{Br}(\Gamma, A)\) and \(\data{Br}(\Gamma, B)\) are disjoint.
    \end{enumerate}
\end{proposition}

\subsection{Branch-labeled name graphs}

\begin{definition}
    \label{defn:bl-graphs}
    A branch-labeled name graph is a pair \(G = \langle V_G, {\adin_G} \rangle\) where
    \(V_G \subseteq \mathcal{N}\) is a set of names and
    \[
        {\adin_G} \subseteq \binom{V_G}{2} \times \pwr(V_G)
    \]
    is a binary relation between unordered pairs of vertices and arbitrary sets of vertices
    of~\(G\), such that
    \begin{equation}
        e \adin_G X \implies e \subseteq X. \tag{\(\star\)}\label{eqn:edge-branch-inclusion}
    \end{equation}
\end{definition}

For~\(G\) any branch-labeled name graph (hereinafter \emph{bl-graph} for brevity), we can
define a set
\[
    E_G = \pi_{\data{l}}({\adin_G}) = \{ e \mid \exists X.\ e \adin_G X \}
\]
of unordered pairs of vertices which we shall call the \emph{edges} of~\(G\). We define
similarly the set of \emph{branches of~\(G\)}:
\[
    \data{Br}(G) = \pi_{\data{r}}({\adin_G}) = \{ X \mid \exists e.\ e \adin_G X \}.
\]
We read the predicate \(xy \adin_G X\) as \emph{\(x, y\) are adjacent in branch~\(X\)},
or \emph{branch~\(X\) has the edge~\(xy\)}. Condition~(\ref{eqn:edge-branch-inclusion})
ensures that edges only connect vertices belonging to the branch they originated in.

We extend the subgraph relation and the union operator to bl-graphs simply by replacing the
edge set with the edge-branch relation in the definition, i.e. let
\[
    G \sqsubseteq H \iff V_G \subseteq V_H \text{ and } {\adin_G} \subseteq {\adin_H}
\]
for all pairs \(G, H\) of bl-graphs, and
\[
    \bigsqcup_{i \in I} G_i = \langle \bigcup_{i \in I} V_{G_i}, \bigcup_{i \in I} {\adin_{G_i}} \rangle
\]
where \(I\) is any index set and~\((G_i)_{i \in I}\) an indexed family of bl-graphs. The
restriction operator must be modified slightly so as to select branches instead of edges:
for~\(G\) any bl-graph and~\(X \subseteq \mathcal{N}\), let
\[
    G \rst_X = \langle V_G \cap X, \{ (e, Y) \in {\adin_G} \mid Y \subseteq X \} \rangle.
\]
Note that while union acts upon edges in the usual way (we have \(E_{G \sqcup H} = E_G
\cup E_H\)), restriction might remove more edges than in the case of simple graphs:
for all~\(e \in E_G\), \(e \in E_{G \rst_X}\) implies \(e \subseteq X\), but the converse
does not hold in general.

\subsection{Branch-sensitive composition}

We come thus to the crux of the refined approach -- the composition of bl-graphs over
some interface. Remember that composition is meant to interpret cuts between derivations
\(P, Q\) of conclusion \(\th \Gamma, A\) and~\(\th \Gamma, \dl{A}\) respectively, with the
cut rule having conclusion \(\th \Gamma\). We start then with atomic branch sets \(\data{Br}(
\Gamma, A)\), \(\data{Br}(\Gamma, \dl{A})\), while the branch set of the final cut-free
derivation will be~\(\data{Br}(\Gamma)\). When constructing alternating paths, we are going
to use edges which come in general from different branches in~\(P, Q\), but we need to ensure
that they all belong to the same branch of the cut-free derivation, otherwise they might
disappear under isolation.

To this end, observe that by \cref{lemma:branches-of-union} all branches \(X \in \data{Br}(
\Gamma, A)\) (resp.~\(\data{Br}(\Gamma, \dl{A})\)) are of the form \(Y \cup Z\) where \(Y \in
\data{Br}(\Gamma)\) and~\(Z \in \data{Br}(A)\) (resp.~\(\data{Br}(\dl{A})\)). Therefore,
we have \(X \setminus \data{names}(A) = X \setminus \data{names}(\dl{A}) \in \data{Br}(\Gamma)\)
(\cref{coro:branches-of-difference}). The idea is then to check that all edges forming an
alternating path share the same branch label up to names in the composition interface.

For any bl-graph~\(G\) and set~\(I \subseteq \mathcal{N}\) of names, we define an
\emph{interface-relativized} edge-branch relation
\[
    {\adin^I_G} = \{ (e, X \setminus I) \mid e \adin_G X \},
\]
or equivalently
\[
    e \adin^I_G X \iff \exists Y.\,e \adin_G Y \text{ and } X = Y \setminus I.
\]

\begin{definition}[Alternating labeled paths]
    \label{defn:bl-alternating-path}
    Let \(G, H\) be arbitrary bl-graphs, \(I, X \subseteq \mathcal{N}\) sets of names,
    \(x_1, \ldots, x_n \in V_G \cup V_H\) a sequence of pairwise distinct vertices
    from~\(G, H\) (with \(n > 1\)), and let \(e_i\) denote the unordered pair~\(x_i x_{i+1}\)
    for all~\(1 \leq i < n\).

    \(x_1, \ldots, x_n\) is an \emph{alternating \(X\)-labeled path} between \(G\) and~\(H\)
    through the interface~\(I\) if and only if
    \begin{enumerate}[(i)]
        \item \(x_i \in I\) for all~\(1 < i < n\) (all internal vertices belong to the interface);
        \item either \(e_i \adin^I_G X\) for all odd~\(1 \leq i < n\) and \(e_i \adin^I_H X\)
        for all even~\(1 \leq i < n\), or \(e_i \adin^I_H X\) for all odd~\(1 \leq i < n\) and
        \(e_i \adin^I_G X\) for all even~\(1 \leq i < n\).
    \end{enumerate}
    We call the path \emph{complete} iff \(x_1, x_n \notin I\).
\end{definition}

\begin{lemma}
    \label{lemma:complete-alternating-path-in-branch}
    If \(z_1, \ldots, z_n\) is a complete \(X\)-labeled alternating path between
    bl-graphs~\(G, H\) through interface~\(I\), then \(z_1, z_n \in X\).
\end{lemma}
\begin{proof}
    By \cref{defn:bl-alternating-path} \(n > 1\), hence there is~\(K \in \{G, H\}\)
    and \(Y \in \data{Br}(K)\) such that \(z_1 z_2 \adin_K Y\) with \(X = Y \setminus I\).
    By \cref{defn:bl-graphs} \(z_1 z_2 \subseteq Y\), and because \(z_1 \notin I\) by
    hypothesis, we must have \(z_1 \in X\). Similar reasoning shows that \(z_n \in X\).
\end{proof}

\begin{definition}[Composition of bl-graphs]
    \label{defn:bl-composition}
    Let \(G, H\) be arbitrary bl-graphs, \(I \subseteq \mathcal{N}\) a set of names.
    We define the \emph{composite of~\(G\) and~\(H\) on interface~\(I\)} as the bl-graph
    \[
        G \odot_I H = \langle V, {\adin} \rangle
    \]
    where
    \[
        V = (V_G \cup V_H) \setminus I
    \]
    and for all~\(x \neq y \in V\) and~\(X \subseteq V\), \(xy \adin X\) if and only if there
    is a complete alternating \(X\)-labeled path~\(z_1, \ldots, z_n\) between \(G\) and~\(H\)
    through the interface~\(I\), such that \(z_1 = x\) and~\(z_n = y\).
    \Cref{lemma:complete-alternating-path-in-branch} guarantees that \(xy \subseteq X\).
\end{definition}

\subsection{Interpreting derivations}

Finally, we define the new inductive interpretation function for derivations. Special attention
must be paid to the weakening case: it is not possible to handle it by a simple graph union,
like in the original axiom graph construction, as we need to update all branch labels to take
weakened formulas into account. Weakenings are then interpreted by an operator on bl-graphs.
Identities also need to be tweaked to account for the way non-atomic axioms are expanded by the
inversion procedures.

\begin{definition}[Weakening and identities]
    \label{defn:bl-wk-id}
    For any bl-graph~\(G\) and sharing-free set~\(\Gamma\) of named formulas, let
    \[
        \blwk{\Gamma}(G) = \langle V_G \cup \data{names}(\Gamma), \{ (e, X \cup Y) \mid e \adin_G X,\,Y \in \data{Br}(\Gamma) \} \rangle.
    \]
    For any pair \(A, \dl{B}\) of disjoint and sharing-free named formulas such that
    \(A \equiv B\), we define by induction on the height of~\(A\) the bl-graph~\(\blid{\{A, \dl{B}\}}\):
    \begin{gather*}
        \blid{\{\nm{x}{\alpha}, \nm{y}{\dl{\alpha}}\}} = \langle \{x, y\}, \{ (xy, \{x, y\}) \} \rangle;\\
        \blid{\{A_1 \lor A_2, \dl{B}_1 \land \dl{B}_2\}}
            = \blwk{A_2}(\blid{\{A_1, \dl{B}_1\}}) \sqcup \blwk{A_1}(\blid{\{A_2, \dl{B}_2\}}).
    \end{gather*}
\end{definition}

\begin{definition}[Branch-labeled axiom graphs]
    \label{defn:bl-axiom-graphs}
    To each derivation tree~\(P \in \data{GS4}^\mathcal{N}\) we associate a bl-graph~\(\blaxg{P}\),
    called \emph{branch-labeled axiom graph of~\(P\)}, defined by structural induction on~\(P\):
    if \(P\) has the form
    \[
        \prfbyaxiom{\rl[\{A, \dl{B}\}]{ax}}
            {\th \Gamma, A, \dl{B}}
    \]
    where \(A \equiv B\), then let \(\blaxg{P} = \blwk{\Gamma}(\blid{\{A,\dl{B}\}})\);
    as for simple axiom graphs, cuts are interpreted by composition and all other rules
    by taking the bl-graph union over their subderivations.
\end{definition}

\begin{proposition}
    \label{propo:bl-axiom-graph-names-branches}
    For all derivations~\(P \in \data{GS4}^\mathcal{N}\) with conclusion~\(\th \Gamma\),
    \[
        V_{\blaxg{P}} = \data{names}(\Gamma) \text{ and } \data{Br}(\blaxg{P}) \subseteq \data{Br}(\Gamma).
    \]
\end{proposition}

\section{Main results}\label{sec:main-results}

\subsection{Behaviour under inversion and isolation}\label{sec:bl-behaviour}

\begin{theorem}
    \label{thm:bl-axiom-graphs-invariance}
    For all derivations~\(P \in \data{GS4}^\mathcal{N}\) with conclusion~\(\th \Gamma, A\),
    where \(A\) is any \emph{non-atomic} named formula,
    \[
        \blaxg{\data{isl}(P, A)} = \blaxg{P}.
    \]
\end{theorem}

The proof of \cref{thm:bl-axiom-graphs-invariance} rests upon the following three lemmas,
whose proofs are detailed in \cref{sec:bl-axiom-graph-proofs}.

\begin{lemma}
    \label{lemma:bl-axiom-graph-disjunction-transform}
    For all derivations~\(P \in \data{GS4}^\mathcal{N}\) with conclusion~\(\th \Gamma, A \lor B\),
    \[
        \blaxg{\data{inv}(P, A \lor B)} = \blaxg{P}.
    \]
\end{lemma}

\begin{lemma}
    \label{lemma:bl-axiom-graph-conjunction-transforms}
    For all derivations~\(P \in \data{GS4}^\mathcal{N}\) with conclusion~\(\th \Gamma, A \land B\),
    \[
        \blaxg{\data{inv_l}(P, A \land B)} = \blaxg{P} \rst_{\Gamma, A} \text{ and }
        \blaxg{\data{inv_r}(P, A \land B)} = \blaxg{P} \rst_{\Gamma, B};
    \]
\end{lemma}

\begin{lemma}
    \label{lemma:bl-axiom-graph-decomposition}
    For all derivations~\(P \in \data{GS4}^\mathcal{N}\) with conclusion~\(\th \Gamma, A \land B\),
    \[
        \blaxg{P} = \blaxg{P} \rst_{\Gamma, A} \cup \blaxg{P} \rst_{\Gamma, B}.
    \]
\end{lemma}

\begin{proof}[Proof of \cref{thm:bl-axiom-graphs-invariance}]
    Immediate by \cref{lemma:bl-axiom-graph-disjunction-transform} if \(A\) is a disjunction.
    If instead~\(A = B \land C\) is a conjunction, we have by construction\footnote{See
    the proof of \cref{coro:conjunction-isolation}.}
    \[
        \blaxg{\data{isl}(P, A)} = \blaxg{\data{inv_l}(P, A \land B)} \sqcup \blaxg{\data{inv_r}(P, A \land B)}
    \]
    and then by \cref{lemma:bl-axiom-graph-conjunction-transforms,lemma:bl-axiom-graph-decomposition}
    \[
        \blaxg{\data{inv_l}(P, A \land B)} \sqcup \blaxg{\data{inv_r}(P, A \land B)}
            = {\blaxg{P} \rst_{\Gamma, B}} \sqcup {\blaxg{P} \rst_{\Gamma, C}} = \blaxg{P}. \qedhere
    \]
\end{proof}

\subsection{Cut-elimination theorem}\label{sec:normalisation}

\begin{theorem}
    \label{thm:normalisation}
    For all derivations~\(P \in \data{GS4}^\mathcal{N}\) with conclusion~\(\th \Gamma\),
    there is a \emph{cut-free} derivation~\(Q \in \data{GS4}^\mathcal{N}\) with conclusion~%
    \(\th \Gamma\) and such that \(\blaxg{P} = \blaxg{Q}\).
\end{theorem}

Since we lack a cut-reduction procedure compatible with the interpretation, we have
to prove \cref{thm:normalisation} through a \emph{normalisation-by-evaluation} argument,
where we first compute the interpretation of the derivation, then reconstruct a cut-free
derivation with the same interpretation.

However, thanks to \cref{thm:bl-axiom-graphs-invariance}, we can employ the isolation
procedure to commute cuts up the derivation until they are reduced to atomic contexts.
We are thus able to limit the need for evaluation to the very specific and much simpler
case of \emph{quasi-cut-free} derivations\footnote{I.e., those formed by cutting together
two cut-free derivations} with atomic conclusion.

The proof then looks like a standard cut-elimination argument, with a lemma for the quasi-cut-free
case and a final general argument by induction on the height of the derivation. We start
with a kind of semantic cut-admissibility result. All omitted proofs are provided as usual
in \cref{sec:bl-axiom-graph-proofs}, except that of \cref{lemma:semantic-cut-admissibility}
to which we devote the whole of \cref{sec:semantic-cut-admissibility-proof} because of its
complexity.

\begin{lemma}
    \label{lemma:edges-join-dual-atoms}
    Let \(P \in \data{GS4}^\mathcal{N}\) be any derivation with conclusion~\(\th \Gamma\).
    All edges in~\(\blaxg{P}\) link dual atom occurrences, i.e.\ for all~\(xy \in E_{\blaxg{P}}\)
    we have \(\Gamma[x] = \dl{\Gamma[y]}\).
\end{lemma}

\begin{lemma}[Semantic cut-admissibility]
    \label{lemma:semantic-cut-admissibility}
    Let \(P, Q \in \data{GS4}^\mathcal{N}\) be cut-free derivations with conclusion~%
    \(\th \Gamma, A\) and~\(\th \Gamma, \dl{A}\) respectively, where all elements of
    the context~\(\Gamma\) are atomic formulas. Then the composite bl-graph~\(\blaxg{P}
    \odot_A \blaxg{Q}\) has at least one edge.
\end{lemma}

\begin{lemma}[Normalisation by evaluation]
    \label{lemma:normalisation-by-evaluation}
    Let \(P, Q \in \data{GS4}^\mathcal{N}\) be cut-free derivations with conclusion~%
    \(\th \Gamma, A\) and~\(\th \Gamma, \dl{A}\) respectively, where all elements of
    the context~\(\Gamma\) are atomic formulas. There is a cut-free derivation \(R \in
    \data{GS4}^\mathcal{N}\) with conclusion~\(\th \Gamma\) and such that \(\blaxg{R} =
    \blaxg{P} \odot_A \blaxg{Q}\).
\end{lemma}
\begin{proof}
    For brevity, let~\(G = \blaxg{P} \odot_A \blaxg{Q}\). Observe that the vertex set of~\(G\)
    is finite by \cref{propo:bl-axiom-graph-names-branches}, hence \(G\) has finitely many edges.
    Moreover, \(\data{Br}(G) \subseteq \data{Br}(\Gamma) = \{\data{names}(\Gamma)\}\)
    (\cref{propo:sequent-branches-decomposition}), i.e.\ all edges have the same branch
    label~\(\data{names}(\Gamma)\). Finally, by \cref{lemma:edges-join-dual-atoms} and
    the assumption that \(\Gamma\) is atomic, there are for all edges~\(xy \in E_G\)
    named formulas~\(\nm{x}{\alpha}, \nm{y}{\dl{\alpha}} \in \Gamma\) where \(\alpha =
    \Gamma[x]\) and \(\dl{\alpha} = \Gamma[y]\).

    Now let \(|E_G|\) denote the number of edges in~\(G\): by \cref{lemma:semantic-cut-admissibility}
    we know that \(|E_G| > 0\). We construct the derivation~\(R\) by induction on~\(n\):
    \begin{itemize}
        \item if \(|E_G| = 1\), then there is a unique edge~\(xy \in E_G\): let \(\alpha
        = \Gamma[x]\), and let \(R\) be the derivation
        \[
            \prfbyaxiom{\rl[\{\nm{x}{\alpha}, \nm{y}{\dl{\alpha}}\}]{ax}}{\th \Gamma}
        \]

        \item if \(|E_G| > 1\), then there is at least one edge~\(xy \in E_G\) and we can
        decompose \(G\) as~\(G' \sqcup G''\), where
        \[
            G' = \langle V_G, \{ (xy, \data{names}(\Gamma)) \} \rangle
            \text{ and }
            G'' = \langle V_G, {\adin_G} \setminus \{ (xy, \data{names}(\Gamma)) \} \rangle.
        \]
        We have obviously \(|E_{G'}|, |E_{G''}| < |E_G|\), hence there are by induction
        hypothesis derivations \(R', R''\) with conclusion~\(\th \Gamma\) and such that
        \(\blaxg{R'} = G'\), \(\blaxg{R''} = G''\): let then \(R\) be the derivation
        \[
            \prftree[r]{\rl{\sqcup}}
                {\prfsummary[\(R'\)]{\th \Gamma}}
                {\prfsummary[\(R''\)]{\th \Gamma}}
                {\th \Gamma}
        \]
    \end{itemize}
    Readers may check easily -- using the facts recalled above --  that \(\blaxg{R} = G\),
    as required.
\end{proof}

We come finally to the general cut-elimination proof, but first we must handle an important
technical detail. Because the isolation procedure might expand axioms, it does not preserve
in general the height of the derivation, hence we cannot proceed by induction on that
measure. We define instead an alternative measure called \emph{virtual height}, which
provides an upper bound to the height that may be attained by expanding the axioms, and
therefore does not increase under isolation:

\begin{definition}
    \label{defn:virtual-height}
    For any derivation~\(P \in \data{GS4}^\mathcal{N}\) with conclusion~\(\th \Gamma\),
    \begin{itemize}
        \item if \(P\) ends with an axiom rule application, let~\(\data{vh}(P) = 1 +
        \data{deg}(\Gamma)\) (\cref{defn:sequent-complexity});
        \item otherwise \(P\) has the form
        \[
            \prftree[r]{\(r\)}
                {\prfsummary[\(P_1\)]{\th \Gamma_1}}
                {\prfassumption{\cdots}}
                {\prfsummary[\(P_n\)]{\th \Gamma_n}}
                {\th \Gamma}
        \]
        where \(r \in \{\rl{cut}, \rl{\sqcup}, \rl{\lor}, \rl{\land}\}\): let then
        \(\data{vh}(P) = 1 + \max_{i = 1}^n \data{vh}(P_i)\).
    \end{itemize}
\end{definition}

\begin{lemma}
    \label{lemma:virtual-height-decrease}
    Let \(P \in \data{GS4}^\mathcal{N}\) be any derivation with conclusion~\(\th \Gamma, A\)
    where \(A\) is non-atomic; then \(\data{vh}(\data{isl}(P, A)) \leq \data{vh}(P)\).
\end{lemma}

\begin{proof}[Proof of \cref{thm:normalisation}]
    If \(P\) is cut-free, let~\(Q = P\). Otherwise, proceed by induction on the virtual
    height of~\(P\):
    \begin{itemize}
        \item if \(P\) has the form
        \[
            \prftree[r]{\(r\)}
                {\prfsummary[\(P_1\)]{\th \Gamma_1}}
                {\prfassumption{\cdots}}
                {\prfsummary[\(P_n\)]{\th \Gamma_n}}
                {\th \Gamma}
        \]
        where \(r \in \{\rl{\sqcup}, \rl{\lor}, \rl{\land}\}\), then apply the induction
        hypothesis to \(P_1, \ldots, P_n\) to get cut-free derivations \(Q_1, \ldots, Q_n\)
        with the same conclusion as~\(P_1, \ldots, P_n\) respectively, and let \(Q\) be the
        cut-free derivation
        \[
            \prftree[r]{\(r\)}
                {\prfsummary[\(Q_1\)]{\th \Gamma_1}}
                {\prfassumption{\cdots}}
                {\prfsummary[\(Q_n\)]{\th \Gamma_n}}
                {\th \Gamma}
        \]
        We have \(\blaxg{Q_i} = \blaxg{P_i}\) for all \(1 \leq i \leq n\) and therefore
        \(\blaxg{Q} = \bigsqcup_{i = 1}^n \blaxg{Q_i} = \bigsqcup_{i = 1}^n \blaxg{P_i} = \blaxg{P}\),
        as required;

        \item if \(P\) has the form
        \[
            \prftree[r]{\rl{cut}}
                {\prfsummary{\th \Gamma, B, A \vphantom{\dl{A}}}}
                {\prfsummary{\th \Gamma, B, \dl{A}}}
                {\th \Gamma, B}
        \]
        where \(B\) is non-atomic, there is a derivation \(P' = \data{isl}(P, B)\) of the form
        \[
            \prftree[r]{\(r\)}
                {\prfsummary[\(P'_1\)]{\th \Gamma_1}}
                {\prfassumption{\cdots}}
                {\prfsummary[\(P'_n\)]{\th \Gamma_n}}
                {\th \Gamma, B}
        \]
        where the last rule application \(r \in \{\rl{\lor}, \rl{\land}\}\) introduces
        the formula~\(B\), and such that \(\blaxg{P'} = \blaxg{P}\) (\cref{thm:bl-axiom-graphs-invariance})
        and~\(\data{vh}(P') \leq \data{vh}(P)\) (\cref{lemma:virtual-height-decrease}): we can
        then apply the induction hypothesis to the \(P'_i\) and proceed as in the previous
        case to get a cut-free derivation~\(Q\) such that \(\blaxg{Q} = \blaxg{P'} = \blaxg{P}\);

        \item finally, if \(P\) has the form
        \[
            \prftree[r]{\rl{cut}}
                {\prfsummary[\(P_1\)]{\th \Gamma, A \vphantom{\dl{A}}}}
                {\prfsummary[\(P_2\)]{\th \Gamma, \dl{A}}}
                {\th \Gamma}
        \]
        where all elements of~\(\Gamma\) are atomic formulas, we apply first the induction
        hypothesis to \(P_1, P_2\) to get cut-free derivations \(Q_1, Q_2\) with the same
        conclusions and such that \(\blaxg{Q_i} = \blaxg{P_i}\) for \(i \in \{1,2\}\).
        We can now apply \cref{lemma:normalisation-by-evaluation} to \(Q_1, Q_2\) to get
        a cut-free derivation~\(Q\) such that \(\blaxg{Q} = \blaxg{Q_1} \odot_A \blaxg{Q_2}
        = \blaxg{P_1} \odot_A \blaxg{P_2} = \blaxg{P}\). \qedhere
    \end{itemize}
\end{proof}

\section{Totality and canonical forms}\label{sec:totality}

The availability of a cut-elimination procedure that preserves branch-labeled axiom graphs
unlocks a powerful technique for proving properties of the interpretation: we can reason
about cut-free derivations and the result generalizes immediately to all derivations.
An important example is the following property, which together with the one described
in \cref{lemma:edges-join-dual-atoms} characterises the class of bl-graphs induced by
derivations in~\(\data{GS4}^\mathcal{N}\).

\begin{lemma}
    \label{lemma:bl-axiom-graphs-cutfree-branches}
    For all \emph{cut-free} derivations~\(P \in \data{GS4}^\mathcal{N}\) with conclusion~\(\th \Gamma\),
    \[
        \data{Br}(\blaxg{P}) = \data{Br}(\Gamma).
    \]
\end{lemma}

\begin{corollary}
    \label{coro:bl-axiom-graphs-branches}
    For all derivations~\(P \in \data{GS4}^\mathcal{N}\) with conclusion~\(\th \Gamma\),
    \[
        \data{Br}(\blaxg{P}) = \data{Br}(\Gamma).
    \]
\end{corollary}
\begin{proof}
    \Cref{thm:normalisation} guarantees the existence of a cut-free derivation~\(Q\) such
    that~\(\blaxg{Q} = \blaxg{P}\), to which we can apply \cref{lemma:bl-axiom-graphs-cutfree-branches}.
\end{proof}

\begin{definition}[Totality]
    \label{defn:totality}
    Call a bl-graph~\(G\) \emph{total} w.r.t.\ a sharing-free named sequent~\(\th \Gamma\)
    if and only if
    \begin{enumerate}[(i)]
        \item \(V_G = \data{names}(\Gamma)\) (\(G\) is a bl-graph on the names of~\(\Gamma\));
        \item \(\data{Br}(G) = \data{Br}(\Gamma)\) (the branches of~\(G\) are those of~\(\Gamma\));
        \item for all~\(xy \in E_G\), \(\Gamma[x] = \dl{\Gamma[y]}\) (the edges of~\(G\) link dual atoms).
    \end{enumerate}
\end{definition}

\begin{corollary}
    \label{coro:desequentialization}
    For all derivations~\(P \in \data{GS4}^\mathcal{N}\) with conclusion~\(\th \Gamma\),
    \(\blaxg{P}\) is total w.r.t.~\(\th \Gamma\).
\end{corollary}
\begin{proof}
    Immediate consequence of \cref{propo:bl-axiom-graph-names-branches,%
    coro:bl-axiom-graphs-branches,lemma:edges-join-dual-atoms}.
\end{proof}

\begin{theorem}
    \label{thm:sequentialization}
    Let \(G\) be any bl-graph that is total w.r.t.\ a sharing-free named sequent~\(\th \Gamma\):
    there is a derivation~\(P \in \data{GS4}^\mathcal{N}\) with conclusion~\(\th \Gamma\)
    and such that \(\blaxg{P} = G\).
\end{theorem}

\begin{lemma}
    \label{lemma:semantic-disjunction-inversion}
    Any bl-graph \(G\) total w.r.t.~\(\th \Gamma, A \lor B\) is also total w.r.t.~\(\th \Gamma, A, B\).
\end{lemma}

\begin{lemma}
    \label{lemma:semantic-conjunction-inversion}
    Any bl-graph \(G\) total w.r.t.~\(\th \Gamma, A \land B\) is decomposable as
    \[
        G = {G \rst_{\Gamma, A}} \sqcup {G \rst_{\Gamma, B}},
    \]
    with \(G \rst_{\Gamma, A}\) and~\(G \rst_{\Gamma, B}\) total w.r.t.~\(\th \Gamma, A\)
    and~\(\th \Gamma, B\), respectively.
\end{lemma}

\begin{proof}[Proof of \cref{thm:sequentialization}]
    By induction on the complexity degree of~\(\th \Gamma\) (\cref{defn:sequent-complexity},
    \cref{sec:complexity}). Observe first that \(\Gamma\) cannot be empty, because then \(G\)
    would be the empty graph and we would have \(\data{Br}(G) = \emptyset \neq \{\emptyset\}
    = \data{Br}(\Gamma)\), contradicting the totality of~\(G\) w.r.t.~\(\Gamma\); then:
    \begin{itemize}
        \item if \(\th \Gamma\) only contains atomic formulas, we can perform the construction
        described in the proof of \cref{lemma:normalisation-by-evaluation};

        \item if \(\th \Gamma = \th \Delta, A \lor B\), then \(G\) is total w.r.t.~%
        \(\th \Delta, A, B\) by \cref{lemma:semantic-disjunction-inversion}: we apply
        the induction hypothesis to get a derivation \(Q\) with conclusion~\(\th \Delta, A, B\)
        and such that~\(\blaxg{Q} = G\); we conclude by applying a disjunction rule to~\(Q\):
        \[
            \prftree[r,l]{\rl{\lor}}{\(P =\)}
                {\prfsummary[\(Q\)]{\th \Delta, A, B}}
                {\th \Delta, A \lor B}
        \]

        \item if \(\th \Gamma = \th \Delta, A \land B\), then \(G = {G \rst_{\Delta, A}}
        \sqcup {G \rst_{\Delta, B}}\) by \cref{lemma:semantic-conjunction-inversion},
        with \(G \rst_{\Delta, A}\) and \(G \rst_{\Delta, B}\) total w.r.t.~\(\th \Delta, A\)
        and~\(\th \Delta, B\) respectively: we apply the induction hypothesis twice to get
        derivations \(Q, R\) with conclusion~\(\th \Delta, A\) and~\(\th \Delta, B\)
        respectively, and such that \(\blaxg{Q} = G \rst_{\Delta, A}\), \(\blaxg{R} =
        G \rst_{\Delta, B}\); we conclude by applying a conjunction rule to~\(Q, R\):
        \[
            \prftree[r,l]{\rl{\land}}{\(P =\)}
                {\prfsummary[\(Q\)]{\th \Delta, A}}
                {\prfsummary[\(R\)]{\th \Delta, B}}
                {\th \Delta, A \land B}
            \qedhere
        \]
    \end{itemize}
\end{proof}

\begin{corollary}
    \label{coro:composition-preserves-totality}
    Let bl-graphs \(G, H\) be total w.r.t.\ the sequents~\(\th \Gamma, A\) and~\(\th \Gamma,
    \dl{A}\), respectively: then their composite~\(G \odot_A H\) on~\(A\) is total w.r.t.\ the
    sequent~\(\th \Gamma\).
\end{corollary}

\subsection{The proof system BLG}\label{sec:blg}

Upon inspection of the proof above, one sees clearly that this is a kind of sequentialization
theorem. The natural question then is whether total bl-graphs can provide a canonical
representation of cut-free derivations up to arbitrary permutations of logical rules.

\begin{samepage}
\begin{definition}
    \label{defn:blg}
    Let \(\data{BLG}\) denote the set of all pairs \(\langle G, \Gamma \rangle\) such that
    \begin{enumerate}[(i)]
        \item \(G\) is a finite bl-graph;
        \item \(\Gamma\) is a finite sharing-free set of named formulas;
        \item \(G\) is total w.r.t.~\(\th \Gamma\).
    \end{enumerate}
\end{definition}
\end{samepage}

We know by \cref{coro:desequentialization,thm:sequentialization} that there exists \(\langle G,
\Gamma \rangle \in \data{BLG}\) if and only if the sequent \(\th \Gamma\) is a classical
tautology. Is \(\data{BLG}\) a \emph{proof system} for classical propositional logic?
The question is subtle: as recalled in the introduction, the mere presence of a correctness
criterion is not sufficient to provide a reasonable notion of proof. We define a notion of
size of a bl-graph/sequent pair and show that totality is checkable in polynomial
time: \(\data{BLG}\) is therefore a proof system in the sense of Cook and Reckhow \cite{CR79}.

\begin{definition}
    \label{defn:blg-size}
    For any pair \(\mathbf{G} = \langle G, \Gamma \rangle\) where \(G\) is a finite
    bl-graph and~\(\Gamma\) a finite sharing-free set of named formulas, let
    \[
        \data{size}(\mathbf{G}) = \data{size}(\th \Gamma) + |V_G| + \sum_{e \adin_G X} |X|,
    \]
    where \(\data{size}(\th \Gamma)\) is the size of the sequent (\cref{defn:sequent-complexity},
    \cref{sec:complexity}), \(|V_G|\) is the number of vertices in the bl-graph~\(G\),
    and the last term is the total number of vertices in the branch labels of~\(G\).
\end{definition}

\begin{proposition}
    \label{propo:polynomial-time-totality}
    For any pair \(\mathbf{G} = \langle G, \Gamma \rangle\) where \(G\) is a finite
    bl-graph and~\(\Gamma\) a finite sharing-free set of named formulas, membership
    in~\(\data{BLG}\) is decidable in polynomial time in the size of~\(\mathbf{G}\).
\end{proposition}

\begin{figure}[t]
    \centering
    \makebox{\prftree
        {\blgbranch{\nm{x}{\alpha}, \nm{z}{\dl{\alpha}}}{x/z}}
        {\blgbranch{\nm{x}{\alpha}, \nm{w}{\dl{\alpha}}}{x/w}}
        {\blgbranch{\nm{y}{\alpha}, \nm{z}{\dl{\alpha}}}{y/z}}
        {\blgbranch{\nm{y}{\alpha}, \nm{w}{\dl{\alpha}}}{y/w}}
        {\th \nm{x}{\alpha} \land \nm{y}{\alpha}, \nm{z}{\dl{\alpha}} \land \nm{w}{\dl{\alpha}}}}
    \\[1em]
    \makebox{\prftree
        {\blgbranch{\nm{x}{\alpha}, \nm{z}{\beta}, \nm{v}{\dl{\alpha}}, \nm{w}{\gamma}}{x/v}}
        {\blgbranch{\nm{x}{\alpha}, \nm{u}{\dl{\gamma}}, \nm{v}{\dl{\alpha}}, \nm{w}{\gamma}}{x/v, u/w}}
        {\blgbranch{\nm{y}{\dl{\beta}}, \nm{z}{\beta}, \nm{v}{\dl{\alpha}}, \nm{w}{\gamma}}{y/z}}
        {\blgbranch{\nm{y}{\dl{\beta}}, \nm{u}{\dl{\gamma}}, \nm{v}{\dl{\alpha}}, \nm{w}{\gamma}}{u/w}}
        {\th \nm{x}{\alpha} \land \nm{y}{\dl{\beta}}, \nm{z}{\beta} \land \nm{u}{\dl{\gamma}}, \nm{v}{\dl{\alpha}} \lor \nm{w}{\gamma}}}
    \caption{A graphical representation of \(\data{BLG}\) proofs. Each proof is drawn as
        a generalised inference rule, with branch labels above the line and the conclusion
        below. Above each branch we draw the associated edges as black lines.}
    \label{fig:blg-graphical-representation}
\end{figure}

\subsection{Properties of the system BLG}

The proof system \(\data{BLG}\) enjoys very good properties. To begin with, all logical
rules of~\(\data{GS4}^\mathcal{N}\) are admissible and invertible:
\begin{gather*}
    \prftree[r]{\rl{{\downarrow}{\lor}}}
        {\prfassumption{\langle G, (\Gamma, A, B) \rangle \in \data{BLG}}}
        {\langle G, (\Gamma, A \lor B) \rangle \in \data{BLG}}
    \qquad
    \prftree[r]{\rl{{\uparrow}{\lor}}}
        {\prfassumption{\langle G, (\Gamma, A \lor B) \rangle \in \data{BLG}}}
        {\langle G, (\Gamma, A, B) \rangle \in \data{BLG}}
    \\[1em]
    \prftree[r]{\rl{{\downarrow}{\land}}}
        {\prfassumption{\langle G, (\Gamma, A) \rangle \in \data{BLG}}}
        {\prfassumption{\langle H, (\Gamma, B) \rangle \in \data{BLG}}}
        {\langle G \sqcup H, (\Gamma, A \land B) \rangle \in \data{BLG}}
    \\[1em]
    \prftree[r]{\rl{{\uparrow}{\land_l}}}
        {\prfassumption{\langle G, (\Gamma, A \land B) \rangle \in \data{BLG}}}
        {\langle G \rst_{\Gamma, A}, (\Gamma, A) \rangle \in \data{BLG}}
    \qquad
    \prftree[r]{\rl{{\uparrow}{\land_r}}}
        {\prfassumption{\langle G, (\Gamma, A \land B) \rangle \in \data{BLG}}}
        {\langle G \rst_{\Gamma, B}, (\Gamma, B) \rangle \in \data{BLG}}
\end{gather*}
The cut-rule is also admissible by \cref{coro:composition-preserves-totality}:
\[
    \prftree[r]{\rl{cut}}
        {\prfassumption{\langle G, (\Gamma, A) \rangle \in \data{BLG}}}
        {\prfassumption{\langle H, (\Gamma, \dl{A}) \rangle \in \data{BLG}}}
        {\langle G \odot_A H, \Gamma \rangle \in \data{BLG}}
\]
Finally, axiom and superposition rules are obviously admissible through their interpretation:
\begin{gather*}
    \prftree[r]{\rl{ax}}
        {\prfassumption{A \equiv B}}
        {\langle \blwk{\Gamma}(\blid{\{A, \dl{B}\}}), (\Gamma, A, \dl{B}) \rangle \in \data{BLG}}
    \qquad
    \prftree[r]{\rl{\sqcup}}
        {\prfassumption{\langle G, \Gamma \rangle \in \data{BLG}}}
        {\prfassumption{\langle H, \Gamma \rangle \in \data{BLG}}}
        {\langle G \sqcup H, \Gamma \rangle \in \data{BLG}}
\end{gather*}
Isolation is as expected an identity on~\(\data{BLG}\) proofs. The inversion procedure
amounts to a simple replacement of the conclusion for disjunctions; in the case of conjunctions
some time must be spent computing the graph restriction, but it will be often significantly
less than the time spent to perform inversion on a sequent calculus derivation, especially
when all axioms in the derivation are atomic. These facts -- together with the immediate
notation and the availability of equational reasoning -- make~\(\data{BLG}\) a very
efficient tool for reasoning about proof content and transformations.

The price to be paid lies in the size of proof objects, which is often exponential in the
complexity of the conclusion even when much smaller derivations would be available in the
context-splitting formulation of sequent calculus or even in~\(\data{GS4}\) with non-atomic
axioms. This is a well-known problem that~\(\data{BLG}\) shares with~\(\data{GS4}\) when
restricted to atomic axioms. It is an open question whether some proof-compression method
could be devised to reduce the size~\(\data{BLG}\) proofs without loosing the good properties
of the system.

\section{On the absence of a cut-reduction procedure}\label{sec:cut-reduction-failure}

By cut-reduction procedure we mean a set of \emph{syntactical} rewriting steps capable
of reducing the complexity of cut-rule applications in a derivation until they become atomic,
while at the same time preserving correctness and the conclusion of the derivation. One such
step is implemented in the proof of \cref{thm:normalisation}, where we apply the isolation
procedure to permute cuts towards the top of the derivation, thus reducing the complexity
of their context.

At present, however, we don't know of any rewriting step capable of reducing the complexity
of cut-formulas in~\(\data{GS4}^\mathcal{N}\) derivations while preserving their associated
bl-graph. There is a known pair of natural cut-reduction steps, considered by Pulcini
in~\cite{Pul22}:
\begin{gather*}
    \makebox[\textwidth]{\small\(
        \vcenter{\prftree[r]{\rl{cut}}
            {\prftree[r]{\rl{\lor}}
                {\prfsummary[\(P\)]{\th \Gamma, A, B \vphantom{\dl{AB}}}}
                {\th \Gamma, A \lor B}}
            {\prftree[r]{\rl{\land}}
                {\prfsummary[\(Q\)]{\th \Gamma, \dl{A}}}
                {\prfsummary[\(R\)]{\th \Gamma, \dl{B}}}
                {\th \Gamma, \dl{A} \land \dl{B}}}
            {\th \Gamma}}
        \quad\rwr\quad
        \vcenter{\prftree[r]{\rl{cut}}
            {\prftree[r]{\rl{cut}}
                {\prfsummary[\(P\)]{\th \Gamma, A, B \vphantom{\dl{A}}}}
                {\prfsummary[\(\data{wk}(Q, B)\)]{\th \Gamma, \dl{A}, B}}
                {\th \Gamma, B}}
            {\prfsummary[\(R\)]{\th \Gamma, \dl{B}}}
            {\th \Gamma}}
    \)}
    \\[1em]
    \makebox[\textwidth]{\small\(
        \vcenter{\prftree[r]{\rl{cut}}
            {\prftree[r]{\rl{\lor}}
                {\prfsummary[\(P\)]{\th \Gamma, A, B \vphantom{\dl{AB}}}}
                {\th \Gamma, A \lor B}}
            {\prftree[r]{\rl{\land}}
                {\prfsummary[\(Q\)]{\th \Gamma, \dl{A}}}
                {\prfsummary[\(R\)]{\th \Gamma, \dl{B}}}
                {\th \Gamma, \dl{A} \land \dl{B}}}
            {\th \Gamma}}
        \quad\rwr\quad
        \vcenter{\prftree[r]{\rl{cut}}
            {\prftree[r]{\rl{cut}}
                {\prfsummary[\(P\)]{\th \Gamma, A, B \vphantom{\dl{B}}}}
                {\prfsummary[\(\data{wk}(R, A)\)]{\th \Gamma, A, \dl{B}}}
                {\th \Gamma, A}}
            {\prfsummary[\(Q\)]{\th \Gamma, \dl{A}}}
            {\th \Gamma}}
    \)}
\end{gather*}
These are straightforward adaptations of the usual key steps for context-splitting systems
to the context-sharing setting of the system~\(\data{GS4}^\mathcal{N}\). The resulting
cut reduction procedure is non-deterministic at the syntactical level, but both rewriting
steps preserve the axiom graph obtained through the unrefined construction (\cref{defn:axiom-graphs}).

Unfortunately, the two steps are incompatible with the refined interpretation~(\cref{defn:bl-axiom-graphs}).
To see why, one must remember that the branch-sensitive notion of composition is designed
to omit all those paths that connect the two sides of conjunctions occurring \emph{outside
the interface}. The two rewriting steps above, in reducing one cut to \emph{two} cuts of
lower complexity, select one subformula of the disjunction to bring temporarily outside the
interface. This might result in the loss of some paths when the selected formula is a conjunction:
\[
    \makebox[\textwidth]{\small\(
        \vcenter{\prftree[r]{\rl{cut}}
            {\prftree[r]{\rl{\lor}}
                {\prfsummary[\(P\)]{\th \Gamma, A, B \land C \vphantom{\dl{ABC}}}}
                {\th \Gamma, A \lor (B \land C)}}
            {\prftree[r]{\rl{\land}}
                {\prfsummary[\(Q\)]{\th \Gamma, \dl{A}}}
                {\prfsummary[\(R\)]{\th \Gamma, \dl{B} \lor \dl{C}}}
                {\th \Gamma, \dl{A} \land (\dl{B} \lor \dl{C})}}
            {\th \Gamma}}
        \quad\rwr\quad
        \vcenter{\prftree[r]{\rl{cut}}
            {\prftree[r]{\rl{cut}}
                {\prfsummary[\(P\)]{\th \Gamma, A, B \land C \vphantom{\dl{A}}}}
                {\prfsummary[\(\data{wk}(Q, B \land C)\)]{\th \Gamma, \dl{A}, B \land C}}
                {\th \Gamma, B \land C}}
            {\prfsummary[\(R\)]{\th \Gamma, \dl{B} \lor \dl{C}}}
            {\th \Gamma}}
    \)}
\]
Here the left side of the rewriting rule is associated to the bl-graph
\[
    \blaxg{P} \odot_{(A \lor (B \land C))} (\blaxg{Q} \sqcup \blaxg{R}),
\]
while the right side is interpreted as
\[
    (\blaxg{P} \odot_A \blaxg{Q}) \odot_{B \land C} \blaxg{R};
\]
the conjunction \(B \land C\) is contained in the interface in the former expression, where
the interpretation is computed by a single composition step, but lies outside the interface
in the first half of the latter expression, where we compute first the intermediate step
\(\blaxg{P} \odot_A \blaxg{Q}\). To show that this is an actual problem, not just a hypothetical
possibility, we provide a complete (but slightly involved) counter-example in \cref{fig:cut-reduction-failure}.

It remains unclear whether there is some syntactical counterpart to the semantical cut-elimination
procedure described in \cref{sec:normalisation}. It is possible that the problem could be
solved simply by superposing the two reducts, thus making the rewriting step deterministic:
\[
    \makebox[\textwidth]{\small\(
        \prftree[r]{\rl{\sqcup}}
            {\prftree[r]{\rl{cut}}
                {\prftree[r]{\rl{cut}}
                    {\prfsummary[\(P\)]{\th \Gamma, A, B \vphantom{\dl{A}}}}
                    {\prfsummary[\(\data{wk}(Q, B)\)]{\th \Gamma, \dl{A}, B}}
                    {\th \Gamma, B}}
                {\prfsummary[\(R\)]{\th \Gamma, \dl{B}}}
                {\th \Gamma}}
            {\prftree[r]{\rl{cut}}
                {\prftree[r]{\rl{cut}}
                    {\prfsummary[\(P\)]{\th \Gamma, A, B \vphantom{\dl{B}}}}
                    {\prfsummary[\(\data{wk}(R, A)\)]{\th \Gamma, A, \dl{B}}}
                    {\th \Gamma, A}}
                {\prfsummary[\(Q\)]{\th \Gamma, \dl{A}}}
                {\th \Gamma}}
            {\th \Gamma}
    \)}
\]
We have not been able to find a counter-example so far, but we have yet to attempt a proof.
We would have to show somehow that whenever some path is erased on one side, then it
must be preserved on the other side, and \emph{vice versa}.

A different possibility would be to further refine the interpretation so as to make it
invariant under the two traditional cut-reduction steps. The key observation in this regard
is that all counter-examples known to us -- including indeed the one in \cref{fig:cut-reduction-failure}
-- seem to rely critically on the possibility of constructing alternating paths using edges
from both branches of a superposition rule. This happens because superposition is not
interpreted as a non-deterministic sum of separate proofs, but as some kind of parallel
composition that “blends” two proofs together, thus obtaining a new and generally distinct one.
We conjecture that by treating superpositions as proper non-deterministic sums we should be
able to recover invariance under the traditional logical reduction steps. The picture however
is complicated by the fact that superpositions might be introduced when reducing weakening-%
weakening cuts: it is not clear then how the new composition operator should look like.

% \subsubsection{Acknowledgements}\todo{Acknowledgements}
% Please place your acknowledgments at the end of the paper, preceded by an
% unnumbered run-in heading (i.e. 3rd-level heading).

\printbibliography

\appendix

\section{Graphs}\label{sec:graphs}

\begin{definition}[Simple graphs]
    \label{defn:graph}
    A \emph{graph} (also known as \emph{simple graph}) is defined by a pair
    \(G = \langle V_G, E_G \rangle\) where \(V_G\) is a set of \emph{vertices} and
    \(E_G\) is a set of unordered pairs of vertices, i.e.\ \emph{two-element subsets}
    of~\(V_G\), called the \emph{edges of \(G\)}. For any pair of distinct vertices
    \(u \neq v \in V_G\), let \(uv\) denote the set \(\{u, v\}\) and say that \(u\)
    and \(v\) are \emph{adjacent} in \(G\) iff \(uv \in E_G\).
\end{definition}

\begin{definition}[Subgraphs]
    \label{defn:subgraphs}
    Given graphs \(H, G\), say that \emph{\(H\) is a subgraph of \(G\)} and write
    \(H \sqsubseteq G\) iff \(V_H \subseteq V_G\) and \(E_H \subseteq E_G\). Write
    \(H \sqsubset G\) iff \(H \sqsubseteq G\) and \(H \neq G\).

    Given a graph \(G\), any set \(S\) of vertices induces a subgraph \(H \sqsubseteq G\)
    by letting \(V_H = V_H \cap S\) and restricting the edge set:
    \[
        E_H = \{ e \in E_G \mid e \subseteq S \}.
    \]
    We say that \(H\) is an \emph{induced subgraph} of \(G\) and denote it by \(G \rst_S\).
\end{definition}

% \begin{definition}[Paths, connectivity]
%     \label{defn:paths}
%     A \emph{path} in a graph \(G\) is a sequence \(u_1 \ldots u_n \in V_G\) (with \(n \geq 1\))
%     of \emph{pairwise distinct} vertices such that \(u_i u_{i+1} \in E_G\) for any \(1 \leq i < n\).

%     Say that vertices \(u, v \in V_G\) are \emph{connected} in \(G\) iff there is a path
%     \(u_1 \ldots u_n\) in \(G\) with \(u_1 = u\) and \(u_n = v\).
% \end{definition}

Graphs can be combined by an operation related to set-theoretic union, which constructs the least
upper bound of a family of graphs w.r.t.\ the subgraph order~\({\sqsubseteq}\):

\begin{definition}[Graph union]
    \label{defn:union-graph}
    Let \(I\) be a set, \((G_i)_{i \in I}\) a family of graphs indexed by~\(I\). We define
    the \emph{union} or \emph{superposition} of the family as the graph
    \[
        \bigsqcup_{i \in I} G_i = \langle \bigcup_{i \in I} V_{G_i}, \bigcup_{i \in I} E_{G_i} \rangle.
    \]
    As a special case, given \(n \geq 0\) and graphs \(G_1, \ldots, G_n\), we define the
    finite union
    \[
        G_1 \sqcup \ldots \sqcup G_n = \bigsqcup_{i = 1}^n G_i.
    \]
\end{definition}

\begin{proposition}
    \label{propo:graph-union-lub}
    Let \(I\) be a set, \((G_i)_{i \in I}\) a family of graphs indexed by~\(I\); the union
    graph \(\bigsqcup_{i \in I} G_i\) is the least upper bound of the family w.r.t.\ the
    subgraph order \({\sqsubseteq}\).
\end{proposition}
% \begin{proof}
%     Immediate consequence of the fact that set-theoretic union is the least upper bound w.r.t.\ inclusion.
% \end{proof}

\begin{proposition}
    \label{propo:restriction-distributes-over-union}
    Let \(I\) be any set, \((G_i)_{i \in I}\) a family of graphs indexed by~\(I\), \(S\) an
    arbitrary set of vertices: then
    \[
        \left(\bigsqcup_{i \in I} G_i\right)\rst_S = \bigsqcup_{i \in I} (G_i\rst_S)
    \]
\end{proposition}
% \begin{proof}
%     Let~\(H = \left(\bigsqcup_{i \in I} G_i\right)\rst_S\) and~\(K = \bigsqcup_{i \in I} (G_i\rst_S)\).
%     We must show that they have the same vertex and edge sets. For the vertices the conclusion
%     is immediate from the fact that intersection distributes over union. For the edges,
%     observe first that
%     \begin{align*}
%         E_H = &\{ e \in \bigcup_{i \in I} E_{G_i} \mid e \subseteq S \},\\[.5em]
%         E_K = &\bigcup_{i \in I} \{ e \in E_{G_i} \mid e \subseteq S \}.
%     \end{align*}
%     Assume then \(e \in E_H\): we have \(e \subseteq S\) and there is~\(i \in I\) such
%     that \(e \in E_{G_i}\), hence \(e \in E_K\); conversely, let~\(e \in E_K\): there
%     is again~\(i \in I\) such that \(e \in E_{G_i}\) with~\(e \subseteq S\), hence \(e \in E_H\).
% \end{proof}

\section{Complexity measures on formulas, sequents and derivations}\label{sec:complexity}

We define in the following paragraphs some useful measures of the complexity of formulas,
sequents and derivations, and discuss the relationships between them.

\begin{definition}[Complexity of formulas]
    \label{defn:formula-complexity}
    We define by structural induction four different measures on named formulas~\(A\):
    the \emph{height}~\(\data{h}(A)\), the \emph{atom count}~\(\data{\#at}(A)\), the
    \emph{degree}~\(\data{deg}(A)\) and the \emph{size}~\(\data{size}(A)\):
    \begin{itemize}%[--]
        \item if \(A\) is atomic, let
        \begin{gather*}
            \data{h}(A) = \data{deg}(A) = 0,\\
            \data{\#at}(A) = \data{size}(A) = 1;
        \end{gather*}

        \item if \(A = B \lor C\) or \(A = B \land C\), let
        \begin{align*}
            \data{h}(A)    &= 1 + \max\{\data{h}(B), \data{h}(C)\},\\
            \data{\#at}(A)  &= \data{\#at}(B) + \data{\#at}(C),\\
            \data{deg}(A)  &= 1 + \data{deg}(B) + \data{deg}(C),\\
            \data{size}(A) &= 1 + \data{size}(B) + \data{size}(C).
        \end{align*}
    \end{itemize}
\end{definition}

The height tracks the length of the longest branch in the syntax tree of the formula,
the atom count is self-explanatory, the degree is the number of logical operators occurring
in the formula, and the size is the number of symbols (atoms and operators).

\begin{proposition}
    \label{propo:formula-complexity-relationships}
    For all named formulas~\(A\):
    \begin{enumerate}[(i)]
        \item \(\data{\#at}(A) = 1 + \data{deg}(A)\);
        \item \(\data{size}(A) = \data{\#at}(A) + \data{deg}(A) = 1 + 2 \cdot \data{deg}(A)\);
        \item \(\data{h}(A) \leq \data{deg}(A) < \data{\#at}(A) \leq \data{size}(A)\).
    \end{enumerate}
\end{proposition}
\begin{proof}
    The inequalities \(\deg(A) < \data{\#at}(A)\) and~\(\data{\#at}(A) \leq \data{size}(A)\)
    are direct consequences of facts~(i) and~(ii). We prove facts~(i) and~(ii) and the first
    inequality simultaneously by structural induction on~\(A\):
    \begin{itemize}
        \item if \(A\) is atomic, then
        \begin{gather*}
            0 = \data{h}(A) \leq \data{deg}(A) = 0,\\
            \data{\#at}(A) = 1 = 1 + 0 = 1 + \data{deg}(A) \text{, and}\\
            \data{size}(A) = 1 = 1 + 0 = \data{\#at}(A) + \data{deg}(A);
        \end{gather*}

        \item if \(A = B \lor C\) or \(A = B \land C\), then
        \begin{align*}
            \data{\#at}(A) &= \data{\#at}(B) + \data{\#at}(C)                                       & \text{(def.~\ref{defn:formula-complexity})}\\
                          &= \data{deg}(B) + 1 + \data{deg}(C) + 1                                  & \text{(ind.\ hyp.)}\\
                          &= 1 + \data{deg}(B) + \data{deg}(C) + 1                                  \\
                          &= \data{deg}(A) + 1                                                      & \text{(def.~\ref{defn:formula-complexity})}\\[1em]
            \data{size}(A) &= 1 + \data{size}(B) + \data{size}(C)                                   & \text{(def.~\ref{defn:formula-complexity})}\\
                           &= 1 + \data{\#at}(B) + \data{deg}(B) + \data{\#at}(C) + \data{deg}(C)   & \text{(ind.\ hyp.)}\\
                           &= \data{\#at}(B) + \data{\#at}(C) + 1 + \data{deg}(B) + \data{deg}(C)   \\
                           &= \data{\#at}(A) + \data{deg}(A)                                        & \text{(def.~\ref{defn:formula-complexity})}\\[1em]
            \data{h}(A) &= 1 + \max\{\data{h}(B), \data{h}(C)\}                                     & \text{(def.~\ref{defn:formula-complexity})}\\
                        &\leq 1 + \data{h}(B) + \data{h}(C)                                         \\
                        &\leq 1 + \data{deg}(B) + \data{deg}(C)                                     & \text{(ind.\ hyp.)}\\
                        &= \data{deg}(A)                                                            & \text{(def.~\ref{defn:formula-complexity})}
        \end{align*}
    \end{itemize}
\end{proof}

\begin{definition}[Complexity of sequents]
    \label{defn:sequent-complexity}
    We define for sequents the same measures as for formulas, by taking sums over all
    formulas contained in the sequent: for all named sequents~\(\th \Gamma\) let
    \begin{align*}
        \data{h}(\th \Gamma)   &= \sum_{A \in \Gamma} \data{h}(A),   &\data{\#at}(\th \Gamma) &= \sum_{A \in \Gamma} \data{\#at}(A),\\[.5em]
        \data{deg}(\th \Gamma) &= \sum_{A \in \Gamma} \data{deg}(A), &\data{size}(\th \Gamma) &= \sum_{A \in \Gamma} \data{size}(A).
    \end{align*}
    Since \(\Gamma\) is required to be finite, all measures are well-defined.
\end{definition}

\begin{proposition}
    \label{propo:sequent-complexity-relationships}
    For all named sequents~\(\th \Gamma\):
    \begin{enumerate}[(i)]
        \item \(\data{\#at}(\th \Gamma) = |\Gamma| + \data{deg}(\th \Gamma)\);
        \item \(\data{size}(\th \Gamma) = \data{\#at}(\th \Gamma) + \data{deg}(\th \Gamma) = |\Gamma| + 2 \cdot \data{deg}(A)\);
        \item \(\data{h}(\th \Gamma) \leq \data{deg}(\th \Gamma) < \data{\#at}(\th \Gamma) \leq \data{size}(\th \Gamma)\).
    \end{enumerate}
\end{proposition}
\begin{proof}
    For fact~(i), observe that
    \begin{align*}
        \data{\#at}(\th \Gamma) &= \sum_{A \in \Gamma} \data{\#at}(A)           & \text{(def.~\ref{defn:sequent-complexity})}\\[.5em]
                               &= \sum_{A \in \Gamma} 1 + \data{deg}(A)         & \text{(lemma~\ref{propo:formula-complexity-relationships} (i))}\\[.5em]
                               &= \left(\sum_{A \in \Gamma} 1\right) + \left(\sum_{A \in \Gamma} \data{deg}(A)\right) & \text{(associativity)}\\[.5em]
                               &= \#\Gamma + \data{deg}(\th \Gamma)             & \text{(def.~\ref{defn:sequent-complexity}).}
    \end{align*}
    Fact~(ii) follows from \cref{defn:sequent-complexity} and the associativity of sums;
    fact~(iii) follows from \cref{propo:formula-complexity-relationships}, point~(iii) and
    the fact that addition is strictly monotonic over the natural numbers.
\end{proof}

\begin{definition}[Complexity of derivations]
    \label{defn:derivation-complexity}
    We define by structural induction two complexity measures on named derivation
    trees~\(P \in \data{GS4}^\mathcal{N}\): the \emph{height}~\(\data{h}(P)\) and the
    \emph{size}~\(\data{size}(P)\):
    \begin{itemize}%[--]
        \item if \(P\) has the form
        \[
            \prfbyaxiom{\rl[\{A,\dl{B}\}]{ax}}{\th \Gamma, A, \dl{B}}
        \]
        where \(A \equiv B\), then let
        \[
            \data{h}(P) = 0,\ \data{size}(P) = 1;
        \]

        \item otherwise \(P\) has the form
        \[
            \prftree[r]{\(r\)}
                {\prfsummary[\(Q_1\)]{\th \Gamma_1}}
                {\prfassumption{\cdots}}
                {\prfsummary[\(Q_n\)]{\th \Gamma_n}}
                {\th \Gamma}
        \]
        where \(r \in \{\rl{cut}, \rl{\sqcup}, \rl{\lor}, \rl{\land}\}\)
        and \(n > 0\): then let
        \begin{align*}
            \data{h}(P)    &= 1 + \max_{1 \leq i \leq n} \data{h}(Q_i),\\[.5em]
            \data{size}(P) &= 1 + \sum_{1 \leq i \leq n} \data{size}(Q_i).
        \end{align*}
    \end{itemize}
\end{definition}

\section{Proofs of the inversion lemmas}\label{sec:inversion-proofs}

We collect in the present section all proofs of the inversion lemmas from \cref{sec:inversion-lemmas}
plus the proofs of admissibility for contraction and weakening.

\begin{proof}[Proof of \cref{lemma:disjunction-transform}]
    By structural induction on~\(P\):
    \begin{itemize}%[--]
        \item if \(P\) has the form
        \[
            \prfbyaxiom{\rl[\{\dl{C} \land \dl{D}, A \lor B\}]{ax}}
                {\th \Delta, \dl{C} \land \dl{D}, A \lor B}
        \]
        where \(A \equiv C\), \(B \equiv D\) and~\(\Gamma = \Delta \cup \{\dl{C} \land \dl{D}\}\),
        then let~\(\data{inv}(P, A \lor B)\) be
        \[
            \prftree[r]{\rl{\land}}
                {\prfbyaxiom{\rl[\{\dl{C}, A\}]{ax}}{\th \Delta, \dl{C}, A, B}}
                {\prfbyaxiom{\rl[\{\dl{D}, B\}]{ax}}{\th \Delta, \dl{D}, A, B}}
                {\th \Delta, \dl{C} \land \dl{D}, A, B}
        \]

        \item if \(P\) has the form
        \[
            \prfbyaxiom{\rl[\{C, \dl{D}\}]{ax}}
                {\th \Delta, C, \dl{D}, A \lor B}
        \]
        where \(C \equiv D\) and~\(\Gamma = \Delta \cup \{C, \dl{D}\}\), then let~\(\data{inv}(
        P, A \lor B)\) be
        \[
            \prfbyaxiom{\rl[\{C, \dl{D}\}]{ax}}
                {\th \Delta, C, \dl{D}, A, B}
        \]

        \item if \(P\) has the form
        \[
            \prftree[r]{\rl{\lor}}
                {\prfsummary[\(Q\)]{\th \Gamma, A, B}}
                {\th \Gamma, A \lor B}
        \]
        then let~\(\data{inv}(P, A \lor B) = Q\);

        \item if \(P\) has the form
        \[
            \prftree[r]{\(r\)}
                {\prfsummary[\(Q_1\)]{\th \Gamma_1, A \lor B}}
                {\prfassumption{\cdots}}
                {\prfsummary[\(Q_n\)]{\th \Gamma_n, A \lor B}}
                {\th \Gamma, A \lor B}
        \]
        where \(r \in \{\rl{cut}, \rl{\sqcup}, \rl{\lor}, \rl{\land}\}\)
        and \(n > 0\), then apply~\(\data{inv}({-}, A \lor B)\) recursively to each~\(Q_i\)
        and let~\(\data{inv}(P, A \lor B)\) be
        \[
            \prftree[r]{\(r\)}
                {\prfsummary[\(\data{inv}(Q_1, A \lor B)\)]{\th \Gamma_1, A, B}}
                {\prfassumption{\cdots}}
                {\prfsummary[\(\data{inv}(Q_n, A \lor B)\)]{\th \Gamma_n, A, B}}
                {\th \Gamma, A, B}\qedhere
        \]
    \end{itemize}
\end{proof}

\begin{proof}[Proof of \cref{lemma:left-conjunction-transform}]
    By structural induction on \(P\):
    \begin{itemize}%[--]
        \item if \(P\) has the form
        \[
            \prfbyaxiom{\rl[\{\dl{C} \lor \dl{D}, A \land B\}]{ax}}
                {\th \Delta, \dl{C} \lor \dl{D}, A \land B}
        \]
        where \(A \equiv C\), \(B \equiv D\) and~\(\Gamma = \Delta \cup \{\dl{C} \lor \dl{D}\}\),
        then let~\(\data{inv_l}(P, A \land B)\) be
        \[
            \prftree[r]{\rl{\lor}}
                {\prfbyaxiom{\rl[\{\dl{C}, A\}]{ax}}{\th \Delta, \dl{C}, \dl{D}, A}}
                {\th \Delta, \dl{C} \lor \dl{D}, A}
        \]

        \item if \(P\) has the form
        \[
            \prfbyaxiom{\rl[\{C, \dl{D}\}]{ax}}
                {\th \Delta, C, \dl{D}, A \land B}
        \]
        where \(C \equiv D\) and~\(\Gamma = \Delta \cup \{C, \dl{D}\}\), then let~\(\data{inv_l}(
        P, A \land B)\) be
        \[
            \prfbyaxiom{\rl[\{C, \dl{D}\}]{ax}}
                {\th \Delta, C, \dl{D}, A}
        \]

        \item if \(P\) has the form
        \[
            \prftree[r]{\rl{\land}}
                {\prfsummary[\(Q\)]{\th \Gamma, A}}
                {\prfsummary[\(R\)]{\th \Gamma, B}}
                {\th \Gamma, A \land B}
        \]
        then let~\(\data{inv_l}(P, A \land B) = Q\);

        \item if \(P\) has the form
        \[
            \prftree[r]{\(r\)}
                {\prfsummary[\(Q_1\)]{\th \Gamma_1, A \land B}}
                {\prfassumption{\cdots}}
                {\prfsummary[\(Q_n\)]{\th \Gamma_n, A \land B}}
                {\th \Gamma, A \land B}
        \]
        where \(r \in \{\rl{cut}, \rl{\sqcup}, \rl{\lor}, \rl{\land}\}\)
        and \(n > 0\), then apply~\(\data{inv_l}({-}, A \land B)\) recursively to each~\(Q_i\)
        and let~\(\data{inv_l}(P, A \land B)\) be
        \[
            \prftree[r]{\(r\)}
                {\prfsummary[\(\data{inv_l}(Q_1, A \land B)\)]{\th \Gamma_1, A}}
                {\prfassumption{\cdots}}
                {\prfsummary[\(\data{inv_l}(Q_n, A \land B)\)]{\th \Gamma_n, A}}
                {\th \Gamma, A}\qedhere
        \]
    \end{itemize}
\end{proof}

\begin{proof}[Proof of \cref{lemma:right-conjunction-transform}]
    By structural induction on \(P\):
    \begin{itemize}%[--]
        \item if \(P\) has the form
        \[
            \prfbyaxiom{\rl[\{\dl{C} \lor \dl{D}, A \land B\}]{ax}}
                {\th \Delta, \dl{C} \lor \dl{D}, A \land B}
        \]
        where \(A \equiv C\), \(B \equiv D\) and~\(\Gamma = \Delta \cup \{\dl{C} \lor \dl{D}\}\),
        then let~\(\data{inv_r}(P, A \land B)\) be
        \[
            \prftree[r]{\rl{\lor}}
                {\prfbyaxiom{\rl[\{\dl{D}, B\}]{ax}}{\th \Delta, \dl{C}, \dl{D}, B}}
                {\th \Delta, \dl{C} \lor \dl{D}, B}
        \]

        \item if \(P\) has the form
        \[
            \prfbyaxiom{\rl[\{C, \dl{D}\}]{ax}}
                {\th \Delta, C, \dl{D}, A \land B}
        \]
        where \(C \equiv D\) and~\(\Gamma = \Delta \cup \{C, \dl{D}\}\), then let~\(\data{inv_r}(
        P, A \land B)\) be
        \[
            \prfbyaxiom{\rl[\{C, \dl{D}\}]{ax}}
                {\th \Delta, C, \dl{D}, B}
        \]

        \item if \(P\) has the form
        \[
            \prftree[r]{\rl{\land}}
                {\prfsummary[\(Q\)]{\th \Gamma, A}}
                {\prfsummary[\(R\)]{\th \Gamma, B}}
                {\th \Gamma, A \land B}
        \]
        then let~\(\data{inv_r}(P, A \land B) = R\);

        \item if \(P\) has the form
        \[
            \prftree[r]{\(r\)}
                {\prfsummary[\(Q_1\)]{\th \Gamma_1, A \land B}}
                {\prfassumption{\cdots}}
                {\prfsummary[\(Q_n\)]{\th \Gamma_n, A \land B}}
                {\th \Gamma, A \land B}
        \]
        where \(r \in \{\rl{cut}, \rl{\sqcup}, \rl{\lor}, \rl{\land}\}\)
        and \(n > 0\), then apply~\(\data{inv_r}({-}, A \land B)\) recursively to each~\(Q_i\)
        and let~\(\data{inv_r}(P, A \land B)\) be
        \[
            \prftree[r]{\(r\)}
                {\prfsummary[\(\data{inv_r}(Q_1, A \land B)\)]{\th \Gamma_1, B}}
                {\prfassumption{\cdots}}
                {\prfsummary[\(\data{inv_r}(Q_n, A \land B)\)]{\th \Gamma_n, B}}
                {\th \Gamma, B}\qedhere
        \]
    \end{itemize}
\end{proof}

\begin{proof}[Proof of \cref{lemma:wk-transform}]
    Let us fix once and for all a complete enumeration \(x_1, x_2, \ldots\) \emph{without
    repetitions} of the countable set~\(\mathcal{N}\) of names. The \(\data{wk}\)-transformation
    will be deterministic up to our choice of enumeration. We are also gonna need the theory of
    renamings from \cref{sec:renamings}. We proceed by induction on the height of~\(P\):
    \begin{itemize}%[--]
        \item if \(P\) has the form
        \[
            \prfbyaxiom{\rl[\{C, \dl{D}\}]{ax}}
                {\th \Gamma', C, \dl{D}}
        \]
        where \(C \equiv D\) and~\(\Gamma = \Gamma' \cup \{C, \dl{D}\}\), then let~\(\data{wk}(
        P, \Delta)\) be
        \[
            \prfbyaxiom{\rl[\{C, \dl{D}\}]{ax}}
                {\th \Gamma', C, \dl{D}, \Delta}
        \]

        \item if \(P\) has the form
        \[
            \prftree[r]{\(r\)}
                {\prfsummary[\(Q_1\)]{\th \Gamma_1}}
                {\prfassumption{\cdots}}
                {\prfsummary[\(Q_n\)]{\th \Gamma_n}}
                {\th \Gamma}
        \]
        where \(r \in \{\rl{\sqcup}, \rl{\lor}, \rl{\land}\}\) and \(n > 0\),
        then one may check easily (by inspection of the three rules listed before) that
        \(\data{names}(\Gamma_i) \subseteq \data{names}(\Gamma)\) for all~\(1 \leq i
        \leq n\): apply~\(\data{wk}({-}, \Delta)\) recursively to each~\(Q_i\) and
        let~\(\data{wk}(P, \Delta)\) be
        \[
            \prftree[r]{\(r\)}
                {\prfsummary[\(\data{wk}(Q_1, \Delta)\)]{\th \Gamma_1, \Delta}}
                {\prfassumption{\cdots}}
                {\prfsummary[\(\data{wk}(Q_n, \Delta)\)]{\th \Gamma_n, \Delta}}
                {\th \Gamma, \Delta}
        \]

        \item if \(P\) has the form
        \[
            \prftree[r]{\rl{cut}}
                {\prfsummary[\(Q\)]{\th \Gamma, A \vphantom{\dl{A}}}}
                {\prfsummary[\(R\)]{\th \Gamma, \dl{A}}}
                {\th \Gamma}
        \]
        we must take into account the possibility that \(A, \dl{A}\) and~\(\Delta\) share some
        names, as the hypothesis only guarantees that \(\Delta\) share no names with~\(\Gamma\).
        Let then \(k\) be the smallest natural number such that, for all~\(1 \leq i \geq k\),
        \(x_i \notin \data{names}(P) \cup \data{names}(\Delta)\): such a~\(k\) must exist
        because only finitely many names may occur either in~\(P\) or~\(\Delta\), and the
        enumeration has no repetitions. We construct a renaming~\(\phi: \data{names}(P) \to \mathcal{N}\)
        of~\(P\) by letting
        \[
            \phi(x_i) = \begin{cases}
                x_{i+k}     & \text{if } x_i \in \data{names}(\Delta),\\
                x_i         & \text{otherwise,}
            \end{cases}
        \]
        for all \(i \in \N\) such that \(x_i \in \data{names}(P)\). Because no name
        in~\(\Gamma\) occurs also in~\(\Delta\), \(\phi\) is guaranteed to act as the
        identity on all formulas in~\(\Gamma\), i.e.~\(\phi P\) has the form
        \[
            \prftree[r]{\rl{cut}}
                {\prfsummary[\(\phi Q\)]{\th \Gamma, \phi A \vphantom{\sdl{(\phi A)}}}}
                {\prfsummary[\(\phi R\)]{\th \Gamma, \sdl{(\phi A)}}}
                {\th \Gamma}
        \]
        where \(\phi A, \sdl{(\phi A)}\) share no location with~\(\Delta\) by construction
        of~\(\phi\), and \(\phi Q, \phi R\) have by \cref{propo:derivation-renaming-correctness}
        the same height as \(Q, R\) respectively: let then \(\data{wk}(P, \Delta)\) be
        \[
            \prftree[r]{\rl{cut}}
                {\prfsummary[\(\data{wk}(\phi Q, \Delta)\)]{\th \Gamma, \phi A, \Delta \vphantom{\sdl{(\phi A)}}}}
                {\prfsummary[\(\data{wk}(\phi R, \Delta)\)]{\th \Gamma, \sdl{(\phi A)}, \Delta}}
                {\th \Gamma, \Delta}
            \qedhere
        \]
    \end{itemize}
\end{proof}

\section{Renamings}\label{sec:renamings}

\begin{definition}
    \label{defn:renaming}
    A \emph{renaming} of a named formula~\(A\) (resp.\ set~\(\Gamma\) of named formulas) is an
    \emph{injective map} from \(\data{names}(A)\) (resp.~\(\data{names}(\Gamma)\))
    to~\(\mathcal{N}\). For any named formula~\(A\) and renaming~\(\phi\) of~\(A\), let
    \[
        \phi A = \begin{cases}
            \nm{\phi x}{\alpha}         & \text{if } A = \nm{x}{\alpha} \text{ for some } \alpha \in \mathcal{A}, x \in \mathcal{N}, \\
            \phi B \lor \phi C          & \text{if } A = B \lor C, \\
            \phi B \land \phi C         & \text{if } A = B \land C.
        \end{cases}
    \]
    The application of a renaming~\(\phi\) to a set~\(\Gamma\) is defined by taking its image
    under~\(\phi\), observing that, for each \(A \in \Gamma\), the restriction of~\(\phi\)
    to~\(\data{names}(A)\) is necessarily a renaming of~\(A\).
\end{definition}

\begin{proposition}
    \label{propo:formula-renaming-correctness}
    For any named formula~\(A\) and renaming~\(\phi\) of~\(A\), \(\phi A\) is also a named
    formula with \(A \equiv \phi A\), \(\data{names}(\phi A) = \phi(\data{names}(A))\),
    and sharing-free iff so is~\(A\).
\end{proposition}
\begin{proof}
    The equivalence result and that about name sets are more or less immediate, by structural
    induction on~\(A\). As regards sharing-freedom, we reason again by structural induction
    on~\(A\):
    \begin{itemize}%[--]
        \item if \(A\) is atomic then so is~\(\phi A\) and they are both sharing-free
        by definition;

        \item if \(A = B \diamond C\) (where \({\diamond} \in \{{\lor},{\land}\}\)), then
        \(\phi A = \phi B \diamond \phi C\). By induction hypothesis \(B, C\) are
        sharing-free if and only if so are \(\phi B, \phi C\), respectively. Moreover,
        because \(\phi\) is injective, \(\data{names}(\phi B) = \phi(\data{names}(B))\)
        and the same holds for~\(C\) \emph{mutatis mutandis}, \(B, C\) are disjoint
        if and only if so are \(\phi B, \phi C\). Then \(A\) is sharing-free, by definition,
        iff \(B, C\) are disjoint and both sharing-free, iff \(\phi B, \phi C\) are disjoint
        and both sharing-free, iff again by definition \(\phi A\) is sharing-free. \qedhere
    \end{itemize}
\end{proof}

\begin{lemma}
    \label{lemma:renaming-commutes-with-negation}
    Let \(A\) be any named formula:
    \begin{enumerate}[(i)]
        \item \(\phi\) is a renaming of~\(A\) if and only if it is a renaming of~\(\dl{A}\);
        \item for any renaming~\(\phi\) of~\(A\), \(\phi(\dl{A}) = \sdl{(\phi A)}\).
    \end{enumerate}
\end{lemma}
\begin{proof}
    Left to the reader.
\end{proof}

\begin{lemma}
    \label{lemma:sequent-renaming-coherence}
    Let \(\Gamma\) be a set of named formulas, \(\phi\) a renaming of~\(\Gamma\):
    \begin{enumerate}[(i)]
        \item for all \(A, B \in \Gamma\), \(A = B\) if and only if \(\phi A = \phi B\);
        \item for all \(A, B \in \Gamma\), \(A, B\) share names if and only if so do~\(\phi A, \phi B\);
        \item for all \(\Delta, \Delta' \subseteq \Gamma\), \(\Delta, \Delta'\) share
        names if and only if so do \(\phi\Delta, \phi\Delta'\).
    \end{enumerate}
\end{lemma}
\begin{proof}
    Point~(ii) is a special case of point~(iii), which is in turn an immediate consequence
    of the injectivity of~\(\phi\) plus the obvious facts that~\(\data{names}(\phi\Delta) =
    \phi(\data{names}(\Delta))\) and~\(\data{names}(\phi\Delta') = \phi(\data{names}(\Delta'))\).
    For point~(i) we proceed by structural induction on~\(A\):
    \begin{itemize}%[--]
        \item \(A\) is atomic: then \(A = B\) if and only if \(A = \nm{x}{\alpha} = B\)
        (resp.~\(A = \nm{x}{\dl{\alpha}} = B\)) for some~\(\alpha \in \mathcal{A}\)
        and~\(x \in \mathcal{N}\), if and only if \(\phi A = \nm{\phi x}{\alpha} = \phi B\)
        (resp.~\(\phi A = \nm{\phi x}{\dl{\alpha}} = \phi B\));

        \item \(A = C \diamond D\) where~\({\diamond} \in \{{\lor},{\land}\}\): then \(A = B\)
        if and only if \(B = C \diamond D\), if and only if \(\phi A = \phi C \lor \phi D =
        \phi B\). \qedhere
    \end{itemize}
\end{proof}

\begin{corollary}
    \label{coro:sequent-renaming-correctness}
    For all named sequents~\(\th \Gamma\) and renamings~\(\phi\) of~\(\th \Gamma\),
    \(\phi(\th \Gamma)\) is also a named sequent with \((\th \Gamma) \equiv \phi(\th \Gamma)\),
    and sharing-free iff so is \(\th \Gamma\).
\end{corollary}
\begin{proof}
    The fact that \(\phi(\th \Gamma)\) is a named sequent and sharing-free iff so is \(\th
    \Gamma\) follows from \cref{propo:formula-renaming-correctness,lemma:sequent-renaming-coherence},
    point~(ii). For the equivalence, observe that by \cref{lemma:sequent-renaming-coherence}, point~(i),
    \(\phi\) induces a bijection between \(\Gamma\) and~\(\phi\Gamma\) such that (again by
    \cref{propo:formula-renaming-correctness}) \(A \equiv \phi A\) for all~\(A \in \Gamma\).
\end{proof}

% \begin{lemma}
%     \label{lemma:disentanglement}
%     Let \(A\) be any named formula, \(X \subseteq \mathcal{N}\) a finite set of names. There
%     is a renaming~\(\phi\) of~\(A\) such that \(\data{names}(\phi A)\) and~\(X\) are disjoint.
% \end{lemma}
% \begin{proof}
%     Let us fix a complete enumeration \(x_1, x_2, \ldots\) \emph{without repetitions} of the
%     set~\(\mathcal{N}\) of names, which is countably infinite by assumption, and let \(k\)
%     be the smallest natural number such that \(x_i \notin X\) for all~\(i \geq k\). Such a~\(k\)
%     must exist because \(X\) is finite and the enumeration has no repetitions.

%     We define \(\phi\) by letting \(\phi(x_i) = x_{i+k}\) for all \(i \in \N\) such that
%     \(x_i \in \data{names}(A)\). \(\phi\) is well-defined and necessarily injective because
%     of the injectivity of addition to a constant term and the absence of repetitions in the
%     enumeration; moreover, \(\phi x \notin X\) for all \(x \in \data{names}(A)\), i.e.~%
%     \(\data{names}(\phi A) = \phi(\data{names}(A))\) and~\(X\) are disjoint.
% \end{proof}

% \begin{corollary}
%     \label{coro:naming}
%     For each anonymous propositional formula~\(F\) (resp.\ set~\(\Phi\) of anonymous
%     formulas), there is at least one named formula~\(A\) (resp.\ named sequent~\(\th \Gamma\))
%     such that~\(F = |A|\) (resp.~\(\Phi = |\Gamma|\)).
% \end{corollary}

Finally, we can rename whole derivation trees while preserving their structure and
conclusion (obviously up to renaming):

\begin{definition}
    \label{defn:derivation-renaming}
    A \emph{renaming} of a named derivation tree~\(P \in \data{GS4}^\mathcal{N}\) is an injective
    map \(\phi: \data{names}(P) \to \mathcal{N}\).
\end{definition}

The application of~\(\phi\) to~\(P\), noted~\(\phi P\), is the tree obtained by applying~\(\phi\)
recursively to the named sequents labeling each node of~\(P\) (plus the selected formulas in
axiom rules).

\begin{proposition}
    \label{propo:derivation-renaming-correctness}
    Let~\(P \in \data{GS4}^\mathcal{N}\) be any named derivation tree with conclusion~%
    \(\th \Gamma\), \(\phi\) a renaming of~\(P\): then \(\phi P \in \data{GS4}^\mathcal{N}\)
    is also a derivation tree with conclusion~\(\phi(\th \Gamma)\) and height~\(\data{h}(\phi P)
    = \data{h}(P)\), cut-free (resp. superposition-free) iff so is~\(P\).
\end{proposition}
\begin{proof}
    By a simple structural induction on~\(P\), using \cref{propo:formula-renaming-correctness}
    and \cref{coro:sequent-renaming-correctness} to prove that constraints on the shape
    of rules are satisfied by~\(\phi P\) at every step.
\end{proof}

\section{Proofs of lemmas regarding simple axiom graphs}\label{sec:axiom-graph-proofs}

\begin{proof}[Proof of \cref{propo:axiom-graphs-conclusion}]
    By structural induction on~\(P\):
    \begin{itemize}[--]
        \item if \(P\) has the form
        \[
            \prfbyaxiom{\rl[\{A, \dl{B}\}]{ax}}
                {\th \Delta, A, \dl{B}}
        \]
        where \(\Gamma = \Delta \cup \{A, \dl{B}\}\) and \(A \equiv B\) then by
        \cref{defn:wk-id-graphs,defn:axiom-graphs} \(V_{\axg{P}} = \data{names}(
        \Delta) \cup \data{names}(\{A, \dl{B}\}) = \data{names}(\Gamma)\);

        \item if \(P\) has the form
        \[
            \prftree[r]{\rl{cut}}
                {\prfsummary[\(Q\)]{\th \Gamma, A \vphantom{\dl{A}}}}
                {\prfsummary[\(R\)]{\th \Gamma, \dl{A}}}
                {\th \Gamma}
        \]
        then by induction hypothesis and \cref{defn:name-graph-composition}
        \[
            V_{\axg{P}} = (\data{names}(\Gamma, A) \cup \data{names}(\Gamma, \dl{A})) \setminus \data{names}(A);
        \]
        the statement follows from the fact that \(\data{names}(A) = \data{names}(\dl{A})\);
    \end{itemize}
    all other cases follow immediately from the induction hypothesis.
\end{proof}

\begin{proof}[Proof of \cref{lemma:axiom-graph-disjunction-transform}]
    By structural induction on~\(P\). For the sake of brevity we let~\(P' = \data{inv}(P,
    A \lor B)\) and omit at each step the shape of~\(P, P'\): the interested reader may
    inspect the corresponding steps in the proof of \cref{lemma:disjunction-transform}
    (\cref{sec:inversion-proofs}). Because \(\data{names}(\Gamma, A, B) = \data{names}(
    \Gamma, A \lor B)\), we have in any case \(V_{\axg{P'}} = V_{\axg{P}}\) and it is
    enough to show that \(E_{\axg{P'}} = E_{\axg{P}}\):
    \begin{itemize}[--]
        \item if \(P\) ends with an axiom rule application of the kind \(\data{ax}_{\{\dl{C}
        \land \dl{D}, A \lor B\}}\) with conclusion \(\th \Delta, \dl{C} \land \dl{D}, A \lor B\)
        where \(\Gamma = \Delta \cup \{\dl{C} \land \dl{D}\}\), then by \cref{defn:axiom-graphs}
        \[
            \axg{P'} = (\data{Wk}_{\Delta, B} \sqcup \data{Id}_{\{\dl{C}, A\}})
                    \sqcup (\data{Wk}_{\Delta, A} \sqcup \data{Id}_{\{\dl{D}, B\}})
        \]
        and by \cref{defn:wk-id-graphs} we have
        \begin{align*}
            \axg{P'} &= \data{Wk}_{\Delta} \sqcup \data{Id}_{\{\dl{C}, A\}}
                        \sqcup \data{Id}_{\{\dl{D}, B\}}\\
                     &= \data{Wk}_{\Delta} \sqcup \data{Id}_{\{\dl{C} \land \dl{D}, A \lor B\}}
                        = \axg{P};
        \end{align*}

        \item if \(P\) ends with an axiom rule application of the kind \(\data{ax}_{\{C,\dl{D}\}}\)
        with conclusion \(\th \Delta, A \lor B, C, \dl{D}\) where \(\Gamma = \Delta \cup
        \{C, \dl{D}\}\), then by \cref{defn:axiom-graphs,defn:wk-id-graphs}
        \[
            \axg{P'} = \data{Wk}_{\Delta, A, B} \sqcup \data{Id}_{\{C, \dl{D}\}}
                = \data{Wk}_{\Delta, A \lor B} \sqcup \data{Id}_{\{C, \dl{D}\}} = \axg{P};
        \]

        \item if \(P\) ends with a \rl{\lor}-rule application introducing \(A \lor B\)
        with premiss subtree~\(Q\), then \(P' = Q\) and \(\axg{P'} = \axg{Q} = \axg{P}\);

        \item if \(P\) ends with a \rl{cut}-rule application on formulas \(C, \dl{C}\),
        with premiss subderivations \(Q, R\), we have by induction hypothesis \(\axg{\data{inv}(
        Q, A \lor B)} = \axg{Q}\) and~\(\axg{\data{inv}(R, A \lor B)} = \axg{R}\), hence
        \[
            \axg{P'} = \axg{\data{inv}(Q, A \lor B)} \odot_C \axg{\data{inv}(R, A \lor B)}
                = \axg{Q} \odot_C \axg{R} = \axg{P};
        \]

        \item if \(P\) ends with a \rl{sum}-, \rl{\lor}- or \rl{\land}-rule
        application that does not introduce \(A \lor B\), with premiss subtrees~\(Q_1, \ldots, Q_n\),
        we have \(\axg{\data{inv}(Q_i, A \lor B)} = \axg{Q_i}\) by induction hypothesis
        for all~\(1 \leq i \leq n\): then
        \begin{multline*}
            \axg{P'} = \axg{\data{inv}(Q_1, A \lor B)} \sqcup \ldots \sqcup \axg{\data{inv}(Q_n, A \lor B)}\\
                = \axg{Q_1} \sqcup \ldots \sqcup \axg{Q_n} = \axg{P}.
        \end{multline*}\qedhere
    \end{itemize}
\end{proof}

\begin{proof}[Proof of \cref{lemma:axiom-graph-conjunction-transforms}]
    We only prove the first part of the statement, i.e.\ that
    \[
        \axg{\data{inv_l}(P, A \land B)} \sqsubseteq \axg{P} \rst_{\Gamma, A};
    \]
    the argument for the second part is analogous. Let \(P' = \data{inv_l}(P, A \land B)\),
    As in the previous proof, we do not recall the shape of \(P, P'\); the reader is invited
    to check our statements against the proofs of \cref{lemma:left-conjunction-transform,%
    lemma:right-conjunction-transform} (\cref{sec:inversion-proofs}). Observe first that by
    \cref{propo:axiom-graphs-conclusion} we have
    \[
        V_{\axg{P'}} = \data{names}(\Gamma, A) = V_{\axg{P} \rst_{\Gamma, A}}.
    \]
    We then need to prove only that \(E_{\axg{P'}} \subseteq E_{\axg{P}}\), as this
    will be enough to guarantee that \(E_{\axg{P'}} \subseteq E_{\axg{P} \rst_{\Gamma, A}}\).
    We proceed by structural induction on~\(P\):
    \begin{itemize}[--]
        \item if \(P\) ends with an axiom rule application of the kind \(\data{ax}_{\{\dl{C}
        \lor \dl{D}, A \land B\}}\) with conclusion \(\th \Delta, \dl{C} \lor \dl{D}, A \land B\)
        where \(\Gamma = \Delta \cup \{\dl{C} \lor \dl{D}\}\), then by \cref{defn:axiom-graphs,%
        defn:wk-id-graphs}
        \begin{multline*}
            \axg{P'} = \data{Wk}_{\Delta, \dl{D}} \sqcup \data{Id}_{\{\dl{C}, A\}}\\
                \sqsubseteq \data{Wk}_\Delta \sqcup \data{Id}_{\{\dl{C}, A\}} \sqcup \data{Id}_{\{\dl{D}, B\}}
                    = \data{Wk}_\Delta \sqcup \data{Id}_{\{\dl{C} \lor \dl{D}, A \land B\}} = \axg{P};
        \end{multline*}

        \item if \(P\) ends with an axiom rule application of the kind \(\data{ax}_{\{C,\dl{D}\}}\)
        with conclusion \(\th \Delta, A \land B, C, \dl{D}\) where \(\Gamma = \Delta \cup
        \{C, \dl{D}\}\), then by \cref{defn:axiom-graphs,defn:wk-id-graphs}
        \[
            \axg{P'} = \data{Wk}_{\Delta, A} \sqcup \data{Id}_{\{C, \dl{D}\}}
                \sqsubseteq \data{Wk}_{\Delta, A \land B} \sqcup \data{Id}_{\{C, \dl{D}\}} = \axg{P};
        \]

        \item if \(P\) ends with a \rl{\land}-rule application introducing \(A \land B\)
        whose premiss subtrees~\(Q, R\) have conclusions respectively \(\th \Gamma, A\)
        and~\(\th \Gamma, B\), then \(P' = Q\) and the conclusion is immediate;

        \item if \(P\) ends with a \rl{cut}-rule application on formulas \(C, \dl{C}\),
        with premiss subderivations \(Q, R\), we have by induction hypothesis \(\axg{
        \data{inv_l}(Q, A \land B)} \sqsubseteq \axg{Q}\) and~\(\axg{\data{inv_l}(R, A \land B)}
        \sqsubseteq \axg{R}\), with
        \[
            \axg{P'} = \axg{\data{inv_l}(Q, A \land B)} \odot_C \axg{\data{inv_l}(R, A \land B)}.
        \]
        By \cref{defn:name-graph-composition}, \(xy \in E_{\axg{P'}}\) if and only if there is
        an alternating path (\cref{defn:simple-alternating-path}) connecting~\(x\) with~\(y\)
        between \(\axg{\data{inv_l}(Q, A \land B)}\) and \(\axg{\data{inv_l}(R, A \land B)}\)
        on inteface~\(\data{names}(C)\). It is easy to check that the same sequence of vertices
        is an alternating path between \(\axg{Q}\) and~\(\axg{R}\) through interface~\(\data{names}(C)\),
        and therefore~\(xy \in E_{\axg{Q} \odot_C \axg{R}} = E_{\axg{P}}\);

        \item if \(P\) ends with a \rl{sum}-, \rl{\lor}- or \rl{\land}-rule
        application that does not introduce \(A \land B\), we conclude by induction hypothesis
        and \cref{propo:restriction-distributes-over-union}. \qedhere
    \end{itemize}
\end{proof}

\begin{proof}[Proof of \cref{propo:axiom-graphs-isl-cut-free-invariance}]
    Immediate by \cref{lemma:axiom-graph-disjunction-transform} when \(A\) is a
    disjunction; for conjunctions we need to show that
    \[
        \axg{P} = \axg{\data{inv_l}(P, A)} \sqcup \axg{\data{inv_r}(P, A)},
    \]
    whenever \(P\) is cut-free. In particular, it suffices to show that
    \[
        E_{\axg{P}} = E_{\axg{\data{inv_l}(P, A)}} \cup E_{\axg{\data{inv_r}(P, A)}}.
    \]
    By structural induction on~\(P\), assuming \(A = A_1 \land A_2\):
    \begin{itemize}%[--]
        \item if \(P\) ends with an axiom rule application of the kind \(\data{ax}_{\{\dl{B}_1
        \lor \dl{B}_2, A\}}\) with conclusion \(\th \Delta, \dl{B}_1 \lor \dl{B}_2, A\)
        where \(\Gamma = \Delta \cup \{\dl{B}_1 \lor \dl{B}_2\}\), then
        \begin{gather*}
            \axg{P} = \data{Wk}_\Delta \sqcup \data{Id}_{\{\dl{B}_1 \lor \dl{B}_2, A\}}
                = \data{Wk}_\Delta \sqcup \data{Id}_{\{\dl{B}_1, A_1\}} \sqcup \data{Id}_{\{\dl{B}_2, A_2\}},\\
            \axg{\data{inv_l}(P, A)} = \data{Wk}_{\Delta, \dl{B}_2} \sqcup \data{Id}_{\{\dl{B}_1, A_1\}},\\
            \axg{\data{inv_r}(P, A)} = \data{Wk}_{\Delta, \dl{B}_1} \sqcup \data{Id}_{\{\dl{B}_2, A_2\}}.
        \end{gather*}
        Because the weakening graphs have no edges, we can conclude that
        \[
            E_{\axg{P}} = E_{\data{Id}_{\{\dl{B}_1, A_1\}}} \sqcup E_{\data{Id}_{\{\dl{B}_2, A_2\}}}
                = E_{\axg{\data{inv_l}(P, A)}} \sqcup E_{\axg{\data{inv_r}(P, A)}}.
        \]

        \item if \(P\) ends with an axiom rule application of the kind \(\data{ax}_{\{B,\dl{C}\}}\)
        with conclusion \(\th \Delta, A, B, \dl{C}\) where \(\Gamma = \Delta \cup \{B, \dl{C}\}\),
        then we have
        \begin{gather*}
            \axg{P} = \data{Wk}_\Delta \sqcup \data{Id}_{\{B, \dl{C}\}},\\
            \axg{\data{inv_l}(P, A)} = \data{Wk}_{\Delta, A_1} \sqcup \data{Id}_{\{B, \dl{C}\}},\\
            \axg{\data{inv_r}(P, A)} = \data{Wk}_{\Delta, A_2} \sqcup \data{Id}_{\{B, \dl{C}\}}.
        \end{gather*}
        Therefore, because the weakening graphs have no edges,
        \[
            E_{\axg{P}} = E_{\data{Id}_{\{B, \dl{C}\}}} = E_{\axg{\data{inv_l}(P, A)}} = E_{\axg{\data{inv_r}(P, A)}},
        \]
        from which the conclusion follows.

        \item if \(P\) ends with a \rl{\land}-rule application introducing \(A\),
        the conclusion is trivial as \(P = \data{isl}(P, A)\);

        \item if \(P\) ends with a \rl{\sqcup}-, \rl{\lor}- or \rl{\land}-rule
        application that does not introduce \(A\), with premiss subtrees~\(Q_1, \ldots, Q_n\),
        we have by induction hypothesis
        \[
            E_{\axg{Q_i}} = E_{\axg{\data{inv_l}(Q_i, A)}} \cup E_{\axg{\data{inv_r}(Q_i, A)}}.
        \]
        for all \(1 \leq i \leq n\). The result then follows from the fact that
        \[
            E_{\axg{P}} = \bigcup_{i = 1}^n E_{\axg{Q_i}}
                = \left(\bigcup_{i = 1}^n E_{\axg{\data{inv_l}(Q_i, A)}}\right)
                    \cup \left(\bigcup_{i = 1}^n E_{\axg{\data{inv_r}(Q_i, A)}}\right). \qedhere
        \]
    \end{itemize}
\end{proof}

\section{Proofs of lemmas regarding branch-labeled axiom graphs}\label{sec:bl-axiom-graph-proofs}

\subsection{Properties of branch sets}

\begin{lemma}
    \label{lemma:branches-in-names}
    For all named formulas~\(A\), if \(X \in \data{Br}(A)\) then \(X \subseteq \data{names}(A)\).
\end{lemma}
\begin{proof}
    Immediate by induction on~\(A\).
\end{proof}

\begin{proof}[Proof of \cref{lemma:sequent-branches-alt}]
    Observe first that \(\data{names}(\Gamma) = \bigcup_{A \in \Gamma} \data{names}(A)\),
    and assume \(X \in \data{Br}(\Gamma)\). By definition \(X \subseteq \data{names}(\Gamma)\),
    therefore \(X = \bigcup_{A \in \Gamma} (X \cap \data{names}(A))\); we know also that
    \((X \cap \data{names}(A)) \in \data{Br}(A)\) for all~\(A \in \Gamma\), hence we can
    choose~\(X_A = X \cap \data{names}(A)\).

    Conversely, assume \(X = \bigcup_{A \in \Gamma} X_A\) with \(X_A \in \data{Br}(A)\)
    for all~\(A \in \Gamma\). By \cref{lemma:branches-in-names}, \(X \subseteq \data{names}(A)\);
    moreover, because \(\Gamma\) is sharing-free, its elements have pairwise disjoint name
    sets, hence the \(X_A\) in particular are pairwise disjoint and \(X \cap \data{names}(A)
    = X_A \in \data{Br}(A)\), as desired.
\end{proof}

\begin{proof}[Proof of \cref{lemma:branches-of-union}]
    Assume \(X \in \data{Br}(\Gamma \cup \Delta)\): by \cref{lemma:sequent-branches-alt} there
    is an indexed family \((X_A)_{A \in \Gamma \cup \Delta}\) of name sets such that \(X =
    \bigcup_{A \in \Gamma \cup \Delta} X_A\) and \(X_A \in \data{Br}(A)\) for all~\(A \in \Gamma
    \cup \Delta\). Let then \(Y = \bigcup_{A \in \Gamma} X_A\) and~\(Z = \bigcup_{A \in \Delta}
    X_A\): again by \cref{lemma:sequent-branches-alt}, \(Y \in \data{Br}(\Gamma)\) and~\(Z
    \in \data{Br}(\Delta)\), and obviously \(X = Y \cup Z\). The argument for the converse is
    analogous, reversing the direction of all implications.
\end{proof}

\begin{proof}[Proof of \cref{propo:sequent-branches-decomposition}]
    Fact~(i) is left to the reader. For fact~(ii), observe that \(\data{Br}(\{A, B\}) =
    \data{Br}(\{A\} \cup \{B\}) = \data{Br}(A \lor B)\) as an immediate consequence of
    \cref{lemma:branches-of-union}: then we have
    \begin{align*}
        \data{Br}(\Gamma, A \lor B) &= \{ X \cup Y \mid X \in \data{Br}(\Gamma), Y \in \data{Br}(A \lor B) \}\\
            &= \{ X \cup Y \mid X \in \data{Br}(\Gamma), Y \in \data{Br}(\{A, B\}) \}\\
            &= \data{Br}(\Gamma, A, B).
    \end{align*}
    For fact~(iii), recall that \(\data{Br}(A \land B) = \data{Br}(A) \cup \data{Br}(B)\);
    then, again by \cref{lemma:branches-of-union},
    \begin{align*}
        \data{Br}(\Gamma, A \land B) &= \{ X \cup Y \mid X \in \data{Br}(\Gamma), Y \in \data{Br}(A \land B) \}\\
            &= \{ X \cup Y \mid X \in \data{Br}(\Gamma), Y \in \data{Br}(A) \cup \data{Br}(B) \}\\
            &= \{ X \cup Y \mid X \in \data{Br}(\Gamma), Y \in \data{Br}(A) \}\\
            & \qquad\cup \{ X \cup Y \mid X \in \data{Br}(\Gamma), Y \in \data{Br}(B) \}\\
            &= \data{Br}(\Gamma, A) \cup \data{Br}(\Gamma, B).
    \end{align*}
    For fact~(iv), observe that (by an easy inductive argument) all branches of~\(A\)
    (resp. of~\(B\)) must be non-empty, and remember that \(\Gamma, A, B\) have disjoint
    name sets by hypothesis. Let then \(X \in \data{Br}(\Gamma, A) \cap \data{Br}(\Gamma, B)\):
    by \cref{lemma:branches-of-union} there is a \emph{non-empty} \(Y \in \data{Br}(A)\)
    such that \(Y \subseteq X\). By construction we have also \(X \subseteq \data{names}(
    \Gamma, B)\), but then there should be some name \(z \in Y \subseteq \data{names}(A)\)
    such that \(z \in \data{names}(\Gamma, B)\), which is impossible.
\end{proof}

\subsection{Properties of branch-labeled axiom graphs}

\begin{lemma}
    \label{lemma:bl-wk-branches}
    For all bl-graphs~\(G\),
    \[
        \data{Br}(\blwk{\Gamma}(G)) = \{ X \cup Y \mid X \in \data{Br}(G), Y \in \data{Br}(\Gamma) \}.
    \]
\end{lemma}
\begin{proof}
    Immediate, see \cref{defn:bl-wk-id}.
\end{proof}

\begin{lemma}
    \label{lemma:bl-id-branches}
    For any pair~\(A, \dl{B}\) of disjoint and sharing-free named formulas such that
    \(A \equiv B\), \(\data{Br}(\blid{\{A, \dl{B}\}}) = \data{Br}(A, \dl{B})\).
\end{lemma}
\begin{proof}
    By induction on the height of~\(A\):
    \begin{itemize}
        \item if \(A\) is atomic, then \(\data{Br}(A, \dl{B}) = \{\data{names}(A, \dl{B})\}
        = \data{Br}(\blid{\{A, \dl{B}\}})\) by \cref{propo:sequent-branches-decomposition,%
        defn:bl-wk-id};

        \item if \(A = A_1 \lor A_2\) and~\(\dl{B} = \dl{B}_1 \land \dl{B}_2\), then
        by \cref{defn:bl-wk-id}
        \begin{align*}
            \data{Br}(\blid{\{A, \dl{B}\}})
                &= \data{Br}(\blwk{A_2}(\blid{\{A_1, \dl{B}_1\}}) \sqcup \blwk{A_1}(\blid{\{A_2, \dl{B}_2\}})) \\
                &= \data{Br}(\blwk{A_2}(\blid{\{A_1, \dl{B}_1\}})) \cup \data{Br}(\blwk{A_1}(\blid{\{A_2, \dl{B}_2\}})).
        \end{align*}
        By induction hypothesis, \(\data{Br}(\blid{\{A_1, \dl{B}_1\}}) = \data{Br}(A_1, \dl{B}_1)\)
        and \(\data{Br}(\blid{\{A_2, \dl{B}_2\}}) = \data{Br}(A_2, \dl{B}_2)\);
        by \cref{lemma:bl-wk-branches,lemma:branches-of-union}
        \begin{gather*}
            \begin{multlined}
                \data{Br}(\blwk{A_2}(\blid{\{A_1, \dl{B}_1\}}))\\
                    = \{ X \cup Y \mid X \in \data{Br}(A_1, \dl{B}_1), Y \in \data{Br}(A_2) \}
                    = \data{Br}(A_1, A_2, \dl{B}_1),
            \end{multlined}\\
            \begin{multlined}
                \data{Br}(\blwk{A_1}(\blid{\{A_2, \dl{B}_2\}}))\\
                    = \{ X \cup Y \mid X \in \data{Br}(A_2, \dl{B}_2), Y \in \data{Br}(A_1) \}
                    = \data{Br}(A_1, A_2, \dl{B}_2).
            \end{multlined}
        \end{gather*}
        Then, again by \cref{propo:sequent-branches-decomposition},
        \begin{align*}
            \data{Br}(\blid{\{A, \dl{B}\}})
                &= \data{Br}(A_1, A_2, \dl{B}_1) \cup \data{Br}(A_1, A_2, \dl{B}_2)\\
                &= \data{Br}(A_1, A_2, \dl{B}_1 \land \dl{B}_2)\\
                &= \data{Br}(A_1 \lor A_2, \dl{B}_1 \land \dl{B}_2)
                    = \data{Br}(A, \dl{B}). \qedhere
        \end{align*}
    \end{itemize}
\end{proof}

\begin{corollary}
    \label{coro:bl-axiom-branches}
    Let \(\Gamma\) be any sharing-free set of named formulas, \(A, \dl{B}\) a pair of
    disjoint and sharing free named formulas that share no names with \(\Gamma\) and
    such that \(A \equiv B\): then
    \[
        \data{Br}(\blwk{\Gamma}(\blid{\{A, \dl{B}\}})) = \data{Br}(\Gamma, A, \dl{B}).
    \]
\end{corollary}

\begin{proof}[Proof of \cref{propo:bl-axiom-graph-names-branches}]
    The proof that \(V_{\blaxg{P}} = \data{names}(\Gamma)\) is analogous to that of
    \cref{propo:axiom-graphs-conclusion} and left to the reader. We prove by structural
    induction on~\(P\) that \(\data{Br}(\blaxg{P}) \subseteq \data{Br}(\Gamma)\):
    \begin{itemize}
        \item if \(P\) is an axiom rule application, the conclusion follows immediately
        from \cref{coro:bl-axiom-branches};

        \item if \(P\) has the form
        \[
            \prftree[r]{\rl{cut}}
                {\prfsummary[\(Q\)]{\th \Gamma, A \vphantom{\dl{A}}}}
                {\prfsummary[\(R\)]{\th \Gamma, \dl{A}}}
                {\th \Gamma}
        \]
        then by \cref{defn:bl-composition}, if \(X \in \data{Br}(\blaxg{Q} \odot_A \blaxg{R})\)
        there is either \(Y \in \data{Br}(\blaxg{Q})\) or~\(Y \in \data{Br}(\blaxg{R})\)
        such that \(X = Y \setminus \data{names}(A)\). By induction hypothesis we have either
        \(Y \in \data{Br}(\Gamma, A)\) or~\(Y \in \data{Br}(\Gamma, \dl{A})\), hence
        \(X \in \data{Br}(\Gamma)\) by \cref{coro:branches-of-difference};
    \end{itemize}
    all other cases follow immediately from the induction hypothesis.
\end{proof}

\subsection{Behaviour under inversion}

\begin{lemma}
    \label{lemma:bl-wk-distribution}
    For all bl-graphs~\(G, H\) and sharing-free sets~\(\Gamma\) of named formulas,
    \[
        \blwk{\Gamma}(G \sqcup H) = \blwk{\Gamma}(G) \sqcup \blwk{\Gamma}(H).
    \]
\end{lemma}
\begin{proof}
    Left to the reader.
\end{proof}

\begin{lemma}
    \label{lemma:bl-wk-fusion}
    For all bl-graphs~\(G\) and sharing-free sets~\(\Gamma \cup \Delta\) of named formulas,
    \[
        \blwk{\Gamma}(\blwk{\Delta}(G)) = \blwk{\Gamma \cup \Delta}(G).
    \]
\end{lemma}
\begin{proof}
    The identity is obvious for vertices. For the edge-branch relation, let
    \(e \adin_{\blwk{\Gamma}(\blwk{\Delta}(G))} X\): there are then \(Y \in \data{Br}(G)\),
    \(Z \in \data{Br}(\Delta)\) and~\(W \in \data{Br}(\Gamma)\) such that \(e \adin_G Y\)
    and~\(X = Y \cup Z \cup W\). By \cref{lemma:branches-of-union}, \(Z \cup W \in
    \data{Br}(\Gamma \cup \Delta)\), hence \(e \adin_{\blwk{\Gamma \cup \Delta}(G)} X\).
    Conversely, let \(e \adin_{\blwk{\Gamma \cup \Delta}(G)} X\): there are \(Y \in
    \data{Br}(G)\), \(Z \in \data{Br}(\Gamma \cup \Delta)\) such that \(e \adin_G Y\)
    and~\(X = Y \cup Z\). Again by \cref{lemma:branches-of-union}, \(Z = Z' \cup Z''\)
    for some \(Z' \in \data{Br}(\Gamma), Z'' \in \data{Br}(\Delta)\), hence \(e \adin_{
    \blwk{\Gamma}(\blwk{\Delta}(G))} X\).
\end{proof}

\begin{proof}[Proof of \cref{lemma:bl-axiom-graph-disjunction-transform}]
    Almost identical to the proof of \cref{lemma:axiom-graph-disjunction-transform}
    (\cref{sec:axiom-graph-proofs}); only the base case is different (axioms).
    Let \(P'\) denote the inverted derivation~\(\data{inv}(P, A \lor B)\); then:
    \begin{itemize}
        \item if \(P\) ends with an axiom rule application of the kind \(\data{ax}_{\{\dl{C}
        \land \dl{D}, A \lor B\}}\) with conclusion \(\th \Delta, \dl{C} \land \dl{D}, A \lor B\)
        where \(\Gamma = \Delta \cup \{\dl{C} \land \dl{D}\}\), then by \cref{defn:bl-axiom-graphs}
        we have
        \begin{gather*}
            \blaxg{P} = \blwk{\Delta}(\blid{\{\dl{C} \land \dl{D}, A \lor B\}}),\\
            \blaxg{P'} = \blwk{\Delta, B}(\blid{\{\dl{C}, A\}})
                    \sqcup \blwk{\Delta, A}(\blid{\{\dl{D}, B\}}),
        \end{gather*}
        and by \cref{defn:bl-wk-id}
        \[
            \blaxg{P} = \blwk{\Delta}(\blwk{B}(\blid{\{\dl{C}, A\}}) \sqcup \blwk{A}(\blid{\{\dl{D}, B\}})).
        \]
        We apply \cref{lemma:bl-wk-distribution,lemma:bl-wk-fusion} to get
        \begin{align*}
            \blaxg{P} &= \blwk{\Delta}(\blwk{B}(\blid{\{\dl{C}, A\}}))
                    \sqcup \blwk{\Delta}(\blwk{A}(\blid{\{\dl{D}, B\}}))\\
                &= \blwk{\Delta, B}(\blid{\{\dl{C}, A\}})
                    \sqcup \blwk{\Delta, A}(\blid{\{\dl{D}, B\}}) = \blaxg{P'}.
        \end{align*}

        \item if \(P\) ends with an axiom rule application of the kind \(\data{ax}_{\{C,\dl{D}\}}\)
        with conclusion \(\th \Delta, A \lor B, C, \dl{D}\) where \(\Gamma = \Delta \cup
        \{C, \dl{D}\}\), then by \cref{defn:bl-axiom-graphs,propo:sequent-branches-decomposition}
        \[
            \blaxg{P'} = \blwk{\Delta, A, B}(\blid{\{C, \dl{D}\}})
                = \blwk{\Delta, A \lor B}(\blid{\{C, \dl{D}\}}) = \blaxg{P}. \qedhere
        \]
    \end{itemize}
\end{proof}

\begin{lemma}
    \label{lemma:restriction-inclusion}
    For all bl-graphs~\(G \sqsubseteq H\) and name sets~\(X \subseteq \mathcal{N}\),
    \[
        V_G \subseteq X \implies G \sqsubseteq H \rst_X.
    \]
\end{lemma}
\begin{proof}
    If \(V_G \subseteq X\) then obviously \(V_G \subseteq V_{H \rst_X} = (V_H \cap X)\).
    Moreover, let \(e \adin_G Y\): we have \(Y \subseteq V_G \subseteq X\) by \cref{defn:bl-graphs}
    and we know by hypothesis that \(e \adin_H Y\), hence \(e \adin_{H \rst_X} Y\).
\end{proof}

\begin{lemma}
    \label{lemma:conj-restriction-branches}
    Let \(\Gamma \cup \{A \land B\}\) be a sharing-free set of named formulas, \(G\)
    any bl-graph with \(\data{Br}(G) \subseteq \data{Br}(\Gamma, A \land B)\): then
    \[
        \data{Br}(G \rst_{\Gamma, A}) \subseteq \data{Br}(\Gamma, A)
        \text{ and }
        \data{Br}(G \rst_{\Gamma, B}) \subseteq \data{Br}(\Gamma, B).
    \]
\end{lemma}
\begin{proof}
    We only prove the first half of the statement, the second half is similar. Observe
    that (by an easy inductive argument) every branch of~\(B\) is non-empty, and remember
    that \(\Gamma, A, B\) have disjoint name sets. Let \(X \in \data{Br}(G \rst_{\Gamma, A})\)
    and assume \(X \in \data{Br}(\Gamma, B)\): by \cref{lemma:branches-of-union} there
    are \(Y \in \data{Br}(\Gamma)\) and a \emph{non-empty} \(Z \in \data{Br}(B)\) such that
    \(X = Y \cup Z\). However, we have \(X \subseteq \data{names}(\Gamma, A)\) by construction,
    hence \(Z \subseteq \data{names}(\Gamma, A)\), which is disjoint from~\(\data{names}(B)\),
    a contradiction. Therefore \(X \notin \data{Br}(\Gamma, B)\), and by exclusion we
    have \(X \in \data{Br}(\Gamma, A)\).
\end{proof}

\begin{proof}[Proof of \cref{lemma:bl-axiom-graph-conjunction-transforms}]
    We only prove the first part of the statement, i.e.\ that
    \[
        \blaxg{\data{inv_l}(P, A \land B)} = \blaxg{P} \rst_{\Gamma, A};
    \]
    the argument for the second part is analogous. Let \(P' = \data{inv_l}(P, A \land B)\),
    As usual we reason about the shape of~\(P, P'\) without recalling it; the interested
    reader may inspect the proofs of \cref{lemma:left-conjunction-transform,lemma:right-conjunction-transform}
    (\cref{sec:inversion-proofs}). Observe first that by \cref{propo:bl-axiom-graph-names-branches}
    we have
    \[
        V_{\blaxg{P'}} = \data{names}(\Gamma, A) = V_{\blaxg{P} \rst_{\Gamma, A}};
    \]
    by \cref{lemma:restriction-inclusion}, then, we can prove that \(\blaxg{P'} \sqsubseteq
    \blaxg{P} \rst_{\Gamma, A}\) by showing that \(\blaxg{P'} \sqsubseteq \blaxg{P}\).
    For the reverse inclusion, on the other hand, it suffices to show that \({\adin_{\blaxg{P}
    \rst_{\Gamma, A}}} \subseteq {\adin_{\blaxg{P'}}}\). By structural induction on~\(P\):
    \begin{itemize}[--]
        \item if \(P\) ends with an axiom rule application of the kind \(\data{ax}_{\{\dl{C}
        \lor \dl{D}, A \land B\}}\) with conclusion \(\th \Delta, \dl{C} \lor \dl{D}, A \land B\)
        where \(\Gamma = \Delta \cup \{\dl{C} \lor \dl{D}\}\), then by \cref{defn:bl-axiom-graphs}
        \[
            \blaxg{P} = \blwk{\Delta}(\blid{\{\dl{C} \lor \dl{D}, A \land B\}})
            \text{ and }
            \blaxg{P'} = \blwk{\Delta, \dl{D}}(\blid{\{\dl{C}, A\}}).
        \]
        By \cref{defn:bl-wk-id,lemma:bl-wk-distribution,lemma:bl-wk-fusion}
        \begin{align*}
            \blaxg{P} &= \blwk{\Delta}(\blwk{\dl{D}}(\blid{\{\dl{C}, A\}})
                    \sqcup \blwk{\dl{C}}(\blid{\{\dl{D}, B\}}))\\
                &= \blwk{\Delta, \dl{D}}(\blid{\{\dl{C}, A\}}) \sqcup \blwk{\Delta, \dl{C}}(\blid{\{\dl{D}, B\}})\\
                &= \blaxg{P'} \sqcup \blwk{\Delta, \dl{C}}(\blid{\{\dl{D}, B\}}),
        \end{align*}
        hence \(\blaxg{P'} \sqsubseteq \blaxg{P}\), and we have also
        \[
            \data{Br}(\blwk{\Delta, \dl{C}}(\blid{\{\dl{D}, B\}})) = \data{Br}(\Delta, \dl{C}, \dl{D}, B) = \data{Br}(\Gamma, B)
        \]
        by \cref{coro:bl-axiom-branches,propo:sequent-branches-decomposition}. Let then
        \(e \adin_{\blaxg{P} \rst_{\Gamma, A}} X\): by \cref{lemma:conj-restriction-branches}
        \(X \in \data{Br}(\Gamma, A)\), hence by \cref{propo:sequent-branches-decomposition},
        point~(iv)
        \[
            X \notin \data{Br}(\Gamma, B) = \data{Br}(\blwk{\Delta, \dl{C}}(\blid{\{\dl{D}, B\}})),
        \]
        and then by exclusion we must have \(e \adin_{\blaxg{P'}} X\);

        \item if \(P\) ends with an axiom rule application of the kind \(\data{ax}_{\{C,\dl{D}\}}\)
        with conclusion \(\th \Delta, A \land B, C, \dl{D}\) where \(\Gamma = \Delta \cup
        \{C, \dl{D}\}\), then by \cref{defn:axiom-graphs}
        \[
            \blaxg{P} = \blwk{\Delta, A \land B}(\blid{\{C, \dl{D}\}})
            \text{ and }
            \blaxg{P'} = \blwk{\Delta, A}(\blid{\{C, \dl{D}\}}).
        \]
        It is easy to check, using \cref{defn:bl-wk-id,propo:sequent-branches-decomposition,%
        coro:bl-axiom-branches}, that
        \[
            \blaxg{P} = \blwk{\Delta, A}(\blid{\{C, \dl{D}\}})
                    \sqcup \blwk{\Delta, B}(\blid{\{C, \dl{D}\}})
                = \blaxg{P'} \sqcup \blwk{\Delta, B}(\blid{\{C, \dl{D}\}})
        \]
        with
        \[
            \data{Br}(\blwk{\Delta, B}(\blid{\{C, \dl{D}\}})) = \data{Br}(\Delta, B, C, \dl{D}) = \data{Br}(\Gamma, B),
        \]
        and we conclude as in the previous case;

        \item if \(P\) ends with a \rl{\land}-rule application introducing \(A \land B\)
        whose premiss subtrees~\(Q, R\) have conclusions respectively \(\th \Gamma, A\)
        and~\(\th \Gamma, B\), then \(P' = Q\). We have obviously \(\blaxg{P'} \sqsubseteq
        \blaxg{Q} \sqcup \blaxg{R} = \blaxg{P}\). For the reverse inclusion, we argue
        by exclusion from \cref{lemma:conj-restriction-branches,propo:sequent-branches-decomposition}
        and the fact that \(\data{Br}(\blaxg{R}) \subseteq \data{Br}(\Gamma, B)\),
        concluding that \({\adin_{\blaxg{P} \rst_{\Gamma, A}}} \subseteq {\adin_{\blaxg{Q}}}
        = {\adin_{\blaxg{P'}}}\);

        \item if \(P\) ends with a \rl{cut}-rule application on formulas \(C, \dl{C}\),
        with premiss subderivations \(Q, R\), we have by induction hypothesis
        \[
            \blaxg{\data{inv_l}(Q, A \land B)} = \blaxg{Q} \rst_{\Gamma, A, C}
            \text{ and }
            \blaxg{\data{inv_l}(R, A \land B)} = \blaxg{R} \rst_{\Gamma, A, \dl{C}},
        \]
        hence
        \[
            \blaxg{P'} = {\blaxg{Q} \rst_{\Gamma, A, C}} \odot_C {\blaxg{R} \rst_{\Gamma, A, \dl{C}}}.
        \]
        Let \(e \adin_{\blaxg{P'}} X\): there is by \cref{defn:bl-composition} an \(X\)-labeled
        alternating path \(z_1,\ldots,z_n\) between \(\blaxg{Q} \rst_{\Gamma, A, C}\)
        and~\(\blaxg{R} \rst_{\Gamma, A, \dl{C}}\) through interface~\(\data{names}(C)\),
        with \(e = z_1 z_n\). This is obviously also an \(X\)-labeled alternating path
        between \(\blaxg{Q}\) and~\(\blaxg{R}\) through the same interface, hence
        \(e \adin_{\blaxg{P}} X\) and~\(\blaxg{P'} \sqsubseteq \blaxg{P}\).

        Now let \(e \adin_{\blaxg{P} \rst_{\Gamma, A}} X\): there is, again by \cref{defn:bl-composition},
        an \(X\)-labeled alternating path \(z_1,\ldots,z_n\) between \(\blaxg{Q}\)
        and~\(\blaxg{R}\) through interface~\(\data{names}(C)\) with \(e = z_1 z_n\). We have
        to show that this is also an \(X\)-labeled alternating path between \(\blaxg{Q}
        \rst_{\Gamma, A, C}\) and~\(\blaxg{R} \rst_{\Gamma, A, \dl{C}}\) through the same
        interface, from which it follows that \(e \adin_{\blaxg{P'}} X\). Let~\(I =
        \data{names}(C)\) and assume \(z_i z_{i+1} \adin_{\blaxg{Q}}^I X\) for some~\(1
        \leq i < n\): there is \(Y \subseteq I\) such that \(z_i z_{i+1} \adin_{\blaxg{Q}}
        (X \cup Y)\); we know also that \(X \subseteq \data{names}(\Gamma, A)\), hence
        \((X \cup Y) \subseteq \data{names}(\Gamma, A, C)\): then
        \[
            z_i z_{i+1} \adin_{\blaxg{Q} \rst_{\Gamma, A, C}} (X \cup Y)
            \text{ and }
            z_i z_{i+1} \adin_{\blaxg{Q} \rst_{\Gamma, A, C}}^I X.
        \]
        Similarly, if \(z_i z_{i+1} \adin_{\blaxg{R}}^I X\) then \(z_i z_{i+1} \adin_{\blaxg{R}
        \rst_{\Gamma, A, \dl{C}}}^I X\), and we conclude as desired;

        \item if \(P\) ends with a \rl{sum}-, \rl{\lor}- or \rl{\land}-rule
        application that does not introduce \(A \land B\), we conclude by induction hypothesis
        and the fact that restriction distributes over graph unions. \qedhere
    \end{itemize}
\end{proof}

\begin{proof}[Proof of \cref{lemma:bl-axiom-graph-decomposition}]
    Observe first that
    \[
        \data{names}(\Gamma, A \land B)
            = \data{names}(\Gamma, A) \cup \data{names}(\Gamma, B),
    \]
    hence
    \[
        V_{\blaxg{P}} = \data{names}(\Gamma, A \land B)
            = V_{\blaxg{P} \rst_{\Gamma, A}} \cup V_{\blaxg{P} \rst_{\Gamma, B}}.
    \]
    For the branch-label relation, we start by recalling the fact that
    \[
        \data{Br}(\blaxg{P}) \subseteq \data{Br}(\Gamma, A \land B) = \data{Br}(\Gamma, A) \cup \data{Br}(\Gamma, B).
    \]
    Now let \(e \adin_{\blaxg{P}} X\): if \(X \in \data{Br}(\Gamma, A)\), then in particular
    \(X \subseteq \data{names}(\Gamma, A)\), hence \(e \adin_{\blaxg{P} \rst_{\Gamma, A}} X\);
    similarly, if \(X \in \data{Br}(\Gamma, B)\) then \(e \adin_{\blaxg{P} \rst_{\Gamma, B}} X\).
    We have then
    \[
        \blaxg{P} \sqsubseteq {\blaxg{P} \rst_{\Gamma, A}} \sqcup {\blaxg{P} \rst_{\Gamma, B}};
    \]
    the reverse inclusion is immediate.
\end{proof}

\subsection{Auxiliary lemmas for cut-elimination}

\begin{proof}[Proof of \cref{lemma:edges-join-dual-atoms}]
    By structural induction on~\(P\); we leave the base case (axioms),
    superpositions and the logical rules to the reader. For cuts, assume \(P\)
    has the form
    \[
        \prftree[r]{\rl{cut}}
            {\prfsummary[\(Q\)]{\th \Gamma, A \vphantom{\dl{A}}}}
            {\prfsummary[\(R\)]{\th \Gamma, \dl{A}}}
            {\th \Gamma}
    \]
    We have then \(\blaxg{P} = \blaxg{Q} \odot_A \blaxg{R}\). Let \(xy \in E_{\blaxg{P}}\):
    there is by \cref{defn:bl-composition} a labeled alternating path \(z_1, \ldots, z_n\)
    between \(\blaxg{Q}\) and~\(\blaxg{R}\) through the interface~\(\data{names}(A)\),
    with \(xy = z_1 z_n\). For all \(1 < i < n\) (the interface nodes) we have \(\dl{A}[z_i]
    = \dl{A[z_i]}\). Let \(\alpha = \Gamma[x]\). There are two possibilities:
    \begin{itemize}
        \item \emph{odd edges in~\(\blaxg{Q}\), even edges in~\(\blaxg{R}\):} we prove
        by secondary induction on~\(i\) that for all~\(1 < i < n\), if \(i\) is odd then
        \(A[z_i] = \alpha\), if it is even then \(A[z_i] = \dl{\alpha}\). If \(i = 2\)
        (even) then \(z_1 z_2 \in E_{\blaxg{Q}}\), hence by the primary induction hypothesis
        \[
            A[z_2] = (\Gamma, A)[z_2] = \dl{(\Gamma, A)[z_1]} = \dl{\Gamma[z_1]}
                = \dl{\Gamma[x]} = \dl{\alpha}.
        \]
        If \(i > 2\) and even, then \(z_{i-1} z_i \in E_{\blaxg{Q}}\); by the secondary
        induction hypothesis \(A[z_{i-1}] = \alpha\), and by the primary hypothesis
        \[
            A[z_i] = (\Gamma, A)[z_i] = \dl{(\Gamma, A)[z_{i-1}]} = \dl{A[z_{i-1}]} = \dl{\alpha}.
        \]
        If \(i > 2\) and odd, then \(z_{i-1} z_i \in E_{\blaxg{R}}\); by the secondary
        induction hypothesis \(A[z_{i-1}] = \dl{\alpha}\) hence \(\dl{A}[z_{i-1}] = \alpha\);
        by the primary hypothesis
        \[
            \dl{A}[z_i] = (\Gamma, \dl{A})[z_i] = \dl{(\Gamma, \dl{A})[z_{i-1}]}
                = \dl{\dl{A}[z_{i-1}]} = \dl{\alpha},
        \]
        therefore \(A[z_i] = \dl{\dl{\alpha}} = \alpha\). We conclude by cases on~\(n\):
        if \(n = 2\), then by the primary induction hypothesis
        \[
            \Gamma[y] = (\Gamma, A)[z_2] = \dl{(\Gamma, A)[z_1]} = \dl{\Gamma[x]};
        \]
        if \(n > 2\) and even, then \(z_{n-1} z_n \in E_{\blaxg{Q}}\) and \(A[z_{n-1}]
        = \alpha\); by the primary induction hypothesis
        \[
            \Gamma[y] = (\Gamma, A)[z_n] = \dl{(\Gamma, A)[z_{n-1}]} = \dl{A[z_{n-1}]}
                = \dl{\alpha} = \dl{\Gamma[x]};
        \]
        if \(n > 2\) and odd, then \(z_{n-1} z_n \in E_{\blaxg{R}}\) and \(A[z_{n-1}]
        = \dl{\alpha}\), hence \(\dl{A}[z_{n-1}] = \alpha\); by the primary induction hypothesis
        \[
            \Gamma[y] = (\Gamma, \dl{A})[z_n] = \dl{(\Gamma, \dl{A})[z_{n-1}]}
                = \dl{\dl{A}[z_{n-1}]} = \dl{\alpha} = \dl{\Gamma[x]}.
        \]

        \item \emph{odd edges in~\(\blaxg{R}\), even edges in~\(\blaxg{Q}\):} we proceed as
        in the previous case, swapping \(Q\) with \(R\) and \(A\) with \(\dl{A}\). \qedhere
    \end{itemize}
\end{proof}

\begin{proof}[Proof of \cref{lemma:virtual-height-decrease}]
    By structural induction on~\(P\). We prove the result only for the base case (i.e.\ the
    axioms), all other cases are more or less immediate.

    \emph{Case~(i).} \(P\) is an axiom rule application of the kind \(\rl[\{\dl{B},A\}]{ax}\)
    with conclusion~\(\th \Delta, \dl{B}, A\), where \(\Gamma = \Delta \cup \{\dl{B}\}\):
    \begin{itemize}
        \item if \(A = A_1 \lor A_2\) is a disjunction, then \(\data{inv}(P, A)\) has the form
        \[
            \prftree[r]{\rl{\lor}}
                {\prftree[r]{\rl{\land}}
                    {\prfbyaxiom{\rl[\{\dl{B}_1, A_1\}]{ax}}{\th \Delta, \dl{B}_1, A_1, A_2}}
                    {\prfbyaxiom{\rl[\{\dl{B}_2, A_2\}]{ax}}{\th \Delta, \dl{B}_2, A_1, A_2}}
                    {\th \Delta, \dl{B}_1 \land \dl{B}_2, A_1, A_2}}
                {\th \Delta, \dl{B}_1 \land \dl{B}_2, A_1 \lor A_2}
        \]
        where \(\dl{B} = \dl{B}_1 \land \dl{B}_2\). We have
        \begin{align*}
            \data{vh}(P) &= 1 + \data{deg}(\th \Delta, \dl{B}_1 \land \dl{B}_2, A_1 \lor A_2)\\
                &= 1 + \data{deg}(\th \Delta)\\
                &   \qquad\! + 1 + \data{deg}(\dl{B}_1) + \data{deg}(\dl{B}_2)\\
                &   \qquad\! + 1 + \data{deg}(A_1) + \data{deg}(A_2)\\
                &= 3 + \data{deg}(\th \Delta, \dl{B}_1, \dl{B}_2, A_1, A_2)
        \end{align*}
        and
        \begin{align*}
            \data{vh}(\data{inv}(P, A)) &= 2 + \max \{
                    1 + \data{deg}(\th \Delta, \dl{B}_1, A_1, A_2),
                    1 + \data{deg}(\th \Delta, \dl{B}_2, A_1, A_2)
                \}\\
                &= 3 + \max \{
                    \data{deg}(\th \Delta, \dl{B}_1, A_1, A_2),
                    \data{deg}(\th \Delta, \dl{B}_2, A_1, A_2)
                \}.
        \end{align*}
        Clearly,
        \[
            \data{deg}(\th \Delta, \dl{B}_1, A_1, A_2), \data{deg}(\th \Delta, \dl{B}_2, A_1, A_2)
                \leq \data{deg}(\th \Delta, \dl{B}_1, \dl{B}_2, A_1, A_2),
        \]
        hence \(\data{vh}(\data{inv}(P, A)) \leq \data{vh}(P)\);

        \item if \(A = A_1 \land A_2\) is a conjunction, then \(\data{inv}(P, A)\) has the form
        \[
            \prftree[r]{\rl{\land}}
                {\prftree[r]{\rl{\lor}}
                    {\prfbyaxiom{\rl[\{\dl{B}_1, A_1\}]{ax}}{\th \Delta, \dl{B}_1, \dl{B}_2, A_1}}
                    {\th \Delta, \dl{B}_1 \lor \dl{B}_2, A_1}}
                {\prftree[r]{\rl{\lor}}
                    {\prfbyaxiom{\rl[\{\dl{B}_2, A_2\}]{ax}}{\th \Delta, \dl{B}_1, \dl{B}_2, A_2}}
                    {\th \Delta, \dl{B}_1 \lor \dl{B}_2, A_2}}
                {\th \Delta, \dl{B}_1 \land \dl{B}_2, A_1 \lor A_2}
        \]
        where \(\dl{B} = \dl{B}_1 \lor \dl{B}_2\). We have as before
        \[
            \data{vh}(P) = 3 + \data{deg}(\th \Delta, \dl{B}_1, \dl{B}_2, A_1, A_2)
        \]
        and
        \begin{align*}
            \data{vh}(\data{inv}(P, A)) &= 1 + \max \{
                    2 + \data{deg}(\th \Delta, \dl{B}_1, \dl{B}_2, A_1),
                    2 + \data{deg}(\th \Delta, \dl{B}_1, \dl{B}_2, A_2)
                \}\\
                &= 3 + \max \{
                    \data{deg}(\th \Delta, \dl{B}_1, \dl{B}_2, A_1),
                    \data{deg}(\th \Delta, \dl{B}_1, \dl{B}_2, A_2)
                \},
        \end{align*}
        with
        \[
            \data{deg}(\th \Delta, \dl{B}_1, \dl{B}_2, A_1),
                \data{deg}(\th \Delta, \dl{B}_1, \dl{B}_2, A_2)
                    \leq \data{deg}(\th \Delta, \dl{B}_1, \dl{B}_2, A_1, A_2),
        \]
        and again \(\data{vh}(\data{inv}(P, A)) \leq \data{vh}(P)\).
    \end{itemize}

    \emph{Case~(ii).} \(P\) is an axiom rule application of the kind \(\rl[\{\dl{B},C\}]{ax}\)
    with conclusion~\(\th \Delta, A, \dl{B}, C\), where \(\Gamma = \Delta \cup \{\dl{B}, C\}\):
    \begin{itemize}
        \item if \(A = A_1 \lor A_2\) is a disjunction, then \(\data{inv}(P, A)\) has the form
        \[
            \prftree[r]{\rl{\lor}}
                {\prfbyaxiom{\rl[\{\dl{B},C\}]{ax}}{\th \Delta, A_1, A_2, \dl{B}, C}}
                {\th \Delta, A_1 \lor A_2, \dl{B}, C}
        \]
        We have
        \begin{align*}
            \data{vh}(P) &= 1 + \data{deg}(\th \Delta, A_1 \lor A_2, \dl{B}, C)\\
                &= 1 + \data{deg}(\th \Delta, \dl{B}, C)\\
                &   \qquad\! + 1 + \data{deg}(A_1) + \data{deg}(A_2)\\
                &= 2 + \data{deg}(\th \Delta, A_1, A_2, \dl{B}, C)
        \end{align*}
        and
        \[
            \data{vh}(\data{inv}(P, A)) = 2 + \data{deg}(\th \Delta, A_1, A_2, \dl{B}, C),
        \]
        hence \emph{a fortiori} \(\data{vh}(\data{inv}(P, A)) \leq \data{vh}(P)\);

        \item if \(A = A_1 \land A_2\) is a conjunction, then \(\data{inv}(P, A)\) has the form
        \[
            \prftree[r]{\rl{\land}}
                {\prfbyaxiom{\rl[\{\dl{B},C\}]{ax}}{\th \Delta, A_1, \dl{B}, C}}
                {\prfbyaxiom{\rl[\{\dl{B},C\}]{ax}}{\th \Delta, A_2, \dl{B}, C}}
                {\th \Delta, A_1 \land A_2, \dl{B}, C}
        \]
        We have as before
        \[
            \data{vh}(P) = 2 + \data{deg}(\th \Delta, A_1, A_2, \dl{B}, C)
        \]
        and
        \begin{align*}
            \data{vh}(\data{inv}(P, A)) &= 1 + \max \{
                    1 + \data{deg}(\th \Delta, A_1, \dl{B}, C),
                    1 + \data{deg}(\th \Delta, A_2, \dl{B}, C)
                \}\\
                &= 2 + \max \{
                    \data{deg}(\th \Delta, A_1, \dl{B}, C),
                    \data{deg}(\th \Delta, A_2, \dl{B}, C)
                \}
        \end{align*}
        with
        \[
            \data{deg}(\th \Delta, A_1, \dl{B}, C),
                \data{deg}(\th \Delta, A_2, \dl{B}, C)
                    \leq \data{deg}(\th \Delta, A_1, A_2, \dl{B}, C)
        \]
        hence \(\data{vh}(\data{inv}(P, A)) \leq \data{vh}(P)\). \qedhere
    \end{itemize}
\end{proof}

\section{Proof of the semantic cut-admissibility lemma}\label{sec:semantic-cut-admissibility-proof}

We devote the present section to the proof of \cref{lemma:semantic-cut-admissibility},
which requires a somewhat involved construction. We start by fixing a set
\[
    \data{Pol} = \{ {\circ}, {\bullet} \}
\]
of \emph{polarities}. For \(p \in \data{Pol}\), we write \(\co{p}\) for the opposite
polarity, i.e.\ let \(\co{\circ} = {\bullet}\) and~\(\co{\bullet} = {\circ}\). As background
data for the construction, we have to provide a pair of bl-graphs \(G_{\circ}, G_{\bullet}\)
and a \emph{finite} set of names \(I \subseteq \mathcal{N}\) to act as the composition
interface. The two graphs are required to satisfy the following property:
\begin{quote}\itshape
    there is a unique branch name \(X \subseteq \mathcal{N}\) such that, for all~\(p \in
    \data{Pol}\) and~\(e \in E_{G_p}\), \(e \adin_{G_p}^I X\);
\end{quote}
in other words, the two graphs must have collectively a unique branch up to names in
the interface. The restriction has two important consequences:
\begin{itemize}
    \item all alternating paths between the two graphs through \(I\) are necessarily
    \(X\)-labeled, hence composable;
    \item any edge can be added to any alternating path, without caring for its label.
\end{itemize}
More formally, let \(z_1, \ldots, z_n\) be an \(X\)-labeled alternating path between
\(G_{\circ}\) and~\(G_{\bullet}\) through interface~\(I\) (\cref{defn:bl-alternating-path}).
For any \(p \in \data{Pol}\) we call the path
\begin{itemize}
    \item \emph{\(p\)-initial}, if all odd edges are in~\(G_p\) and all even edges in~\(G_{\co{p}}\);
    \item \emph{\(p\)-final}, if it is \(p\)-initial with even~\(n\), or \(\co{p}\)-initial
    with odd~\(n\).
\end{itemize}
Because the graphs and interface -- as well as the unique branch label -- are fixed, from now
on we are going to speak simply of alternating paths, without mentioning the full expression
(\(X\)-labeled alt.\ paths between \(G_{\circ}\) and~\(G_{\bullet}\) through interface~\(I\)).
The following properties hold:

\begin{lemma}
    \label{lemma:single-edge-is-altpath}
    For all~\(p \in \data{Pol}\) and edge~\(xy \in E_{G_p}\), \(x,y\) is a \(p\)-initial
    and \(p\)-final alternating path.
\end{lemma}

\begin{lemma}
    \label{lemma:altpath-reversal}
    For any~\(p, q \in \data{Pol}\), if \(z_1, \ldots, z_n\) is a \(p\)-initial and
    \(q\)-final alternating path then \(z_n, \ldots, z_1\) is a \(q\)-initial and \(p\)-final
    alternating path.
\end{lemma}

\begin{lemma}
    \label{lemma:altpath-composition}
    Let \(p \in \data{Pol}\), and let \(\vec{z} = z_1, \ldots, z_n\), \(\vec{w} =
    w_1, \ldots, w_n\) be alternating paths such that \(\vec{z}\) is \(p\)-final,
    \(\vec{w}\) is \(\co{p}\)-initial and \(z_n = w_1\): then
    \[
        \vec{zw} = z_1, \ldots, (z_n = w_1), \ldots, w_n
    \]
    is also an alternating path, \(q\)-initial iff so is \(\vec{z}\) and \(q\)-final
    iff so is \(\vec{w}\).
\end{lemma}

The proofs are tedious but straightforward and we leave them to the interested reader.
The construction operates on \emph{finite} sets \(\sigma\) of pairs with each pair of
the form
\[
    \langle S, (v_1, \ldots, v_n) \rangle \in \sigma
\]
where \(S \subseteq I \times \data{Pol}\) is a \emph{set of name-polarity pairs} (with names
from the interface~\(I\)) and~\(v_1, \ldots, v_n \in (V_{G_{\circ}} \cup V_{G_{\bullet}})^\ast\)
is a \emph{finite sequence of vertices} from the two graphs. We call \(\sigma\) a \emph{state}
of the construction. We distinguish \emph{consistent} and \emph{live} states based on
certain sets of local and global properties, respectively:

\begin{definition}[Consistency]
    \label{defn:consistency}
    Call a state \(\sigma\) of the construction \emph{consistent} if and only if all
    pairs~\(\langle S, (v_1, \ldots, v_n) \rangle \in \sigma\) satisfy the following conditions:
    \begin{enumerate}[(i)]
        \item \emph{mutual exclusion:} at most one polarity is assigned to each name in \(S\),
        i.e.\ for all~\((x, p), (y, q) \in S\), \(x = y\) implies \(p = q\); equivalently,
        if \((x, p) \in S\) then \((x, \co{p}) \notin S\).

        \item \emph{alternation:} \(v_1, \ldots, v_n\) is an alternating path;

        \item \emph{inclusion:} if \(v_1 \in I\) (resp.~\(v_n \in I\)), then there is
        \(p \in \data{Pol}\) such that \((v_1, p) \in S\) (resp.~\((v_n, p) \in S\));

        \item \emph{coloring:} if \((v_1, p) \in S\) (resp.~\((v_n, p) \in S\)), then
        the alternating path \(v_1, \ldots, v_n\) is \(p\)-initial (resp.\ \(p\)-final).
    \end{enumerate}
\end{definition}

Consider now two sets \(S, T \subseteq I \times \data{Pol}\) of name-polarity pairs.
We say that \(S, T\) \emph{compose on name~\(x \in \mathcal{N}\)} (notation \(S \stackrel{x}{\bowtie} T\))
iff there is a polarity~\(p \in \data{Pol}\) such that \((x, p) \in S\), \((x, \co{p}) \in T\)
and \(S \setminus \{(x,p)\} = T \setminus \{(x, \co{p})\}\).

\begin{definition}[Liveness]
    \label{defn:liveness}
    Call a state \(\sigma\) of the construction \emph{live} iff it is non-empty and,
    for all~\(\langle S, \vec{z} \rangle \in \sigma\) and \((x, p) \in S\), there is
    \(\langle T, \vec{w} \rangle \in \sigma\) such that \(S \stackrel{x}{\bowtie} T\).
\end{definition}

\noindent Now let us define for any state \(\sigma\) the set of \emph{live names of \(\sigma\)}:
\[
    \data{names}(\sigma) = \{ x \in \mathcal{N} \mid \langle S, \vec{z} \rangle \in \sigma, (x, p) \in S \}.
\]

\begin{definition}[Terminal state]
    \label{defn:terminal-state}
    Call a state \(\sigma\) of the construction \emph{terminal} iff for all~\(\langle S,
    \vec{z} \rangle \in \sigma\), \(S\) is empty; equivalently iff \(\data{names}(\sigma)
    = \emptyset\).
\end{definition}

\begin{lemma}
    Let \(\sigma\) be any consistent, live and terminal state: there is a pair~\(\langle
    S, \vec{z} \rangle \in \sigma\) such that \(\vec{z}\) is a complete alternating
    path.
\end{lemma}
\begin{proof}
    By liveness \(\sigma\) is non-empty, i.e. there is at least one pair \(\langle S,
    \vec{z} \rangle \in \sigma\). Let \(\vec{z} = z_1, \ldots, z_n\) and assume \(z_1 \in I\):
    by the inclusion condition (consistency) there should be \(p \in \data{Pol}\) such that
    \((z_1, p) \in S\), against the hypothesis that \(\sigma\) be terminal, hence \(z_1 \notin I\).
    The same argument shows that \(z_n \notin I\).
\end{proof}

\begin{corollary}
    If there is a consistent, live and terminal state, then there is a complete \(X\)-labeled
    alternating path between \(G_{\circ}\) and \(G_{\bullet}\) through interface~\(I\).
\end{corollary}

\begin{corollary}
    If there is a consistent, live and terminal state, then the composite bl-graph \(G_{\circ}
    \odot_I G_{\bullet}\) has at least one edge.
\end{corollary}

We come now to the key lemma of this proof, showing that every consistent and live state
can be transformed into a consistent, live and terminal one in a finite number of steps.
Thus we reduce the problem of showing that the composite graph is non-empty to that of
constructing a consistent and live state.

\begin{lemma}[Reduction lemma]
    \label{lemma:state-reduction}
    For any consistent and live state \(\sigma\) such that \(\data{names}(\sigma)\)
    is non-empty, there is a consistent and live state \(\tau\) such that
    \[
        \data{names}(\tau) \subsetneq \data{names}(\sigma).
    \]
\end{lemma}

Because the interface \(I\) is finite by assumption, the set of live names of any state
must be finite, therefore we can reach a terminal one by iterating the reduction lemma
finitely many times:

\begin{corollary}
    If there is a consistent and live state, then there is a consistent, live and terminal state.
\end{corollary}

\begin{corollary}
    If there is a consistent and live state, then the composite bl-graph \(G_{\circ}
    \odot_I G_{\bullet}\) has at least one edge.
\end{corollary}

\begin{proof}[Proof of the reduction lemma]
    By assumption there is \(x \in \data{names}(\sigma)\). We associate to every pair
    \(\langle S, \vec{z} \rangle \in \sigma\) a new pair \(\langle S', \vec{z}' \rangle\)
    constructed as follows, while simultaneously proving that the new pair satisfies all
    consistency conditions (\cref{defn:consistency}):
    \begin{itemize}
        \item if there is no \(p \in \data{Pol}\) such that \((x, p) \in S\), then let
        \(S' = S\), \(\vec{z}' = \vec{z}\); all consistency conditions are clearly preserved;
        \item otherwise let \(\langle T, \vec{w} \rangle \in \sigma\) be the pair
        such that \(S \stackrel{x}{\bowtie} T\), whose existence is guaranteed by liveness:
        we have \(p \in \data{Pol}\) such that \((x, p) \in S\) and \((x, \co{p}) \in T\);
        \item let \(S' = S \setminus \{(x, p)\} = T \setminus \{(x, \co{p})\}\); since we
        are just removing one pair, mutual exclusion is preserved;
        \item let \(\vec{z} = z_1, \ldots, z_n\) and \(\vec{w} = w_1, \ldots, w_n\);
        \item if both \(z_1, z_n \neq x\), then let \(\vec{z}' = \vec{z}\); the last three
        consistency conditions are preserved because no endpoint is~\(x\);
        \item if on the other hand \(w_1, w_n \neq x\), then let \(\vec{z}' = \vec{w}\);
        the last three consistency conditions are preserved because no endpoint is~\(x\)
        and \(S' = T \setminus \{(x, \co{p})\}\);
        \item otherwise there are four mutually exclusive cases:
        \begin{itemize}[\(\circ\)]
            \item if \(z_n = x, w_1 = x\) then let \(\vec{z}' = z_1, \ldots, (z_n = w_1), \ldots, w_n\);
            \item if \(z_n = x, w_n = x\) then let \(\vec{z}' = z_1, \ldots, (z_n = w_n), \ldots, w_1\);
            \item if \(z_1 = x, w_1 = x\) then let \(\vec{z}' = z_n, \ldots, (z_1 = w_1), \ldots, w_n\);
            \item if \(z_1 = x, w_n = x\) then let \(\vec{z}' = z_n, \ldots, (z_1 = w_n), \ldots, w_1\).
        \end{itemize}
        For the inclusion condition, observe that no repetition of vertices is allowed in
        an alternating path; because the endpoint \(x\) has become internal in~\(\vec{z}'\),
        inclusion must still hold for the other endpoints. For alternation: \(\vec{z}'\) is
        obtained by joining two sequences on \(x\), after possibly reversing them. Reversal
        preserves alternating paths by \cref{lemma:altpath-reversal}; the composition is
        correct (\cref{lemma:altpath-composition}) because the two original pairs satisfy
        the coloring condition, hence the first half is \(p\)-final and the second half
        is~\(\co{p}\)-initial. Finally, the coloring condition is preserved because the
        endpoints preserve their original initiality and finality by \cref{lemma:altpath-reversal,%
        lemma:altpath-composition}.
    \end{itemize}
    Clearly the construction is performed in such a way that for any \(\langle S, \vec{z} \rangle
    \in \sigma\), there is no \(p \in \data{Pol}\) such that \((x, p) \in S'\). We define
    then \(\tau\) as
    \[
        \tau = \{ \langle S', \vec{z}' \rangle \mid \langle S, \vec{z} \rangle \in \sigma \}.
    \]
    Since \(\sigma\) is finite, so is \(\tau\), hence it is a consistent state. We have obviously
    \[
        \data{names}(\tau) \subseteq \data{names}(\sigma)
    \]
    and from what we said above we known that \(x \notin \data{names}(\tau)\), hence the
    inclusion is strict, as required.

    We have now to show that \(\tau\) is live (\cref{defn:liveness}). It must be non-empty
    because so is \(\sigma\); let then \(\langle S', \vec{z}' \rangle \in \tau\) and~\((y, p)
    \in S'\): by construction there is some pair~\(\langle S, \vec{z} \rangle \in \sigma\)
    such that \(S' = S \setminus \{(x, q)\}\), and by the liveness of~\(\sigma\) there is
    another pair~\(\langle T, \vec{w} \rangle \in \sigma\) such that \(S \stackrel{y}{\bowtie}
    T\). Then there is a pair~\(\langle T', \vec{w}' \rangle \in \tau\) such that \(T' =
    T \setminus \{(x, \co{q})\}\), and we have \(S' \stackrel{y}{\bowtie} T'\): remembering
    that \(y \neq x\),
    \begin{itemize}
        \item \((y, p) \in (S \setminus \{(x, q)\}) = S'\);
        \item \((y, \co{p}) \in (T \setminus \{(x, \co{q})\}) = T'\);
        \item \(S' \setminus \{(y, p)\} = (S \setminus \{(y, p)\}) \setminus \{(x, q)\}\\
            \null\qquad\qquad
            = (T \setminus \{(y, \co{p})\}) \setminus \{(x, \co{q})\}
                = T' \setminus \{(y, \co{p})\}\). \qedhere
    \end{itemize}
\end{proof}

As announced above, we complete the proof of \cref{lemma:semantic-cut-admissibility} by
showing that a consistent and live state can be constructed. We start with an auxiliary
lemma:

\begin{lemma}
    \label{lemma:dual-formulas-branch-inclusion}
    For any sharing-free named formula~\(A\) and function
    \[
        f: \data{names}(A) \to \data{Pol}
    \]
    there is either \(X \in \data{Br}(A)\) such that \(fX = \{{\circ}\}\) or \(Y \in
    \data{Br}(\dl{A})\) such that \(fY = \{{\bullet}\}\).
\end{lemma}
\begin{proof}
    By induction on~\(A\). If \(A = \nm{x}{\alpha}\) then \(\data{Br}(A) = \data{Br}(\dl{A})
    = \{{x}\}\), and necessarily either \(f\{x\} = \{{\circ}\}\) or \(f\{x\} = \{{\bullet}\}\).

    If \(A = B \land C\) (with~\(\dl{A} = \dl{B} \lor \dl{C}\)) then \(\data{Br}(A) =
    \data{Br}(B) \cup \data{Br}(C)\)). By induction hypothesis there is either \(X' \in
    \data{Br}(B)\) such that \(fX' = \{{\circ}\}\) or \(Y' \in \data{Br}(\dl{B})\)
    such that \(fY' = \{{\bullet}\}\). In the first case, let \(X = X'\) and we're done.
    Otherwise, there is again by induction hypothesis either \(X'' \in \data{Br}(C)\)
    s.t.~\(fX'' = \{{\circ}\}\), or~\(Y'' \in \data{Br}(\dl{C})\) s.t.~\(fY'' = \{{\bullet}\}\).
    If there is such a~\(X''\), let \(X = X''\); otherwise, let \(Y = Y' \cup Y'' \in
    \data{Br}(\dl{B} \lor \dl{C})\): we have \(fY = fY' \cup fY'' = \{{\bullet}\}\).

    If \(A = B \lor C\), apply the same argument, swapping polarities and \(A\) with~\(\dl{A}\).
\end{proof}

\begin{proof}[Proof of \cref{lemma:semantic-cut-admissibility}]
    We fix \(G_{\circ} = \blaxg{P}\), \(G_{\bullet} = \blaxg{Q}\) and \(I = \data{names}(A)\).
    \(I\) is obviously finite, and because the context of the conclusions of \(P, Q\) is
    atomic by hypothesis, the condition about branches is satisfied too (the unique branch
    name~\(X\) is the unique one in \(\data{Br}(\Gamma)\), i.e.~\(\data{names}(\Gamma)\)).

    We construct a state \(\sigma\) as follows:
    \[
        \sigma = \{ \langle f, (x,y) \rangle \mid f \in \data{Pol}^I, p \in \data{Pol}, xy \adin_{G_p} X \cup Y, fY = \{p\} \}.
    \]
    Let us unpack the construction first, then we shall prove that the state is consistent
    and live. \(\data{Pol}^I\) is the set of all functions from~\(I = \data{names}(A)\)
    to~\(\data{Pol}\) seen as sets of pairs, i.e.~\(\data{Pol}^I \subseteq I \times \data{Pol}\);
    we pair each \(f \in \data{Pol}^I\) with any finite sequence~\(x,y\) of vertices
    such that, for at least one polarity~\(p \in \data{Pol}\), \(xy\) is an edge in~\(G_p\),
    and moreover one of its branch labels satisfies the following conditions:
    \begin{itemize}
        \item it is equal to \(X \cup Y\) for some \(Y\) (where \(X\) is the unique branch
        name of~\(\Gamma\) described above);
        \item it is such that \(fY = \{p\}\).
    \end{itemize}
    \(\sigma\) is clearly finite, hence a state. For consistency (\cref{defn:consistency}),
    let \(\langle f, (x,y) \rangle \in \sigma\):
    \begin{itemize}
        \item \emph{mutual exclusion:} this condition can be read as asking that \(f\)
        be the graph of a function, hence it is obviously satisfied;

        \item \emph{alternation:} \(xy\) is an edge in either \(G_{\circ}\) or~\(G_{\bullet}\)
        by construction, hence an alternating path by \cref{lemma:single-edge-is-altpath};

        \item \emph{inclusion:} because \(f\) is total, if \(x \in I\) then \((x, fx) \in f\),
        and similarly for \(y\);

        \item \emph{coloring:} let \((x, p) \in f\); by construction there is \(q \in
        \data{Pol}\), \(Y \in \mathcal{N}\) such that \(xy \adin_{G_q} X \cup Y\) with
        \(fY = \{q\}\). By \cref{lemma:branches-of-union}, every branch name of~\(G_q\)
        is the union of the unique branch name~\(X\) of~\(\Gamma\) with some branch name
        of either~\(A\) (if \(q = {\circ}\)) or~\(\dl{A}\) (if \(q = {\bullet}\)).
        If \(q = {\circ}\), then, \(Y \in \data{Br}(A)\), otherwise \(Y \in \data{Br}(\dl{A})\).
        By \cref{defn:bl-graphs} \(x \in X \cup Y\), and because \(x \in I\) we must
        have \(x \in Y\). We know that \(fY = \{q\}\), hence in particular \(p = f(x) = q\).
        The path \(x,y\) is obviously \(q\)-initial, hence \(p\)-initial. Analogous reasoning
        shows that if \((y, p') \in f\), then the path is \(p'\)-final.
    \end{itemize}
    For liveness (\cref{defn:liveness}) we argue first that \(\sigma\) is non-empty:
    by \cref{lemma:dual-formulas-branch-inclusion} above, there is for each \(f \in
    \data{Pol}^I\) at least one branch \(Y \in \data{Br}(A)\) with \(fY = \{{\circ}\}\)
    or \(Y \in \data{Br}(\dl{A})\) with \(fY = \{{\bullet}\}\). Let \(fY = \{p\}\);
    by \cref{lemma:bl-axiom-graphs-cutfree-branches}, \((X \cup Y) \in \data{Br}(G_p)\),
    i.e.\ there is \(e \adin_{G_p} X \cup Y\). We have thus proven not just that \(\sigma\)
    is non-empty, but also that for each \(f \in \data{Pol}^I\) there is at least one
    pair \(\langle f, \vec{z} \rangle \in \sigma\).

    Now let \(\langle f, \vec{z} \rangle \in \sigma\), \((x, p) \in f\); by inverting the polarity
    assignment of \(x\) in~\(f\) we obtain \(f' = (f \setminus \{(x,p)\}) \cup \{(x, \co{p})\}
    \in \data{Pol}^I\). By the reasoning above there is \(\langle f', \vec{w} \rangle \in \sigma\),
    and it is immediate by construction that \(f \stackrel{x}{\bowtie} f'\).
\end{proof}

\section{Totality lemmas and correctness algorithm for BLG}\label{sec:totality-proofs}

\begin{proof}[Proof of \cref{lemma:bl-axiom-graphs-cutfree-branches}]
    By structural induction on~\(P\). If \(P\) ends with an axiom rule application,
    the conclusion is equivalent to \cref{coro:bl-axiom-branches}; if \(P\) ends with
    a superposition rule application, it follows immediately from the induction hypothesis;
    if \(P\) ends with a logical rule, we conclude from the induction hypothesis and
    \cref{propo:sequent-branches-decomposition}, points~(ii) and~(iii). The details
    are left to the reader.
\end{proof}

\begin{proof}[Proof of \cref{lemma:semantic-disjunction-inversion}]
    We have \(\data{names}(\Gamma, A, B) = \data{names}(\Gamma, A \lor B)\); \(\data{Br}(
    \Gamma, A, B) = \data{Br}(\Gamma, A \lor B)\) by \cref{propo:sequent-branches-decomposition};
    and finally \((\Gamma, A, B)[x] = (\Gamma, A \lor B)[x]\) for all~\(x \in \data{names}(
    \Gamma, A, B)\).
\end{proof}

\begin{proof}[Proof of \cref{lemma:semantic-conjunction-inversion}]
    We have two prove two facts: \emph{(i)} that \(G\) is effectively equal to the union
    of the two restrictions, and \emph{(ii)} that the two restrictions are total w.r.t.\ the
    restricted conclusions. For fact~(i), we have obviously
    \[
        \data{names}(\Gamma, A \land B) = \data{names}(\Gamma, A) \cup \data{names}(\Gamma, B).
    \]
    Then by construction \(V_{G \rst_{\Gamma, A}} = \data{names}(\Gamma, A)\) and
    \(V_{G \rst_{\Gamma, B}} = \data{names}(\Gamma, B)\), hence \(V_G = V_{G \rst_{\Gamma, A}}
    \cup V_{G \rst_{\Gamma, B}}\). Now let \(e \adin_G X\). By totality \(X \in \data{Br}(\Gamma,
    A \land B)\); then by \cref{propo:sequent-branches-decomposition} either \(X \in
    \data{Br}(\Gamma, A)\) or \(X \in \data{Br}(\Gamma, B)\): in the first case we have
    \(X \subseteq \data{names}(\Gamma, A)\) hence \(e \adin_{G \rst_{\Gamma, A}} X\);
    similarly in the second case we have \(e \adin_{G \rst_{\Gamma, B}} X\). The reverse
    inclusion is obvious.

    For fact~(ii), we have already argued that the two restriction have the appropriate vertex
    sets (\cref{defn:totality}, point~(i)); that all edges link dual atoms (\cref{defn:totality},
    point~(iii)) follows from the fact that the edges come from~\(G\), which satisfies the same
    condition, and for all~\(x \in \data{names}(\Gamma, A)\) (resp.~\(\data{names}(\Gamma, B)\))
    we have \((\Gamma, A \land B)[x] = (\Gamma, A)[x]\) (resp.~\((\Gamma, B)[x]\)). Finally
    we have to show that \(\data{Br}(G \rst_{\Gamma, A}) = \data{Br}(\Gamma, A)\) and
    \(\data{Br}(G \rst_{\Gamma, B}) = \data{Br}(\Gamma, B)\) (\cref{defn:totality}, point~(ii)).
    The forward inclusions are proved already in \cref{lemma:conj-restriction-branches}.
    For the reverse, let \(X \in \data{Br}(\Gamma, A)\) (resp.~\(\data{Br}(\Gamma, B)\));
    by totality there is \(e \adin_G X\), and because \(X \subseteq \data{names}(\Gamma, A)\)
    (resp.~\(\data{names}(\Gamma, B)\)) we have \(e \adin_{G \rst_{\Gamma, A}} X\) (resp.~%
    \(e \adin_{G \rst_{\Gamma, B}} X\)).
\end{proof}

\subsection{Correctness algorithm for BLG}

\begin{proof}[Proof of \cref{propo:polynomial-time-totality}]
    We work under the reasonable assumption that checking name equality requires constant time.
    Let us start by recalling the definition of the size of~\(\mathbf{G}\) (\cref{defn:blg-size}):
    \[
        \data{size}(\mathbf{G}) = \data{size}(\th \Gamma) + |V_G| + \sum_{e \adin_G X} |X|.
    \]
    The size of~\(\th \Gamma\) (\cref{defn:sequent-complexity}, \cref{sec:complexity}) is by
    \cref{propo:sequent-complexity-relationships} the sum of the number of atom occurrences
    in~\(\Gamma\) (notation \(\data{\#at}(\th \Gamma)\)) with the number of logical symbols,
    also called the \emph{degree of~\(\Gamma\)} (notation \(\data{deg}(\Gamma)\)). Because
    \(\Gamma\) is sharing-free by hypothesis, the number of atom occurrences coincides with
    the number of names, i.e.~\(\data{\#at}(\th \Gamma) = |\data{names}(\Gamma)|\).

    Observe also that the sum at the end of the expression is taken not over the set
    of branches of~\(G\), but over all edge-branch pairs, i.e.\ a branch may be counted
    more than once if it has multiple edges. Because \(e \adin_G X\) implies \(e \subseteq X\),
    i.e.\ no branch in a bl-graph is empty, the sum provides an upper bound to the
    number of edge-branch pairs as well as to that of edges and branches, i.e.\ we have
    \[
        |E_G|, |\data{Br}(G)| \leq |{\adin_G}| \leq \data{size}(\mathbf{G}).
    \]
    We have to check three conditions separately (\cref{defn:totality}):
    \begin{enumerate}[(i)]
        \item \(V_G = \data{names}(\Gamma)\): checking name set equality is worst-case
        polynomial in their cardinalities, and constructing the set \(\data{names}(\Gamma)\)
        from \(\Gamma\) requires polynomial time in the size of \(\th \Gamma\);

        \item \(\data{Br}(G) = \data{Br}(\Gamma)\): this is the most complex problem.
        Generating just one element in~\(\data{Br}(\Gamma)\) amounts to persistently expanding
        one branch of a~\(\data{GS4}^\mathcal{N}\) derivation of~\(\Gamma\) until an atomic
        sequent is reached; the length of such a branch is known to be bounded by the
        complexity degree of~\(\th \Gamma\) \cite{Pul22}, hence the cost of generating
        one branch name is polynomially bounded in the size of~\(\th \Gamma\). However,
        the total number of branches is in the worst case exponential in the complexity
        degree of~\(\th \Gamma\): thus we cannot take the naive approach of constructing
        the whole set~\(\data{Br}(\Gamma)\), then testing for equality. We describe an
        informal algorithm: the idea is to generate the branch names of~\(\Gamma\)
        incrementally, match them with some element from the set \(\data{Br}(G)\) and
        erase that element. For every matching attempt, either the algorithm fails or
        the number of branches still to be matched decreases: in this way the number
        of generated branch names is always bounded by the cardinality of~\(\data{Br}(G)\).
        Let \(S\) be the branch set to test (initially \(S = \data{Br}(G)\)), \(\Delta\)
        the sequent to test against (initially \(\Delta = \Gamma\)); the algorithm has
        three phases:
        \begin{itemize}
            \item \emph{problem reduction phase:} if the sequent \(\Delta\) to test
            against contains a disjunction, i.e.\ is of the form \(\Delta', A \lor B\),
            then replace it with \(\Delta', A, B\). If it contains no disjunction but
            at least a conjunction, i.e.\ is of the form \(\Delta', A \land B\), then
            replace it with \(\Delta', A\) and append \(\Delta', B\) to a list of
            sequents to test against later. By \cref{propo:sequent-branches-decomposition},
            the set of branches to test against has not changed;

            \item \emph{matching phase:} if the sequent \(\Delta\) to test against is atomic,
            then it has a unique branch \(X = \data{names}(\Delta)\). We search for that branch
            in~\(S\): if not found, then \(X \in \data{Br}(\Gamma)\) but \(X \notin \data{Br}(G)\),
            and we stop; if found, we erase \(X\) from \(S\) and move on to the backtracking
            phase;

            \item \emph{backtracking phase:} if both \(S\) and the list of delayed sequents are
            empty, then we're done. If \(S\) is empty but the list is not, then some branches
            are missing from \(\data{Br}(G)\) and we stop. If \(S\) has some elements but
            the list is empty, then there are excess branches in \(\data{Br}(G)\) and we
            stop. If both are non-empty, we pick the first element of the list as the new
            \(\Delta\), erase it from the list and move back to the reduction phase.
        \end{itemize}
        It is clear that the number of steps in the reduction phase is bounded by the complexity
        degree of the active sequent~\(\Delta\), which always decreases: when it reaches zero,
        we move to the matching phase. The matching phase either fails or decreases the
        size of the set~\(B\); once the set~\(B\) becomes empty, the algorithm stops:
        thus the size of~\(B\) bounds the number of future reduction phases. Observe now that
        a reduction step requires polynomial time in the size of the active sequent;
        a matching step requires polynomial time in the cardinality of the generated branch
        and the sum of the cardinalities of all elements of~\(B\). Backtracking steps require
        constant time. Every measure mentioned above is itself bounded by \(\data{size}(\mathbf{G})\).
        We have then a polynomial bound on the execution time of the algorithm with parameter
        \((\data{size}(\mathbf{G}))^3\).

        \item for all~\(xy \in E_G\), \(\Gamma[x] = \dl{\Gamma[y]}\): we observed at the
        beginning of the proof that \(\data{size}(\mathbf{G})\) bounds \(|E_G|\); finding
        the atom associated to a given name in~\(\Gamma\) is polynomial in the size
        of~\(\th \Gamma\). \qedhere
    \end{enumerate}
\end{proof}

\clearpage

{
\null
\vfill
\begin{figure}[!h]
    \begin{subfigure}{\textwidth}
        \noindent\makebox[\textwidth]{
            \prftree[r]{\rl{cut}}
                {\prftree[r]{\rl{\land}}
                    {\prftree[r]{\rl{\lor}}
                        {\prfbyaxiom{\rl[\{\nm{x}{\dl{\alpha}}, \nm{u}{\alpha}\}]{ax}}
                            {\th \nm{x}{\dl{\tikzmarknode{x-1}{\alpha}}}, \nm{y}{\alpha}, \nm{u}{\tikzmarknode{u-1}{\alpha}}}}
                        {\th \nm{x}{\dl{\tikzmarknode{x-2}{\alpha}}} \lor \nm{y}{\alpha}, \nm{u}{\tikzmarknode{u-2}{\alpha}}}}
                    {\prftree[r]{\rl{\lor}}
                        {\prfbyaxiom{\rl[\{\nm{z}{\dl{\alpha}}, \nm{u}{\alpha}\}]{ax}}
                            {\th \nm{z}{\dl{\alpha}}, \nm{w}{\alpha}, \nm{u}{\alpha}}}
                        {\th \nm{z}{\dl{\alpha}} \lor \nm{w}{\alpha}, \nm{u}{\alpha}}}
                    {\th (\nm{x}{\dl{\tikzmarknode{x-3}{\alpha}}} \lor \nm{y}{\alpha}) \land (\nm{z}{\dl{\alpha}} \lor \nm{w}{\alpha}), \nm{u}{\tikzmarknode{u-3}{\alpha}}}}
                {\prftree[r]{\rl{\land}}
                    {\prftree[r]{\rl{\lor}}
                        {\prfbyaxiom{\rl[\{\nm{y}{\alpha}, \nm{u}{\dl{\alpha}}\}]{ax}}
                            {\th \nm{x}{\dl{\alpha}}, \nm{y}{\alpha}, \nm{u}{\dl{\alpha}}}}
                        {\th \nm{x}{\dl{\alpha}} \lor \nm{y}{\alpha}, \nm{u}{\dl{\alpha}}}}
                    {\prftree[r]{\rl{\lor}}
                        {\prfbyaxiom{\rl[\{\nm{w}{\alpha}, \nm{u}{\dl{\alpha}}\}]{ax}}
                            {\th \nm{z}{\dl{\alpha}}, \nm{w}{\tikzmarknode{w-1}{\alpha}}, \nm{u}{\dl{\tikzmarknode{u-4}{\alpha}}}}}
                        {\th \nm{z}{\dl{\alpha}} \lor \nm{w}{\tikzmarknode{w-2}{\alpha}}, \nm{u}{\dl{\tikzmarknode{u-5}{\alpha}}}}}
                    {\th (\nm{x}{\dl{\alpha}} \lor \nm{y}{\alpha}) \land (\nm{z}{\dl{\alpha}} \lor \nm{w}{\tikzmarknode{w-3}{\alpha}}), \nm{u}{\dl{\tikzmarknode{u-6}{\alpha}}}}}
                {\th (\nm{x}{\dl{\tikzmarknode{x-4}{\alpha}}} \lor \nm{y}{\alpha}) \land (\nm{z}{\dl{\alpha}} \lor \nm{w}{\tikzmarknode{w-4}{\alpha}})}
        }%
        \begin{tikzpicture}[remember picture, overlay, thick, blue]
            % Axioms
            \draw ([yshift=2pt]x-1.north) .. controls ++(0,1.3em) and ++(0,1.3em) .. ([yshift=2pt]u-1.north);
            \draw ([yshift=2pt]w-1.north) .. controls ++(0,1.3em) and ++(0,1.3em) .. ([yshift=2pt]u-4.north);
            % History of x ([xshift=-.5pt]x-4.west)
            \draw ([yshift=2pt]x-4.north) .. controls ++(0,15pt) and ++(0,-15pt) .. ([yshift=-1pt]x-3.south);
            \draw ([yshift=2pt]x-3.north) .. controls ++(0,5pt) and ++(0,-5pt) .. ([yshift=-1pt]x-2.south);
            \draw ([yshift=2pt]x-2.north) .. controls ++(0,2pt) and ++(0,-2pt) .. ([yshift=-1pt]x-1.south);
            % History of u (left)
            \draw ([yshift=1pt]u-2.north) .. controls ++(0,4pt) and ++(0,-4pt) .. ([yshift=-1pt]u-1.south);
            \draw ([yshift=1pt]u-3.north) .. controls ++(0,10pt) and ++(0,-7pt) .. ([yshift=-1pt]u-2.south);
            % History of w
            \draw ([yshift=1pt]w-4.north) .. controls ++(0,15pt) and ++(0,-15pt) .. ([yshift=-1pt]w-3.south);
            \draw ([yshift=1pt]w-3.north) .. controls ++(0,8pt) and ++(0,-4pt) .. ([yshift=-1pt]w-2.south);
            \draw ([yshift=1pt]w-2.north) .. controls ++(0,3pt) and ++(0,-3pt) .. ([yshift=-1pt]w-1.south);
            % History of u (right)
            \draw ([yshift=2pt]u-5.north) .. controls ++(0,3pt) and ++(0,-3pt) .. ([yshift=-1pt]u-4.south);
            \draw ([yshift=2pt]u-6.north) .. controls ++(0,5pt) and ++(0,-5pt) .. ([yshift=-1pt]u-5.south);
            % Cuts
            \path[draw, densely dashed, semithick, rounded corners=1em]
                ([yshift=-1pt]u-3.south) -- ([yshift=-2.3em]u-3.south) -- ([yshift=-2.3em]u-6.south) -- ([yshift=-1pt]u-6.south);
        \end{tikzpicture}%
        \\[.3em]
        \caption{The initial derivation, where the conjunction in the conclusion is not
            introduced by the last rule. Assume names \(x, y, z, w, u \in \mathcal{N}\)
            are pairwise distinct.}
        \label{fig:axiom-graph-inequality-example1:pre}
    \end{subfigure}
    \par\vspace{4em}%
    \begin{subfigure}{\textwidth}
        \noindent\makebox[\textwidth]{
            \prftree[r]{\rl{\land}}
                {\prftree[r]{\rl{cut}}
                    {\prftree[r]{\rl{\lor}}
                        {\prfbyaxiom{\rl[\{\nm{x}{\dl{\alpha}}, \nm{u}{\alpha}\}]{ax}}
                            {\th \nm{x}{\dl{\tikzmarknode{x-1}{\alpha}}}, \nm{y}{\alpha}, \nm{u}{\tikzmarknode{u-1}{\alpha}}}}
                        {\th \nm{x}{\dl{\tikzmarknode{x-2}{\alpha}}} \lor \nm{y}{\alpha}, \nm{u}{\tikzmarknode{u-2}{\alpha}}}}
                    {\prftree[r]{\rl{\lor}}
                        {\prfbyaxiom{\rl[\{\nm{y}{\alpha}, \nm{u}{\dl{\alpha}}\}]{ax}}
                            {\th \nm{x}{\dl{\alpha}}, \nm{y}{\alpha}, \nm{u}{\dl{\alpha}}}}
                        {\th \nm{x}{\dl{\alpha}} \lor \nm{y}{\alpha}, \nm{u}{\dl{\tikzmarknode{u-3}{\alpha}}}}}
                    {\th \nm{x}{\dl{\tikzmarknode{x-3}{\alpha}}} \lor \nm{y}{\alpha}}}
                {\prftree[r]{\rl{cut}}
                    {\prftree[r]{\rl{\lor}}
                        {\prfbyaxiom{\rl[\{\nm{z}{\dl{\alpha}}, \nm{u}{\alpha}\}]{ax}}
                            {\th \nm{z}{\dl{\alpha}}, \nm{w}{\alpha}, \nm{u}{\alpha}}}
                        {\th \nm{z}{\dl{\alpha}} \lor \nm{w}{\alpha}, \nm{u}{\tikzmarknode{u-6}{\alpha}}}}
                    {\prftree[r]{\rl{\lor}}
                        {\prfbyaxiom{\rl[\{\nm{w}{\alpha}, \nm{u}{\dl{\alpha}}\}]{ax}}
                            {\th \nm{z}{\dl{\alpha}}, \nm{w}{\tikzmarknode{w-1}{\alpha}}, \nm{u}{\dl{\tikzmarknode{u-4}{\alpha}}}}}
                        {\th \nm{z}{\dl{\alpha}} \lor \nm{w}{\tikzmarknode{w-2}{\alpha}}, \nm{u}{\dl{\tikzmarknode{u-5}{\alpha}}}}}
                    {\th \nm{z}{\dl{\alpha}} \lor \nm{w}{\tikzmarknode{w-3}{\alpha}}}}
                {\th (\nm{x}{\dl{\tikzmarknode{x-4}{\alpha}}} \lor \nm{y}{\alpha}) \land (\nm{z}{\dl{\alpha}} \lor \nm{w}{\tikzmarknode{w-4}{\alpha}})}
        }%
        \begin{tikzpicture}[remember picture, overlay, thick, blue]
            % Axioms
            \draw ([yshift=2pt]x-1.north) .. controls ++(0,1.3em) and ++(0,1.3em) .. ([yshift=2pt]u-1.north);
            \draw ([yshift=2pt]w-1.north) .. controls ++(0,1.3em) and ++(0,1.3em) .. ([yshift=2pt]u-4.north);
            % History of x ([xshift=-.5pt]x-4.west)
            \draw ([yshift=2pt]x-4.north) .. controls ++(0,15pt) and ++(0,-15pt) .. ([yshift=-1pt]x-3.south);
            \draw ([yshift=2pt]x-3.north) .. controls ++(0,8pt) and ++(0,-8pt) .. ([yshift=-1pt]x-2.south);
            \draw ([yshift=2pt]x-2.north) .. controls ++(0,2pt) and ++(0,-2pt) .. ([yshift=-1pt]x-1.south);
            % History of u (left)
            \draw ([yshift=1pt]u-2.north) .. controls ++(0,4pt) and ++(0,-4pt) .. ([yshift=-1pt]u-1.south);
            % History of w
            \draw ([yshift=1pt]w-4.north) .. controls ++(0,15pt) and ++(0,-15pt) .. ([yshift=-1pt]w-3.south);
            \draw ([yshift=1pt]w-3.north) .. controls ++(0,10pt) and ++(0,-7pt) .. ([yshift=-1pt]w-2.south);
            \draw ([yshift=1pt]w-2.north) .. controls ++(0,3pt) and ++(0,-3pt) .. ([yshift=-1pt]w-1.south);
            % History of u (right)
            \draw ([yshift=2pt]u-5.north) .. controls ++(0,3pt) and ++(0,-3pt) .. ([yshift=-1pt]u-4.south);
            % Cuts
            \path[draw, densely dashed, semithick, rounded corners=1em]
                ([yshift=-1pt]u-2.south) -- ([yshift=-3em]u-2.south) -- ([shift={(-5em,-3em)}]u-3.south)
                -- ([shift={(-5em,-1em)}]u-3.south) -- ([yshift=-1em]u-3.south) -- ([yshift=-1pt]u-3.south);
            \path[draw, densely dashed, semithick, rounded corners=1em]
                ([yshift=-1pt]u-5.south) -- ([yshift=-3em]u-5.south) -- ([yshift=-3em]u-6.south) -- ([yshift=-1pt]u-6.south);
        \end{tikzpicture}%
        \\
        \caption{The transformed derivation, after isolating the conjuction. The two halves of the original alternating path
            are now disconnected.}
        \label{fig:axiom-graph-inequality-example1:post}
    \end{subfigure}
    \par\vspace{4em}%
    \begin{subfigure}{\textwidth}
        \centering
        \begin{tikzpicture}
            \graph[math nodes]{
                x --[bend left] y --[bend right] z --[bend left] w;
                x --[bend right] w;
            };
        \end{tikzpicture}%
        \qquad\qquad\qquad\qquad%
        \begin{tikzpicture}
            \graph[math nodes]{
                x --[bend left] y -!- z --[bend left] w;
                x --[bend right, opacity=0] w;
            };
        \end{tikzpicture}\\[1em]
        \caption{The axiom graph of the original derivation~(left) and that of the transformed
            one~(right).}
        \label{fig:axiom-graph-inequality-example1:graphs}
    \end{subfigure}
    \par\vspace{4em}%
    \caption{A derivation with cuts whose axiom graph decreases when isolating the conjunction
        in its conclusion.}
    \label{fig:axiom-graph-inequality-example1}
\end{figure}
\vfill
\clearpage
}

{
\addtolength{\columnwidth}{\pdfpagewidth}
\addtolength{\textwidth}{\pdfpagewidth}
\pdfpagewidth=2\pdfpagewidth
\null
\vfill
\begin{figure}[!h]
    \begin{subfigure}{\textwidth}
        \noindent\makebox[\textwidth]{
            \prftree[r]{\rl{cut}}
            {\prftree[r]{\rl{\land}}
                {\prftree[r]{\rl{\land}}
                    {\prfbyaxiom{\rl[\{\nm{x}{\dl{\alpha}}, \nm{t}{\alpha}\}]{ax}}
                        {\hvp\th \nm{x}{\dl{\alpha}}, \nm{y}{\dl{\alpha}}, \nm{z}{\dl{\beta}}, \nm{t}{\alpha}, \nm{v}{\alpha}}}
                    {\prfbyaxiom{\rl[\{\nm{x}{\dl{\alpha}}, \nm{t}{\alpha}\}]{ax}}
                        {\hvp\th \nm{x}{\dl{\alpha}}, \nm{y}{\dl{\alpha}}, \nm{z}{\dl{\beta}}, \nm{t}{\alpha}, \nm{w}{\beta}}}
                    {\hvp\th \nm{x}{\dl{\alpha}}, \nm{y}{\dl{\alpha}}, \nm{z}{\dl{\beta}}, \nm{t}{\alpha}, \nm{v}{\alpha} \land \nm{w}{\beta}}}
                {\prftree[r]{\rl{\land}}
                    {\prfbyaxiom{\rl[\{\nm{y}{\dl{\alpha}}, \nm{v}{\alpha}\}]{ax}}
                        {\hvp\th \nm{x}{\dl{\alpha}}, \nm{y}{\dl{\tikzmarknode{y-l-1}{\alpha}}}, \nm{z}{\dl{\beta}}, \nm{u}{\beta}, \nm{v}{\tikzmarknode{v-l-1}{\alpha}}}}
                    {\prfbyaxiom{\rl[\{\nm{z}{\dl{\beta}}, \nm{w}{\beta}\}]{ax}}
                        {\hvp\th \nm{x}{\dl{\alpha}}, \nm{y}{\dl{\alpha}}, \nm{z}{\dl{\beta}}, \nm{u}{\beta}, \nm{w}{\beta}}}
                    {\hvp\th \nm{x}{\dl{\alpha}}, \nm{y}{\dl{\tikzmarknode{y-l-2}{\alpha}}}, \nm{z}{\dl{\beta}}, \nm{u}{\beta}, \nm{v}{\tikzmarknode{v-l-2}{\alpha}} \land \nm{w}{\beta}}}
                {\th \nm{x}{\dl{\alpha}}, \nm{y}{\dl{\tikzmarknode{y-l-3}{\alpha}}}, \nm{z}{\dl{\beta}}, \nm{t}{\alpha} \land \nm{u}{\beta}, \nm{v}{\tikzmarknode{v-l-3}{\alpha}} \land \nm{w}{\beta}}}
            {\prftree[r]{\rl{\land}}
                {\prftree[r]{\rl{\lor}}
                    {\prfbyaxiom{\rl[\{\nm{t}{\alpha}, \nm{v}{\dl{\alpha}}\}]{ax}}
                        {\hvp\th \nm{x}{\dl{\alpha}}, \nm{y}{\dl{\alpha}}, \nm{z}{\dl{\beta}}, \nm{t}{\tikzmarknode{t-r-1}{\alpha}}, \nm{v}{\dl{\tikzmarknode{v-r-1}{\alpha}}}, \nm{w}{\dl{\beta}}}}
                    {\hvp\th \nm{x}{\dl{\alpha}}, \nm{y}{\dl{\alpha}}, \nm{z}{\dl{\beta}}, \nm{t}{\tikzmarknode{t-r-2}{\alpha}}, \nm{v}{\dl{\tikzmarknode{v-r-2}{\alpha}}} \lor \nm{w}{\dl{\beta}}}}
                {\prftree[r]{\rl{\lor}}
                    {\prfbyaxiom{\rl[\{\nm{u}{\beta}, \nm{w}{\dl{\beta}}\}]{ax}}
                        {\hvp\th \nm{x}{\dl{\alpha}}, \nm{y}{\dl{\alpha}}, \nm{z}{\dl{\beta}}, \nm{u}{\beta}, \nm{v}{\dl{\alpha}}, \nm{w}{\dl{\beta}}}}
                    {\hvp\th \nm{x}{\dl{\alpha}}, \nm{y}{\dl{\alpha}}, \nm{z}{\dl{\beta}}, \nm{u}{\beta}, \nm{v}{\dl{\alpha}} \lor \nm{w}{\dl{\beta}}}}
                {\hvp\th \nm{x}{\dl{\alpha}}, \nm{y}{\dl{\alpha}}, \nm{z}{\dl{\beta}}, \nm{t}{\tikzmarknode{t-r-3}{\alpha}} \land \nm{u}{\beta}, \nm{v}{\dl{\tikzmarknode{v-r-3}{\alpha}}} \lor \nm{w}{\dl{\beta}}}}
            {\hvp\th \nm{x}{\dl{\alpha}}, \nm{y}{\dl{\tikzmarknode{y-4}{\alpha}}}, \nm{z}{\dl{\beta}}, \nm{t}{\tikzmarknode{t-4}{\alpha}} \land \nm{u}{\beta}}
        }%
        \begin{tikzpicture}[remember picture, overlay, thick, blue]
            % Axioms
            \draw ([yshift=2pt]y-l-1.north) .. controls ++(0,1.3em) and ++(0,1.3em) .. ([yshift=2pt]v-l-1.north);
            \draw ([yshift=2pt]t-r-1.north) .. controls ++(0,1.3em) and ++(0,1.3em) .. ([yshift=2pt]v-r-1.north);
            % History of y (left)
            \draw ([yshift=2pt]y-4.north) .. controls ++(0,30pt) and ++(0,-20pt) .. ([yshift=-1pt]y-l-3.south);
            \draw ([yshift=2pt]y-l-3.north) .. controls ++(0,25pt) and ++(0,-10pt) .. ([yshift=-1pt]y-l-2.south);
            \draw ([yshift=2pt]y-l-2.north) .. controls ++(0,10pt) and ++(0,-10pt) .. ([yshift=-1pt]y-l-1.south);
            % History of v (left)
            \draw ([yshift=1pt]v-l-3.north) .. controls ++(0,15pt) and ++(0,-20pt) .. ([yshift=-1pt]v-l-2.south);
            \draw ([yshift=1pt]v-l-2.north) .. controls ++(0,10pt) and ++(0,-10pt) .. ([yshift=-1pt]v-l-1.south);
            % History of t (right)
            \draw ([yshift=1pt]t-4.north) .. controls ++(0,30pt) and ++(0,-20pt) .. ([yshift=-1pt]t-r-3.south);
            \draw ([yshift=1pt]t-r-3.north) .. controls ++(0,10pt) and ++(0,-10pt) .. ([yshift=-1pt]t-r-2.south);
            \draw ([yshift=1pt]t-r-2.north) .. controls ++(0,7pt) and ++(0,-7pt) .. ([yshift=-1pt]t-r-1.south);
            % History of v (right)
            \draw ([yshift=2pt]v-r-3.north) .. controls ++(0,15pt) and ++(0,-7pt) .. ([yshift=-1pt]v-r-2.south);
            \draw ([yshift=2pt]v-r-2.north) .. controls ++(0,7pt) and ++(0,-7pt) .. ([yshift=-1pt]v-r-1.south);
            % Cuts
            \path[draw, densely dashed, semithick, rounded corners=1em]
                ([yshift=-1pt]v-l-3.south) -- ([yshift=-3.3em]v-l-3.south) -- ([yshift=-3.3em]v-r-3.south) -- ([yshift=-1pt]v-r-3.south);
        \end{tikzpicture}%
        \\[.5em]
        \caption{The initial derivation, where the conjunction in the conclusion is not
            introduced by the last rule. Assume names \(x, y, z, t, u, v, w \in \mathcal{N}\)
            are pairwise distinct.}
        \label{fig:axiom-graph-inequality-example2:pre}
    \end{subfigure}
    \par\vspace{4em}%
    \begin{subfigure}{\textwidth}
        \noindent\makebox[\textwidth]{
            \prftree[r]{\rl{\land}}
                {\prftree[r]{\rl{cut}}
                    {\prftree[r]{\rl{\land}}
                        {\prfbyaxiom{\rl[\{\nm{x}{\dl{\alpha}}, \nm{t}{\alpha}\}]{ax}}
                            {\hvp\th \nm{x}{\dl{\alpha}}, \nm{y}{\dl{\alpha}}, \nm{z}{\dl{\beta}}, \nm{t}{\alpha}, \nm{v}{\alpha}}}
                        {\prfbyaxiom{\rl[\{\nm{x}{\dl{\alpha}}, \nm{t}{\alpha}\}]{ax}}
                            {\hvp\th \nm{x}{\dl{\alpha}}, \nm{y}{\dl{\alpha}}, \nm{z}{\dl{\beta}}, \nm{t}{\alpha}, \nm{w}{\beta}}}
                        {\hvp\th \nm{x}{\dl{\alpha}}, \nm{y}{\dl{\alpha}}, \nm{z}{\dl{\beta}}, \nm{t}{\alpha}, \nm{v}{\tikzmarknode{v-l}{\alpha}} \land \nm{w}{\beta}}}
                    {\prftree[r]{\rl{\lor}}
                        {\prfbyaxiom{\rl[\{\nm{t}{\alpha}, \nm{v}{\dl{\alpha}}\}]{ax}}
                            {\hvp\th \nm{x}{\dl{\alpha}}, \nm{y}{\dl{\alpha}}, \nm{z}{\dl{\beta}}, \nm{t}{\tikzmarknode{t-r-1}{\alpha}}, \nm{v}{\dl{\tikzmarknode{v-r-1}{\alpha}}}, \nm{w}{\dl{\beta}}}}
                        {\hvp\th \nm{x}{\dl{\alpha}}, \nm{y}{\dl{\alpha}}, \nm{z}{\dl{\beta}}, \nm{t}{\tikzmarknode{t-r-2}{\alpha}}, \nm{v}{\dl{\tikzmarknode{v-r-2}{\alpha}}} \lor \nm{w}{\dl{\beta}}}}
                    {\hvp\th \nm{x}{\dl{\alpha}}, \nm{y}{\dl{\alpha}}, \nm{z}{\dl{\beta}}, \nm{t}{\tikzmarknode{t-r-3}{\alpha}}}}
                {\prftree[r]{\rl{cut}}
                    {\prftree[r]{\rl{\land}}
                        {\prfbyaxiom{\rl[\{\nm{y}{\dl{\alpha}}, \nm{v}{\alpha}\}]{ax}}
                            {\hvp\th \nm{x}{\dl{\alpha}}, \nm{y}{\dl{\tikzmarknode{y-l-1}{\alpha}}}, \nm{z}{\dl{\beta}}, \nm{u}{\beta}, \nm{v}{\tikzmarknode{v-l-1}{\alpha}}}}
                        {\prfbyaxiom{\rl[\{\nm{z}{\dl{\beta}}, \nm{w}{\beta}\}]{ax}}
                            {\hvp\th \nm{x}{\dl{\alpha}}, \nm{y}{\dl{\alpha}}, \nm{z}{\dl{\beta}}, \nm{u}{\beta}, \nm{w}{\beta}}}
                        {\hvp\th \nm{x}{\dl{\alpha}}, \nm{y}{\dl{\tikzmarknode{y-l-2}{\alpha}}}, \nm{z}{\dl{\beta}}, \nm{u}{\beta}, \nm{v}{\tikzmarknode{v-l-2}{\alpha}} \land \nm{w}{\beta}}}
                    {\prftree[r]{\rl{\lor}}
                        {\prfbyaxiom{\rl[\{\nm{u}{\beta}, \nm{w}{\dl{\beta}}\}]{ax}}
                            {\hvp\th \nm{x}{\dl{\alpha}}, \nm{y}{\dl{\alpha}}, \nm{z}{\dl{\beta}}, \nm{u}{\beta}, \nm{v}{\dl{\alpha}}, \nm{w}{\dl{\beta}}}}
                        {\hvp\th \nm{x}{\dl{\alpha}}, \nm{y}{\dl{\alpha}}, \nm{z}{\dl{\beta}}, \nm{u}{\beta}, \nm{v}{\dl{\tikzmarknode{v-r}{\alpha}}} \lor \nm{w}{\dl{\beta}}}}
                    {\hvp\th \nm{x}{\dl{\alpha}}, \nm{y}{\dl{\tikzmarknode{y-l-3}{\alpha}}}, \nm{z}{\dl{\beta}}, \nm{u}{\beta}}}
                {\hvp\th \nm{x}{\dl{\alpha}}, \nm{y}{\dl{\tikzmarknode{y-4}{\alpha}}}, \nm{z}{\dl{\beta}}, \nm{t}{\tikzmarknode{t-4}{\alpha}} \land \nm{u}{\beta}}
        }%
        \begin{tikzpicture}[remember picture, overlay, thick, blue]
            % Axioms
            \draw ([yshift=2pt]y-l-1.north) .. controls ++(0,1.3em) and ++(0,1.3em) .. ([yshift=2pt]v-l-1.north);
            \draw ([yshift=2pt]t-r-1.north) .. controls ++(0,1.3em) and ++(0,1.3em) .. ([yshift=2pt]v-r-1.north);
            % History of y (left)
            \draw ([yshift=2pt]y-4.north) .. controls ++(0,30pt) and ++(0,-20pt) .. ([yshift=-1pt]y-l-3.south);
            \draw ([yshift=2pt]y-l-3.north) .. controls ++(0,25pt) and ++(0,-10pt) .. ([yshift=-1pt]y-l-2.south);
            \draw ([yshift=2pt]y-l-2.north) .. controls ++(0,10pt) and ++(0,-10pt) .. ([yshift=-1pt]y-l-1.south);
            % History of v (left)
            \draw ([yshift=1pt]v-l-2.north) .. controls ++(0,10pt) and ++(0,-10pt) .. ([yshift=-1pt]v-l-1.south);
            % History of t (right)
            \draw ([yshift=1pt]t-4.north) .. controls ++(0,30pt) and ++(0,-20pt) .. ([yshift=-1pt]t-r-3.south);
            \draw ([yshift=1pt]t-r-3.north) .. controls ++(0,10pt) and ++(0,-10pt) .. ([yshift=-1pt]t-r-2.south);
            \draw ([yshift=1pt]t-r-2.north) .. controls ++(0,7pt) and ++(0,-7pt) .. ([yshift=-1pt]t-r-1.south);
            % History of v (right)
            \draw ([yshift=2pt]v-r-2.north) .. controls ++(0,7pt) and ++(0,-7pt) .. ([yshift=-1pt]v-r-1.south);
            % Cuts
            \path[draw, densely dashed, semithick, rounded corners=1em]
                 ([yshift=-1pt]v-l-2.south) -- ([yshift=-4.7em]v-l-2.south) -- ([yshift=-4.7em]v-r.south) -- ([yshift=-1pt]v-r.south);
            \path[draw, densely dashed, semithick, rounded corners=1em]
                 ([yshift=-1pt]v-r-2.south) -- ([yshift=-4.7em]v-r-2.south) -- ([yshift=-4.7em]v-l.south) -- ([yshift=-1pt]v-l.south);
        \end{tikzpicture}%
        \\
        \caption{The transformed derivation, after isolating the conjuction. The two halves of the original alternating path
            are now disconnected.}
        \label{fig:axiom-graph-inequality-example2:post}
    \end{subfigure}
    \par\vspace{4em}%
    \begin{subfigure}{\textwidth}
        \centering
        \begin{tikzpicture}
            \graph[math nodes]{
                x -!- y -!- z -!- t -!- u;
                x --[bend left] t;
                z --[bend left] u;
                y --[bend right] t;
            };
        \end{tikzpicture}%
        \qquad\qquad\qquad\qquad%
        \begin{tikzpicture}
            \graph[math nodes]{
                x -!- y -!- z -!- t -!- u;
                x --[bend left] t;
                z --[bend left] u;
                y --[bend right, opacity=0] t;
            };
        \end{tikzpicture}\\[1em]
        \caption{The axiom graph of the original derivation~(left) and that of the transformed
            one~(right).}
        \label{fig:axiom-graph-inequality-example2:graphs}
    \end{subfigure}
    \par\vspace{4em}%
    \caption{A derivation with cuts whose axiom graph decreases when isolating the conjunction
        in its conclusion. The lost path does not visibly cross a conjunction.}
    \label{fig:axiom-graph-inequality-example2}
\end{figure}
\vfill
\clearpage
}

{
\addtolength{\columnwidth}{\pdfpagewidth}
\addtolength{\textwidth}{\pdfpagewidth}
\pdfpagewidth=2\pdfpagewidth
\null
\vfill
\begin{figure}[!h]
    \begin{subfigure}{\textwidth}
        \noindent\makebox[\textwidth]{
            \prftree[r]{\rl{cut}}
                {\prftree[r]{\rl{\lor}}
                    {\prftree[r]{\rl{\land}}
                        {\prftree[r]{\rl{\land}}
                            {\prfbyaxiom{\rl[\{\nm{z}{\alpha},\nm{t}{\dl{\alpha}}\}]{ax}}
                                {\hvp\th \nm{x}{\beta}, \nm{y}{\dl{\beta}}, \nm{z}{\tikzmarknode{z-l-1}{\alpha}}, \nm{t}{\dl{\tikzmarknode{t-l-1}{\alpha}}}, \nm{v}{\dl{\alpha}}, \nm{w}{\alpha}}}
                            {\prfbyaxiom{\rl[\{\nm{x}{\beta},\nm{y}{\dl{\beta}}\}]{ax}}
                                {\hvp\th \nm{x}{\beta}, \nm{y}{\dl{\beta}}, \nm{z}{\alpha}, \nm{u}{\alpha}, \nm{v}{\dl{\alpha}}, \nm{w}{\alpha}}}
                            {\hvp\th \nm{x}{\beta}, \nm{y}{\dl{\beta}}, \nm{z}{\tikzmarknode{z-l-2}{\alpha}}, \nm{t}{\dl{\tikzmarknode{t-l-2}{\alpha}}} \land \nm{u}{\alpha}, \nm{v}{\dl{\alpha}}, \nm{w}{\alpha}}}
                        {\prftree[r]{\rl{\land}}
                            {\prfbyaxiom{\rl[\{\nm{x}{\beta},\nm{y}{\dl{\beta}}\}]{ax}}
                                {\hvp\th \nm{x}{\beta}, \nm{y}{\dl{\beta}}, \nm{s}{\dl{\alpha}}, \nm{t}{\dl{\alpha}}, \nm{v}{\dl{\alpha}}, \nm{w}{\alpha}}}
                            {\prfbyaxiom{\rl[\{\nm{s}{\dl{\alpha}},\nm{u}{\alpha}\}]{ax}}
                                {\hvp\th \nm{x}{\beta}, \nm{y}{\dl{\beta}}, \nm{s}{\dl{\tikzmarknode{s-l-1}{\alpha}}}, \nm{u}{\tikzmarknode{u-l-1}{\alpha}}, \nm{v}{\dl{\alpha}}, \nm{w}{\alpha}}}
                            {\hvp\th \nm{x}{\beta}, \nm{y}{\dl{\beta}}, \nm{s}{\dl{\tikzmarknode{s-l-2}{\alpha}}}, \nm{t}{\dl{\alpha}} \land \nm{u}{\tikzmarknode{u-l-2}{\alpha}}, \nm{v}{\dl{\alpha}}, \nm{w}{\alpha}}}
                        {\hvp\th \nm{x}{\beta}, \nm{y}{\dl{\beta}}, \nm{z}{\tikzmarknode{z-l-3}{\alpha}} \land \nm{s}{\dl{\tikzmarknode{s-l-3}{\alpha}}}, \nm{t}{\dl{\tikzmarknode{t-l-3}{\alpha}}} \land \nm{u}{\tikzmarknode{u-l-3}{\alpha}}, \nm{v}{\dl{\alpha}}, \nm{w}{\alpha}}}
                    {\hvp\th \nm{x}{\beta}, \nm{y}{\dl{\beta}}, (\nm{z}{\tikzmarknode{z-l-4}{\alpha}} \land \nm{s}{\dl{\tikzmarknode{s-l-4}{\alpha}}}) \lor (\nm{t}{\dl{\tikzmarknode{t-l-4}{\alpha}}} \land \nm{u}{\tikzmarknode{u-l-4}{\alpha}}), \nm{v}{\dl{\alpha}}, \nm{w}{\alpha}}}
                {\prftree[r]{\rl{\land}}
                    {\prftree[r]{\rl{\lor}}
                        {\prfbyaxiom{\rl[\{\nm{z}{\dl{\alpha}},\nm{s}{\alpha}\}]{ax}}
                            {\hvp\th \nm{x}{\beta}, \nm{y}{\dl{\beta}}, \nm{z}{\dl{\tikzmarknode{z-r-1}{\alpha}}}, \nm{s}{\tikzmarknode{s-r-1}{\alpha}}, \nm{v}{\dl{\alpha}}, \nm{w}{\alpha}}}
                        {\hvp\th \nm{x}{\beta}, \nm{y}{\dl{\beta}}, \nm{z}{\dl{\tikzmarknode{z-r-2}{\alpha}}} \lor \nm{s}{\tikzmarknode{s-r-2}{\alpha}}, \nm{v}{\dl{\alpha}}, \nm{w}{\alpha}}}
                    {\prftree[r]{\rl{\lor}}
                        {\prftree[r]{\rl{\sqcup}}
                            {\prfbyaxiom{\rl[\{\nm{t}{\alpha},\nm{v}{\dl{\alpha}}\}]{ax}}
                                {\hvp\th \nm{x}{\beta}, \nm{y}{\dl{\beta}}, \nm{t}{\tikzmarknode{t-r-1}{\alpha}}, \nm{u}{\dl{\alpha}}, \nm{v}{\dl{\tikzmarknode{v-r-1}{\alpha}}}, \nm{w}{\alpha}}}
                            {\prfbyaxiom{\rl[\{\nm{u}{\dl{\alpha}},\nm{w}{\alpha}\}]{ax}}
                                {\hvp\th \nm{x}{\beta}, \nm{y}{\dl{\beta}}, \nm{t}{\alpha}, \nm{u}{\dl{\tikzmarknode{u-r-1}{\alpha}}}, \nm{v}{\dl{\alpha}}, \nm{w}{\tikzmarknode{w-r-1}{\alpha}}}}
                            {\hvp\th \nm{x}{\beta}, \nm{y}{\dl{\beta}}, \nm{t}{\tikzmarknode{t-r-2}{\alpha}}, \nm{u}{\dl{\tikzmarknode{u-r-2}{\alpha}}}, \nm{v}{\dl{\tikzmarknode{v-r-2}{\alpha}}}, \nm{w}{\tikzmarknode{w-r-2}{\alpha}}}}
                        {\hvp\th \nm{x}{\beta}, \nm{y}{\dl{\beta}}, \nm{t}{\tikzmarknode{t-r-3}{\alpha}} \lor \nm{u}{\dl{\tikzmarknode{u-r-3}{\alpha}}}, \nm{v}{\dl{\tikzmarknode{v-r-3}{\alpha}}}, \nm{w}{\tikzmarknode{w-r-3}{\alpha}}}}
                    {\hvp\th \nm{x}{\beta}, \nm{y}{\dl{\beta}}, (\nm{z}{\dl{\tikzmarknode{z-r-3}{\alpha}}} \lor \nm{s}{\tikzmarknode{s-r-3}{\alpha}}) \land (\nm{t}{\tikzmarknode{t-r-4}{\alpha}} \lor \nm{u}{\dl{\tikzmarknode{u-r-4}{\alpha}}}), \nm{v}{\dl{\tikzmarknode{v-r-4}{\alpha}}}, \nm{w}{\tikzmarknode{w-r-4}{\alpha}}}}
                {\hvp\th \nm{x}{\beta}, \nm{y}{\dl{\beta}}, \nm{v}{\dl{\tikzmarknode{v-5}{\alpha}}}, \nm{w}{\tikzmarknode{w-5}{\alpha}}}
        }%
        \begin{tikzpicture}[remember picture, overlay, thick, blue]
            % Axioms
            \draw ([yshift=2pt]z-l-1.north) .. controls ++(0,1.3em) and ++(0,1.3em) .. ([yshift=2pt]t-l-1.north);
            \draw ([yshift=2pt]s-l-1.north) .. controls ++(0,1.3em) and ++(0,1.3em) .. ([yshift=2pt]u-l-1.north);
            \draw ([yshift=2pt]z-r-1.north) .. controls ++(0,1.3em) and ++(0,1.3em) .. ([yshift=2pt]s-r-1.north);
            \draw ([yshift=2pt]t-r-1.north) .. controls ++(0,1.3em) and ++(0,1.3em) .. ([yshift=2pt]v-r-1.north);
            \draw ([yshift=2pt]u-r-1.north) .. controls ++(0,1.3em) and ++(0,1.3em) .. ([yshift=2pt]w-r-1.north);
            % History of z (left)
            \draw ([yshift=1pt]z-l-4.north) .. controls ++(0,7pt) and ++(0,-7pt) .. ([yshift=-1pt]z-l-3.south);
            \draw ([yshift=1pt]z-l-3.north) .. controls ++(0,15pt) and ++(0,-15pt) .. ([yshift=-1pt]z-l-2.south);
            \draw ([yshift=1pt]z-l-2.north) .. controls ++(0,10pt) and ++(0,-10pt) .. ([yshift=-1pt]z-l-1.south);
            % History of s (left)
            \draw ([yshift=2pt]s-l-4.north) .. controls ++(0,7pt) and ++(0,-7pt) .. ([yshift=-1pt]s-l-3.south);
            \draw ([yshift=2pt]s-l-3.north) .. controls ++(0,20pt) and ++(0,-10pt) .. ([yshift=-1pt]s-l-2.south);
            \draw ([yshift=2pt]s-l-2.north) .. controls ++(0,15pt) and ++(0,-5pt) .. ([yshift=-1pt]s-l-1.south);
            % History of t (left)
            \draw ([yshift=2pt]t-l-4.north) .. controls ++(0,7pt) and ++(0,-7pt) .. ([yshift=-1pt]t-l-3.south);
            \draw ([yshift=2pt]t-l-3.north) .. controls ++(0,20pt) and ++(0,-10pt) .. ([yshift=-1pt]t-l-2.south);
            \draw ([yshift=2pt]t-l-2.north) .. controls ++(0,15pt) and ++(0,-5pt) .. ([yshift=-1pt]t-l-1.south);
            % History of u (left)
            \draw ([yshift=1pt]u-l-4.north) .. controls ++(0,7pt) and ++(0,-7pt) .. ([yshift=-1pt]u-l-3.south);
            \draw ([yshift=1pt]u-l-3.north) .. controls ++(0,15pt) and ++(0,-15pt) .. ([yshift=-1pt]u-l-2.south);
            \draw ([yshift=1pt]u-l-2.north) .. controls ++(0,10pt) and ++(0,-10pt) .. ([yshift=-1pt]u-l-1.south);
            % History of z (right)
            \draw ([yshift=2pt]z-r-3.north) .. controls ++(0,15pt) and ++(0,-15pt) .. ([yshift=-1pt]z-r-2.south);
            \draw ([yshift=2pt]z-r-2.north) .. controls ++(0,7pt) and ++(0,-7pt) .. ([yshift=-1pt]z-r-1.south);
            % History of s (right)
            \draw ([yshift=1pt]s-r-3.north) .. controls ++(0,20pt) and ++(0,-10pt) .. ([yshift=-1pt]s-r-2.south);
            \draw ([yshift=1pt]s-r-2.north) .. controls ++(0,7pt) and ++(0,-7pt) .. ([yshift=-1pt]s-r-1.south);
            % History of t (right)
            \draw ([yshift=1pt]t-r-4.north) .. controls ++(0,25pt) and ++(0,-10pt) .. ([yshift=-1pt]t-r-3.south);
            \draw ([yshift=1pt]t-r-3.north) .. controls ++(0,7pt) and ++(0,-7pt) .. ([yshift=-1pt]t-r-2.south);
            \draw ([yshift=1pt]t-r-2.north) .. controls ++(0,10pt) and ++(0,-10pt) .. ([yshift=-1pt]t-r-1.south);
            % History of u (right)
            \draw ([yshift=2pt]u-r-4.north) .. controls ++(0,20pt) and ++(0,-15pt) .. ([yshift=-1pt]u-r-3.south);
            \draw ([yshift=2pt]u-r-3.north) .. controls ++(0,7pt) and ++(0,-7pt) .. ([yshift=-1pt]u-r-2.south);
            \draw ([yshift=2pt]u-r-2.north) .. controls ++(0,15pt) and ++(0,-5pt) .. ([yshift=-1pt]u-r-1.south);
            % History of v (right)
            \draw ([yshift=2pt]v-5.north) .. controls ++(0,35pt) and ++(0,-25pt) .. ([yshift=-1pt]v-r-4.south);
            \draw ([yshift=2pt]v-r-4.north) .. controls ++(0,15pt) and ++(0,-20pt) .. ([yshift=-1pt]v-r-3.south);
            \draw ([yshift=2pt]v-r-3.north) .. controls ++(0,7pt) and ++(0,-7pt) .. ([yshift=-1pt]v-r-2.south);
            \draw ([yshift=2pt]v-r-2.north) .. controls ++(0,15pt) and ++(0,-5pt) .. ([yshift=-1pt]v-r-1.south);
            % History of w (right)
            \draw ([yshift=1pt]w-5.north) .. controls ++(0,20pt) and ++(0,-30pt) .. ([yshift=-1pt]w-r-4.south);
            \draw ([yshift=1pt]w-r-4.north) .. controls ++(0,10pt) and ++(0,-25pt) .. ([yshift=-1pt]w-r-3.south);
            \draw ([yshift=1pt]w-r-3.north) .. controls ++(0,7pt) and ++(0,-7pt) .. ([yshift=-1pt]w-r-2.south);
            \draw ([yshift=1pt]w-r-2.north) .. controls ++(0,10pt) and ++(0,-10pt) .. ([yshift=-1pt]w-r-1.south);
            % Cuts
            \path[draw, densely dashed, semithick, rounded corners=1em]
                ([yshift=-1pt]z-l-4.south) -- ([yshift=-5em]z-l-4.south) -- ([yshift=-5em]z-r-3.south) -- ([yshift=-1pt]z-r-3.south);
            \path[draw, densely dashed, semithick, rounded corners=1em]
                ([yshift=-1pt]s-l-4.south) -- ([yshift=-6em]s-l-4.south) -- ([yshift=-6em]s-r-3.south) -- ([yshift=-1pt]s-r-3.south);
            \path[draw, densely dashed, semithick, rounded corners=1em]
                ([yshift=-1pt]t-l-4.south) -- ([yshift=-3em]t-l-4.south) -- ([yshift=-3em]t-r-4.south) -- ([yshift=-1pt]t-r-4.south);
            \path[draw, densely dashed, semithick, rounded corners=1em]
                ([yshift=-1pt]u-l-4.south) -- ([yshift=-4em]u-l-4.south) -- ([yshift=-4em]u-r-4.south) -- ([yshift=-1pt]u-r-4.south);
        \end{tikzpicture}%
        \\[3em]
        \caption{The initial derivation, with a unique cut-rule. Assume names \(x, y, z, s, t, u,
            v, w \in \mathcal{N}\) are pairwise distinct. There is an alternating path
            connecting~\(\nm{v}{\dl{\alpha}}\) with~\(\nm{w}{\alpha}\). Because the conclusion
            is atomic, there is only one branch name up to names in the interface, hence all
            edges are compatible.}
        \label{fig:cut-reduction-failure:pre}
    \end{subfigure}
    \par\vspace{4em}%
    \begin{subfigure}{\textwidth}
        \noindent\makebox[\textwidth]{
            \prftree[r]{\rl{cut}}
                {\prftree[r]{\rl{cut}}
                    {\prftree[r]{\rl{\land}}
                        {\prftree[r]{\rl{\land}}
                            {\prfbyaxiom{\rl[\{\nm{z}{\alpha},\nm{t}{\dl{\alpha}}\}]{ax}}
                                {\hvp\th \nm{x}{\beta}, \nm{y}{\dl{\beta}}, \nm{z}{\tikzmarknode{z-l-1}{\alpha}}, \nm{t}{\dl{\tikzmarknode{t-l-1}{\alpha}}}, \nm{v}{\dl{\alpha}}, \nm{w}{\alpha}}}
                            {\prfbyaxiom{\rl[\{\nm{x}{\beta},\nm{y}{\dl{\beta}}\}]{ax}}
                                {\hvp\th \nm{x}{\beta}, \nm{y}{\dl{\beta}}, \nm{z}{\alpha}, \nm{u}{\alpha}, \nm{v}{\dl{\alpha}}, \nm{w}{\alpha}}}
                            {\hvp\th \nm{x}{\beta}, \nm{y}{\dl{\beta}}, \nm{z}{\tikzmarknode{z-l-2}{\alpha}}, \nm{t}{\dl{\tikzmarknode{t-l-2}{\alpha}}} \land \nm{u}{\alpha}, \nm{v}{\dl{\alpha}}, \nm{w}{\alpha}}}
                        {\prftree[r]{\rl{\land}}
                            {\prfbyaxiom{\rl[\{\nm{x}{\beta},\nm{y}{\dl{\beta}}\}]{ax}}
                                {\hvp\th \nm{x}{\beta}, \nm{y}{\dl{\beta}}, \nm{s}{\dl{\alpha}}, \nm{t}{\dl{\alpha}}, \nm{v}{\dl{\alpha}}, \nm{w}{\alpha}}}
                            {\prfbyaxiom{\rl[\{\nm{s}{\dl{\alpha}},\nm{u}{\alpha}\}]{ax}}
                                {\hvp\th \nm{x}{\beta}, \nm{y}{\dl{\beta}}, \nm{s}{\dl{\tikzmarknode{s-l-1}{\alpha}}}, \nm{u}{\tikzmarknode{u-l-1}{\alpha}}, \nm{v}{\dl{\alpha}}, \nm{w}{\alpha}}}
                            {\hvp\th \nm{x}{\beta}, \nm{y}{\dl{\beta}}, \nm{s}{\dl{\tikzmarknode{s-l-2}{\alpha}}}, \nm{t}{\dl{\alpha}} \land \nm{u}{\tikzmarknode{u-l-2}{\alpha}}, \nm{v}{\dl{\alpha}}, \nm{w}{\alpha}}}
                        {\hvp\th \nm{x}{\beta}, \nm{y}{\dl{\beta}}, \nm{z}{\tikzmarknode{z-l-3}{\alpha}} \land \nm{s}{\dl{\tikzmarknode{s-l-3}{\alpha}}}, \nm{t}{\dl{\tikzmarknode{t-l-3}{\alpha}}} \land \nm{u}{\tikzmarknode{u-l-3}{\alpha}}, \nm{v}{\dl{\alpha}}, \nm{w}{\alpha}}}
                    {\prftree[r]{\rl{\lor}}
                        {\prfbyaxiom{\rl[\{\nm{z}{\dl{\alpha}},\nm{s}{\alpha}\}]{ax}}
                            {\hvp\th \nm{x}{\beta}, \nm{y}{\dl{\beta}}, \nm{z}{\dl{\tikzmarknode{z-r-1}{\alpha}}}, \nm{s}{\tikzmarknode{s-r-1}{\alpha}}, \nm{t}{\dl{\alpha}} \land \nm{u}{\alpha}, \nm{v}{\dl{\alpha}}, \nm{w}{\alpha}}}
                        {\hvp\th \nm{x}{\beta}, \nm{y}{\dl{\beta}}, \nm{z}{\dl{\tikzmarknode{z-r-2}{\alpha}}} \lor \nm{s}{\tikzmarknode{s-r-2}{\alpha}}, \nm{t}{\dl{\alpha}} \land \nm{u}{\alpha}, \nm{v}{\dl{\alpha}}, \nm{w}{\alpha}}}
                    {\hvp\th \nm{x}{\beta}, \nm{y}{\dl{\beta}}, \nm{t}{\dl{\tikzmarknode{t-l-4}{\alpha}}} \land \nm{u}{\tikzmarknode{u-l-4}{\alpha}}, \nm{v}{\dl{\alpha}}, \nm{w}{\alpha}}}
                {\prftree[r]{\rl{\lor}}
                    {\prftree[r]{\rl{\sqcup}}
                        {\prfbyaxiom{\rl[\{\nm{t}{\alpha},\nm{v}{\dl{\alpha}}\}]{ax}}
                            {\hvp\th \nm{x}{\beta}, \nm{y}{\dl{\beta}}, \nm{t}{\tikzmarknode{t-r-1}{\alpha}}, \nm{u}{\dl{\alpha}}, \nm{v}{\dl{\tikzmarknode{v-r-1}{\alpha}}}, \nm{w}{\alpha}}}
                        {\prfbyaxiom{\rl[\{\nm{u}{\dl{\alpha}},\nm{w}{\alpha}\}]{ax}}
                            {\hvp\th \nm{x}{\beta}, \nm{y}{\dl{\beta}}, \nm{t}{\alpha}, \nm{u}{\dl{\tikzmarknode{u-r-1}{\alpha}}}, \nm{v}{\dl{\alpha}}, \nm{w}{\tikzmarknode{w-r-1}{\alpha}}}}
                        {\hvp\th \nm{x}{\beta}, \nm{y}{\dl{\beta}}, \nm{t}{\tikzmarknode{t-r-2}{\alpha}}, \nm{u}{\dl{\tikzmarknode{u-r-2}{\alpha}}}, \nm{v}{\dl{\tikzmarknode{v-r-2}{\alpha}}}, \nm{w}{\tikzmarknode{w-r-2}{\alpha}}}}
                    {\hvp\th \nm{x}{\beta}, \nm{y}{\dl{\beta}}, \nm{t}{\tikzmarknode{t-r-3}{\alpha}} \lor \nm{u}{\dl{\tikzmarknode{u-r-3}{\alpha}}}, \nm{v}{\dl{\tikzmarknode{v-r-3}{\alpha}}}, \nm{w}{\tikzmarknode{w-r-3}{\alpha}}}}
                {\hvp\th \nm{x}{\beta}, \nm{y}{\dl{\beta}}, \nm{v}{\dl{\tikzmarknode{v-4}{\alpha}}}, \nm{w}{\tikzmarknode{w-4}{\alpha}}}
        }%
        \begin{tikzpicture}[remember picture, overlay, thick, blue]
            % Axioms
            \draw[red] ([yshift=2pt]z-l-1.north) .. controls ++(0,1.3em) and ++(0,1.3em) .. ([yshift=2pt]t-l-1.north);
            \draw[black] ([yshift=2pt]s-l-1.north) .. controls ++(0,1.3em) and ++(0,1.3em) .. ([yshift=2pt]u-l-1.north);
            \draw ([yshift=2pt]z-r-1.north) .. controls ++(0,1.3em) and ++(0,1.3em) .. ([yshift=2pt]s-r-1.north);
            \draw ([yshift=2pt]t-r-1.north) .. controls ++(0,1.3em) and ++(0,1.3em) .. ([yshift=2pt]v-r-1.north);
            \draw ([yshift=2pt]u-r-1.north) .. controls ++(0,1.3em) and ++(0,1.3em) .. ([yshift=2pt]w-r-1.north);
            % History of z (left)
            \draw ([yshift=1pt]z-l-3.north) .. controls ++(0,15pt) and ++(0,-15pt) .. ([yshift=-1pt]z-l-2.south);
            \draw ([yshift=1pt]z-l-2.north) .. controls ++(0,10pt) and ++(0,-10pt) .. ([yshift=-1pt]z-l-1.south);
            % History of s (left)
            \draw ([yshift=2pt]s-l-3.north) .. controls ++(0,20pt) and ++(0,-10pt) .. ([yshift=-1pt]s-l-2.south);
            \draw ([yshift=2pt]s-l-2.north) .. controls ++(0,15pt) and ++(0,-5pt) .. ([yshift=-1pt]s-l-1.south);
            % History of t (left)
            \begin{scope}[red]
                \draw ([yshift=2pt]t-l-4.north) .. controls ++(0,15pt) and ++(0,-15pt) .. ([yshift=-1pt]t-l-3.south);
                \draw ([yshift=2pt]t-l-3.north) .. controls ++(0,20pt) and ++(0,-10pt) .. ([yshift=-1pt]t-l-2.south);
                \draw ([yshift=2pt]t-l-2.north) .. controls ++(0,15pt) and ++(0,-5pt) .. ([yshift=-1pt]t-l-1.south);
            \end{scope}
            % History of u (left)
            \begin{scope}[black]
                \draw ([yshift=1pt]u-l-4.north) .. controls ++(0,25pt) and ++(0,-10pt) .. ([yshift=-1pt]u-l-3.south);
                \draw ([yshift=1pt]u-l-3.north) .. controls ++(0,15pt) and ++(0,-15pt) .. ([yshift=-1pt]u-l-2.south);
                \draw ([yshift=1pt]u-l-2.north) .. controls ++(0,10pt) and ++(0,-10pt) .. ([yshift=-1pt]u-l-1.south);
            \end{scope}
            % History of z (right)
            \draw ([yshift=2pt]z-r-2.north) .. controls ++(0,7pt) and ++(0,-7pt) .. ([yshift=-1pt]z-r-1.south);
            % History of s (right)
            \draw ([yshift=1pt]s-r-2.north) .. controls ++(0,7pt) and ++(0,-7pt) .. ([yshift=-1pt]s-r-1.south);
            % History of t (right)
            \draw ([yshift=1pt]t-r-3.north) .. controls ++(0,7pt) and ++(0,-7pt) .. ([yshift=-1pt]t-r-2.south);
            \draw ([yshift=1pt]t-r-2.north) .. controls ++(0,10pt) and ++(0,-10pt) .. ([yshift=-1pt]t-r-1.south);
            % History of u (right)
            \draw ([yshift=2pt]u-r-3.north) .. controls ++(0,7pt) and ++(0,-7pt) .. ([yshift=-1pt]u-r-2.south);
            \draw ([yshift=2pt]u-r-2.north) .. controls ++(0,15pt) and ++(0,-5pt) .. ([yshift=-1pt]u-r-1.south);
            % History of v (right)
            \draw ([yshift=2pt]v-4.north) .. controls ++(0,35pt) and ++(0,-25pt) .. ([yshift=-1pt]v-r-3.south);
            \draw ([yshift=2pt]v-r-3.north) .. controls ++(0,7pt) and ++(0,-7pt) .. ([yshift=-1pt]v-r-2.south);
            \draw ([yshift=2pt]v-r-2.north) .. controls ++(0,15pt) and ++(0,-5pt) .. ([yshift=-1pt]v-r-1.south);
            % History of w (right)
            \draw ([yshift=1pt]w-4.north) .. controls ++(0,20pt) and ++(0,-30pt) .. ([yshift=-1pt]w-r-3.south);
            \draw ([yshift=1pt]w-r-3.north) .. controls ++(0,7pt) and ++(0,-7pt) .. ([yshift=-1pt]w-r-2.south);
            \draw ([yshift=1pt]w-r-2.north) .. controls ++(0,10pt) and ++(0,-10pt) .. ([yshift=-1pt]w-r-1.south);
            % Cuts
            \path[draw, densely dashed, semithick, rounded corners=1em]
                 ([yshift=-1pt]z-l-3.south) -- ([yshift=-5em]z-l-3.south) -- ([yshift=-5em]z-r-2.south) -- ([yshift=-1pt]z-r-2.south);
            \path[draw, densely dashed, semithick, rounded corners=1em]
                 ([yshift=-1pt]s-l-3.south) -- ([yshift=-4em]s-l-3.south) -- ([yshift=-4em]s-r-2.south) -- ([yshift=-1pt]s-r-2.south);
            \path[draw, densely dashed, semithick, rounded corners=1em]
                 ([yshift=-1pt]t-l-4.south) -- ([yshift=-3.5em]t-l-4.south) -- ([yshift=-3.5em]t-r-3.south) -- ([yshift=-1pt]t-r-3.south);
            \path[draw, densely dashed, semithick, rounded corners=1em]
                 ([yshift=-1pt]u-l-4.south) -- ([yshift=-4.5em]u-l-4.south) -- ([yshift=-4.5em]u-r-3.south) -- ([yshift=-1pt]u-r-3.south);
        \end{tikzpicture}%
        \\[1.5em]
        \caption{A logical cut-reduction step has been applied (see \cref{sec:cut-reduction-failure}).
            The conjunction \(\nm{t}{\dl{\alpha}} \land \nm{u}{\alpha}\) is now outside the
            interface of the upper cut: the alternating path that connects \(\nm{t}{\dl{\alpha}}\)
            with~\(\nm{u}{\alpha}\) must be omitted when computing the interpretation, as it
            uses edges from incompatible branches (highlighted in red and black). As a result,
            the edge between \(\nm{v}{\dl{\alpha}}\) and~\(\nm{w}{\alpha}\) is lost.}
        \label{fig:cut-reduction-failure:post}
    \end{subfigure}
    \par\vspace{3em}%
    \begin{subfigure}{\textwidth}
        \centering
        \begin{tikzpicture}
            \graph[math nodes]{
                x -!- y -!- v -!- w;
                x --[bend left, "\(\scriptstyle \{x,y,v,w\}\)"] y;
                v --[bend left, "\(\scriptstyle \{x,y,v,w\}\)"] w;
            };
        \end{tikzpicture}%
        \qquad\qquad\qquad\qquad%
        \begin{tikzpicture}
            \graph[math nodes]{
                x -!- y -!- v -!- w;
                x --[bend left, "\(\scriptstyle \{x,y,v,w\}\)"] y;
            };
        \end{tikzpicture}\\[1em]
        \caption{The branch-labeled axiom graph of the original derivation~(left) and that
            of the reduced one~(right).}
        \label{fig:cut-reduction-failure:graphs}
    \end{subfigure}
    \par\vspace{4em}%
    \caption{A derivation whose branch-labeled axiom graph decreases when a logical cut-reduction
        step is applied.}
    \label{fig:cut-reduction-failure}
\end{figure}
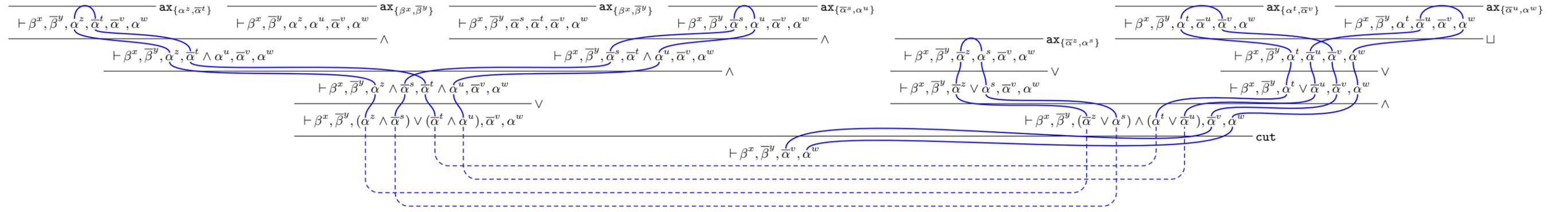
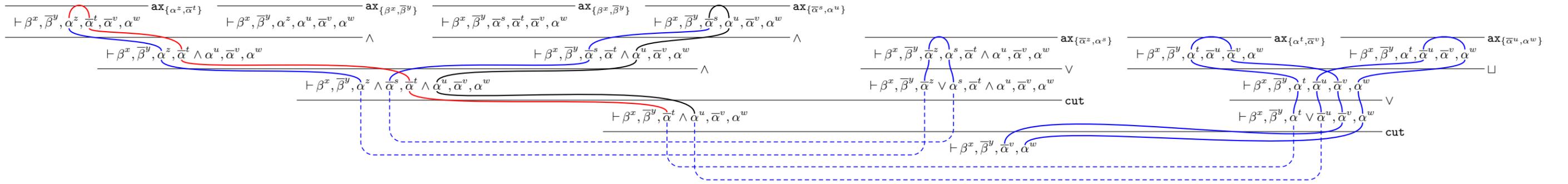
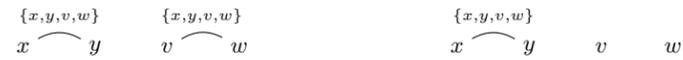
\vfill
\clearpage
}

\end{document}